\documentclass[11pt, reqno]{amsart}

\usepackage{html}
\usepackage{url}
\usepackage[dcucite]{harvard}

\usepackage{graphicx,amsmath,amssymb,epsfig,harvard}
\usepackage{colortbl}

\long\def\comment#1{}
\long\def\comment#1{}
\newtheorem{algorithm}{Algorithm}
\newtheorem{theorem}{Theorem}

\newtheorem{corollary}{Corollary}
\newtheorem{lemma}{Lemma}

\theoremstyle{definition}

\newtheorem{remark}{Comment}[section]

\newcommand{\citen}{\citeasnoun}

\newcommand{\be}{\begin{eqnarray}}
\newcommand{\ee}{\end{eqnarray}}

\newcommand{\ba}{\begin{array}}
\newcommand{\ea}{\end{array}}
\newcommand{\bs}{\begin{align}\begin{split}\nonumber}
\newcommand{\bsnumber}{\begin{align}\begin{split}}
\newcommand{\es}{\end{split}\end{align}}

\renewcommand{\(}{\left(}
\renewcommand{\)}{\right)}
\renewcommand{\[}{\left[}
\renewcommand{\]}{\right]}
\renewcommand{\hat}{\widehat}

\newcommand{\Gn}{\mathbb{G}_n}

\newcommand{\Ep}{{\mathrm{E}}}
\newcommand{\En}{{\mathbb{E}_n}}

\newcommand{\conflvl}{\gamma}
\newcommand{\barEp}{\bar \Ep}

\newcommand{\PX}{\mathcal{P}_{\hat I}}

\newcommand{\MX}{\mathcal{M}_{\hat I}}

\newcommand{\MXd}{\mathcal{M}_{\hat I_1}}
\newcommand{\MXy}{\mathcal{M}_{\hat I_2}}

\renewcommand{\Pr}{{\mathrm{P}}}

\def\RR{ {\Bbb{R}}}

\def\supp{{\rm support}}

\newcommand{\semin}[1]{\phi_{{\rm min}}(#1)}
\newcommand{\semax}[1]{\phi_{{\rm max}}(#1)}
\renewcommand{\hat}{\widehat}
\renewcommand{\leq}{\leqslant}
\renewcommand{\geq}{\geqslant}
\def\aa{{a}}

\begin{document}

\title[Inference  after Model Selection]{ Inference on Treatment Effects After Selection Amongst High-Dimensional Controls}
\author[Belloni \ Chernozhukov \ Hansen]{A. Belloni \and V. Chernozhukov \and C. Hansen}

\date{First version:  May 2010.  This version is of  \today. This is a revision of a  2011 ArXiv/CEMMAP  paper entitled ``Estimation of Treatment Effects with High-Dimensional Controls"}

\thanks{We thank  Ted Anderson, Takeshi Amemiya, St\'ephane Bonhomme, Mathias Cattaneo, Gary Chamberlain, Denis Chetverikov, Graham Elliott, Eric Gretchen, Bruce Hansen, James Hamilton, Han Hong, Guido Imbens,  Tom MaCurdy, Anna Mikusheva,  Whitney Newey, Alexei Onatsky, Joseph Romano, Andres Santos, Chris Sims, and  participants of 10th Econometric World Congress in Shanghai 2010,  CIREQ-Montreal, Harvard-MIT, UC San-Diego, Princeton, Stanford, and  Infometrics Workshop for helpful comments.}

\maketitle

\begin{abstract}
We propose robust methods for inference on the effect of a treatment variable on a scalar outcome in the presence of very many controls.  Our setting is a partially linear model with possibly non-Gaussian and heteroscedastic disturbances where  the number of controls may be much larger than the sample size.  To make informative inference feasible, we require the model to be approximately sparse; that is, we require that the effect of confounding factors can be controlled for up to a small approximation error by conditioning on a relatively small number of controls whose identities are unknown. The latter condition makes it possible to estimate the treatment effect by selecting approximately the right set of controls. We develop a novel estimation and uniformly valid inference method for the treatment effect in this setting, called the ``post-double-selection" method. Our results apply to Lasso-type methods used for covariate selection as well as to any other model selection method that is able to find a sparse model with good approximation properties.

The main attractive feature of our method is that it allows for imperfect selection of the controls and provides confidence intervals that are valid uniformly across a large class of models.   In contrast, standard post-model selection estimators fail to provide uniform inference even in simple cases with a small, fixed number of controls. Thus our method resolves the problem of uniform inference after model selection for a large, interesting class of models. We illustrate the use of the developed methods with numerical simulations and an application to the effect of abortion on crime rates. \\


\end{abstract}

\section{Introduction}

Many empirical analyses in economics focus on estimating the structural, causal, or treatment effect of some variable on an outcome of interest.  For example, we might be interested in estimating the causal effect of some government policy on an economic outcome such as employment.  Since economic policies and many other economic variables are not randomly assigned, economists rely on a variety of quasi-experimental approaches based on observational data when trying to estimate such effects.  One important method is based on the assumption that the variable of interest can be taken as randomly assigned once a sufficient set of other factors has been controlled for.  Economists, for example, might argue that changes in state-level public policies can be taken as randomly assigned relative to unobservable factors that could affect changes in state-level outcomes after controlling for aggregate macroeconomic activity, state-level economic activity, and state-level demographics; see, for example, \citen{heckman:metricslabormarkets} or \citen{imbens:review}.

A problem empirical researchers face when relying on an identification strategy for estimating a structural effect that relies on a conditional on observables argument is knowing which controls to include.  Typically, economic intuition will suggest a set of variables that might be important but will not identify exactly which variables are important or the functional form with which variables should enter the model.  This lack of clear guidance about what variables to use leaves researchers with the problem of selecting a set of controls from a potentially vast set of control variables including raw regressors available in the data as well as interactions and other transformations of these regressors.  A typical economic study will rely on an \emph{ad hoc} sensitivity analysis in which a researcher reports results for several different sets of controls in an attempt to show that the parameter of interest that summarizes the causal effect of the policy variable is insensitive to changes in the set of control variables.  See \citen{levitt:abortion}, which we use as the basis for the empirical study in this paper, or examples in \citen{AngristBook} among many other references.

We present an approach to estimating and performing inference on structural effects in an environment where the treatment variable may be taken as exogenous conditional on observables that complements existing strategies.  We pose the problem in the framework of a partially linear model
\begin{equation}\label{plm:intro}
y_{i} = d_i \alpha_0 + g(z_i) + \zeta_i
\end{equation}
where $d_i$ is the treatment/policy variable of interest, $z_i$ is a set of control variables, and $\zeta_i$ is an unobservable that satisfies $\textnormal{E}[\zeta_i\mid d_i,z_i] = 0$.\footnote{ We note that $d_i$ does not need to be binary.}  The goal of the econometric analysis is to conduct inference on the treatment effect $\alpha_0$.  We examine the problem of selecting a set of variables from among $p$ potential controls $x_i=P(z_i)$, which may consist of $z_i$ and transformations of $z_i$, to adequately approximate $g(z_i)$ allowing for $p > n$. Of course, useful inference about $\alpha_0$ is unavailable in this framework without imposing further structure on the data.  We impose such structure by assuming that exogeneity of $d_i$ may be taken as given once one controls linearly for a relatively small number $s < n$ of variables in $x_i$  whose identities are \textit{a priori} unknown.  This assumption implies that a linear combination of these $s$ unknown controls provides an approximation to $g(z_i)$ which produces relatively small approximation errors.\footnote{We carefully define what we mean by small approximation errors in Section 2.}  This assumption, which is termed approximate sparsity or simply sparsity, allows us to approach the problem of estimating $\alpha_0$ as a variable selection problem.
This framework allows for the realistic scenario in which the researcher is unsure about exactly which variables or transformations are important for approximating $g(z_i)$ and so must search among a broad set of controls.

The assumed sparsity includes as special cases the most common approaches to parametric and nonparametric regression analysis. Sparsity justifies the use of fewer variables than there are observations in the sample.  When the initial number of
variables is high, the assumption justifies the use of variable selection methods to reduce the number of variables to a manageable size. In many economic applications, formal and informal strategies are often used to select such smaller sets of potential control variables.   Most of these standard variable selection strategies are non-robust and may produce poor inference.\footnote{An example of inference going wrong is
given in Figure 1 (left panel), presented in the next section, where a standard post-model selection estimator has a bimodal distribution which sharply deviates from the standard normal distribution. More examples are given in Section 6 where we document the poor inferential performance of a standard post-model selection method. }
In an effort to demonstrate robustness of their conclusions, researchers often employ \emph{ad hoc} sensitivity analyses which examine the robustness of inferential conclusions to variations in the set of controls.  Such sensitivity analyses are useful but lack rigorous justification. As a complement to these \emph{ad hoc} approaches, we propose a formal, rigorous approach to inference allowing for selection of controls.  Our proposal uses modern variable selection methods in a novel manner which results in  robust and valid inference.

 The main contributions of this paper are providing a  robust estimation and inference method within a partially linear model
  with potentially very high-dimensional  controls and developing the supporting theory.  The method relies on the use of Lasso-type or other sparsity-inducing  procedures for variable selection.
Our approach differs from usual post-model-selection methods that rely on a single selection step.  Rather, we use two different variable selection steps followed by a final estimation step as follows:
  \begin{itemize}
   \item[1.] In the first step, we select a set of control variables that are useful for predicting the treatment $d_i$. This step helps to insure robustness by finding control variables that are strongly related to the treatment
and thus potentially important confounding factors.

\item[2.] In the second step, we select additional variables by selecting control variables
that predict $y_{i}$.  This step helps to insure that we have captured important elements in the equation of interest,
ideally helping keep the residual variance small as well as intuitively providing an additional chance to find important
confounds.

 \item[3.] In the final step, we estimate the treatment effect $\alpha_0$ of interest by the linear regression of $y_{i}$ on the treatment
$d_i$ and the union of the set of variables selected in the two variable selection steps.

\end{itemize}

We provide theoretical results on the properties of the resulting treatment effect estimator and show that it provides inference
that is uniformly valid over large classes of models and also achieves
 the semi-parametric efficiency
 bound under some conditions.  Importantly, our theoretical results allow for imperfect variable selection in either of
  the two variable selection steps as well as allowing for non-Gaussianity and heteroscedasticity of the model's errors.\footnote{In a companion paper that presents an overview of results for $\ell_1$-penalized estimators, \citen{BCH2011:InferenceGauss}, we provide similar results in the idealized Gaussian homoscedastic framework. }

We illustrate the theoretical results through an examination of the effect of abortion on crime rates following \citen{levitt:abortion}.  In this example, we find that the formal variable selection procedure produces a qualitatively different result than that obtained through the \textit{ad hoc} set of sensitivity results presented in the original paper.  By using formal variable selection, we select a small set of between eight and fourteen variables depending on the outcome, compared to the set of eight variables considered by \citen{levitt:abortion}.  Once this set of variables is linearly controlled for, the estimated abortion effect is rendered imprecise.  It is interesting that the key variable selected by the variable selection procedure is the initial condition for the abortion rate.  The selection of this initial condition and the resulting imprecision of the estimated treatment effect suggest that one cannot determine precisely whether the effect attributed to abortion found when this initial condition is omitted from the model is due to changes in the abortion rate or some other persistent state-level factor that is related to relevant changes in the abortion rate and current changes in the crime rate.\footnote{Note that all models are estimated in first-differences to eliminate any state-specific factors that might be related to both the relevant level of the abortion rate and the level of the crime rate.}  It is interesting that \citen{FooteGoetzAbortion} raise a similar concern based on intuitive grounds and additional data in a comment on \citen{levitt:abortion}.  \citen{FooteGoetzAbortion} find that a linear trend interacted with crime rates before abortion could have had an effect renders the estimated abortion effects imprecise.\footnote{\citen{DLAbortionResponse} provide yet more data and a more complicated specification in response to \citen{FooteGoetzAbortion}.  In a supplement available at http://faculty.chicagobooth.edu/christian.hansen/research/, we provide additional results based on \citen{DLAbortionResponse}.  The conclusions are similar to those obtained in this paper in that we find the estimated abortion effect becomes imprecise once one allows for a broad set of controls and selects among them.  However, the specification of \citen{DLAbortionResponse} relies on a large number of district cross time fixed effects and so does not immediately fit into our regularity conditions.  We conjecture the methodology continues to work in this case but leave verification to future research.}  Overall, finding that a formal, rigorous approach to variable selection produces a qualitatively different result than a more \textit{ad hoc} approach suggests that these methods might be used to complement economic intuition in selecting control variables for estimating treatment effects in settings where treatment is taken as exogenous conditional on observables.

\textbf{Relationship to literature.}  We  contribute to several existing literatures.  First, we contribute to the literature on series estimation of partially linear models (\citeasnoun{donald:newey:pl}, \citeasnoun{hardle:pl}, \citeasnoun{robinson}, and others).   We differ from most of the existing literature which considers $p\ll n$ series terms by allowing $p \gg n$ series terms from which we select $\hat s \ll n$ terms to construct the regression fits. Considering an initial broad set of terms allows for more refined approximations of regression functions relative to the usual approach that uses only a few low-order terms.  See, for example, \citeasnoun{BCH2011:InferenceGauss} for a wage function example and Section 5 for theoretical examples.  However, our most important contribution is to  allow for data-dependent selection of the appropriate series terms. The previous literature on inference in the partially linear model generally takes the series terms as given without allowing for their data-driven selection. However, selection of series terms is crucial for achieving consistency when $p \gg n$ and is needed for increasing efficiency even when $p  =C n$ with $C<1$.  That the standard estimator can be be highly inefficient in the latter case follows from results in \citeasnoun{CJN:PLMStandardError}.\footnote{ \citeasnoun{CJN:PLMStandardError} derive properties of series estimator under $p = Cn$, $C<1$,  asymptotics. It follows  from their
results that under homoscedasticity the series estimator achieves the semiparametric efficiency bound only if $C\to 0$. }    We focus on Lasso for performing this selection as a theoretically and computationally attractive device but note that any other method, such as selection using the traditional generalized cross-validation criteria, will work as long as the method guarantees sufficient sparsity in its solution.    After model selection, one may apply conventional standard errors or the refined standard errors proposed by \citen{CJN:PLMStandardError}.\footnote{If the selected number of terms $\hat s$ is a substantial fraction of $n$, we recommend using \citeasnoun{CJN:PLMStandardError} standard errors after applying our model selection procedure. }


Second, we contribute to the literature on the estimation of treatment effects. We note that the policy variable $d_i$ does not have to be binary in our framework.  However, our method has a useful interpretation related to the propensity score when $d_i$ is binary.   In the first selection step, we select terms from $x_i$ that predict the treatment $d_i$,  i.e. terms that explain the propensity score.  We also select terms from $x_i$ that predict $y_i$, i.e. terms that explain the outcome regression function. Then we run a final regression of $y_i$ on the treatment $d_i$ and the union of selected terms.   Thus, our procedure relies on the selection of variables relevant for both the propensity score and the outcome  regression.  Relying on selecting variables that are important for both objects allows us to achieve two goals: we obtain uniformly valid confidence sets for $\alpha_0$ despite imperfect model selection and we achieve full efficiency for estimating $\alpha_0$ in the homoscedastic case.  The relation of our approach to the propensity score brings about interesting connections to the treatment effects literature.  \citen{hahn:prop}, \citen{heckman:ichimura:smith:mathching}, and \citen{abadie:imbens} have constructed efficient regression or matching-based estimates of average treatment effects.   \citen{hahn:prop} also shows that conditioning on the propensity score is unnecessary for efficient estimation of average treatment effects.  \citen{hirano:imbens:ridder} demonstrate that one can efficiently estimate average treatment effects using estimated propensity score weighting alone. \citen{robins:dr} have shown that using propensity score modeling coupled with a parametric regression model leads to efficient estimates if either the propensity score model or the parametric regression model is correct.  While our contribution  is quite distinct from these approaches, it also highlights the important  robustness role played by the propensity score model in the selection of the right control terms for the final regression.

Third, we contribute to the literature on estimation and inference with high-dimensional data and to the uniformity literature.  There has been extensive work on estimation and perfect model selection in both low and high-dimensional contexts,\footnote{For  reviews focused on econometric applications, see, e.g., \citen{Hansen2005} and \citen{BellChernHans:Gauss}.} but there has been little work on inference after imperfect model selection.  Perfect model selection relies on unrealistic assumptions, and  model selection mistakes can have serious consequences for inference as has been shown in \citen{potscher}, \citen{leeb:potscher:pms}, and others.    In work on instrument selection for estimation of a linear instrumental variables model, \citen{BellChenChernHans:nonGauss} have shown that model selection mistakes do not prevent valid inference about low-dimensional structural parameters due to the inherent adaptivity of the problem: Omission of a relevant instrument does not affect consistency of an IV estimator as long as there is another relevant instrument. The partially linear regression model (\ref{plm:intro}) does not have the same adaptivity structure, and model selection based on the outcome regression alone produces non-robust confidence intervals.\footnote{The poor performance of inference on a treatment effect after model selection on only the outcome equation is shown through simulations in Section 6.}
Our post-double selection procedure creates the necessary adaptivity by performing two separate model selection steps, making it possible to perform robust/uniform inference after model selection.  The uniformity holds over large, interesting classes of high-dimensional sparse models.  In that regard, our contribution is in the spirit and builds upon  the classical contribution by \citen{romano:uniform} on the uniform validity of t-tests for the univariate mean.  It also shares the spirit of recent contributions, among others, by \citen{mikusheva} on uniform inference in autoregressive models, by \citen{andrews:cheng}  on uniform inference in moment condition models that are potentially unidentified, and  by \citen{andrews:cheng:guggen}  on a generic framework for uniformity analysis.

Finally, we contribute to the broader literature on high-dimensional estimation.  For variable selection we use  $\ell_1$-penalization methods, though our method and theory will allow for the use of other methods. $\ell_1$-penalized methods have been proposed for model selection problems in high-dimensional least squares problems, e.g. Lasso in \citen{FF:1993} and \citen{T1996}, in part because they are computationally efficient. Many $\ell_1$-penalized methods have been shown to have good estimation properties even when perfect variable selection is not feasible; see, e.g., \citen{CandesTao2007}, \citen{MY2007}, \citen{BickelRitovTsybakov2009},  \citen{horowitz:lasso}, \citen{BC-PostLASSO} and the references therein. Such methods have also been shown to extend suitably to nonparametric and non-Gaussian cases as in \citen{BickelRitovTsybakov2009} and \citen{BellChenChernHans:nonGauss}. These methods also produce models with a relatively small set of variables. The last property is important in that it leaves the researcher with a set of variables that may be examined further; in addition it corresponds to the usual approach in economics that relies on considering a small number of controls.

\textbf{Paper Organization.}   In Section 2, we formally present the modeling environment including the key sparsity condition and develop our advocated estimation and inference method.  We  establish the consistency and asymptotic normality of our estimator of $\alpha_0$ uniformly over large classes of models in Section 3.  In Section 4, we present a generalization of the basic procedure to allow for model selection methods other than Lasso.  In Section 5, we present a
series of theoretical examples in which we provide primitive condition that imply the higher-level conditions of Section 3.   In Section 6, we present a series of numerical examples that verify our theoretical results numerically, and we apply our method to the abortion and crime example of \citen{levitt:abortion} in Section 7.  In appendices, we provide the proofs.

\textbf{Notation.}   In what follows, we work with triangular array data $\{\(\omega_{i,n}, i=1,...,n\), n=1,2,3,...\}$
defined on probability space $(\Omega, \mathcal{A}, \Pr_n)$,
where $\Pr = \Pr_n$ can change with $n$.    Each  $\omega_{i,n}= (y_{i,n}', z_{i,n}', d_{i,n}')'$
is a vector with components defined below, and these vectors are i.n.i.d. -- independent across $i$, but not necessarily identically distributed. Thus, all parameters that characterize the distribution of  $\{\omega_{i,n}, i=1,...,n\}$ are
implicitly indexed by $\Pr_n$ and thus by $n$.  We omit the dependence on these objects from the notation in what follows for notational simplicity.  We use array asymptotics to better capture some finite-sample phenomena and to insure the robustness of conclusions with respect to perturbations of the data-generating process $\Pr$
along various sequences. This robustness, in turn, translates into uniform validity of confidence regions over certain regions of data-generating processes.

We use the following empirical process notation, $\En[f] := \En[f(\omega_i)] := \sum_{i=1}^n f(\omega_i)/n,$  and $\Gn(f) := \sum_{i=1}^n ( f(\omega_i)
- \Ep[f(\omega_i)] )/\sqrt{n}.$
Since we want to deal with i.n.i.d. data, we also introduce the average expectation operator:
$
\barEp[f] := \Ep \En[f] =  \Ep \En[f(\omega_i)] = \sum_{i=1}^n \Ep[f(\omega_i)]/n.
$
The ${l}_2$-norm is denoted by
$\|\cdot\|$, and the ${l}_0$-norm, $\|\cdot\|_0$, denotes the number of non-zero components of a vector.  We use $\| \cdot \|_{\infty}$ to denote the maximal element of a vector.  
Given a vector $\delta \in \RR^p$, and a set of
indices $T \subset \{1,\ldots,p\}$, we denote by $\delta_T \in \RR^p$ the vector in which $\delta_{Tj} = \delta_j$ if $j\in T$, $\delta_{Tj}=0$ if $j \notin T$. We use the notation $(a)_+ = \max\{a,0\}$, $a \vee b = \max\{ a, b\}$, and $a \wedge b = \min\{ a , b \}$. We also use the notation $a \lesssim b$ to denote $a \leqslant c b$ for some constant $c>0$ that does not depend on $n$; and $a\lesssim_P b$ to denote $a=O_P(b)$. For an event $E$, we say that $E$ wp $\to$ 1 when $E$ occurs with probability approaching one as $n$ grows.  Given a $p$-vector $b$, we denote
$\text{support}(b) = \{ j \in \{1,...,p\}: b_j \neq 0\} $.

\section{Inference on Treatment and Structural Effects Conditional on Observables}\label{Sec:Treatment}

\subsection{Framework}

We consider the partially linear model
\begin{eqnarray}\label{eq: PL1}
 & y_{i}  = d_i\alpha_0 + g(z_i) + \zeta_i,  &  \Ep[\zeta_i \mid z_i, d_i]= 0,\\
  & d_i  = m(z_i) + v_i, \label{eq: PL2}  &   \Ep[v_i \mid z_i] = 0,
\end{eqnarray}
where $y_{i}$ is the outcome variable, $d_i$ is the policy/treatment variable whose impact $\alpha_0$ we would like to infer, $z_i$ represents confounding factors on which we need to condition, and $\zeta_i$ and $v_i$ are disturbances. 
The parameter $\alpha_0$ is the average treatment or structural effect under appropriate conditions given, for example, in \citen{heckman:metricslabormarkets} or \citen{imbens:review} and is of major interest in many empirical studies.

The confounding factors $z_i$ affect the policy variable via the function $m(z_i)$ and the outcome variable via the function $g(z_i)$. Both of these functions are unknown and potentially complicated. We use linear combinations of control terms $x_i = P(z_i)$ to approximate $g(z_i)$ and $m(z_i)$, writing (\ref{eq: PL1}) and (\ref{eq: PL2}) as
\begin{eqnarray}\label{eq: appPL1}
& & y_{i}  = d_i\alpha_0 + \underbrace{x_i'\beta_{g0} + r_{gi}}_{g(z_i)} + \zeta_i, \\
& & d_i  = \underbrace{x_i'\beta_{m0} + r_{mi}}_{m(z_i)} + v_i \label{eq: appPL2},
\end{eqnarray}
where $ x_i'\beta_{g0}$ and $x_i'\beta_{m0}$ are approximations to $g(z_i)$ and $m(z_i)$, and $r_{gi}$ and $r_{mi}$ are the corresponding approximation errors. In order to allow for a flexible specification and incorporation of pertinent confounding factors, the vector of controls, $x_i = P(z_i)$, can have a dimension $p=p_n$ which can be large relative to the sample size. Specifically, our results only require $\log p = o(n^{1/3})$ along with other technical conditions. High-dimensional regressors $x_i = P(z_i)$ could arise for different reasons.  For instance, the list of available controls could be large, i.e. $x_i=z_i$ as in e.g. \citen{koenker:jappliedeconometircs}.  It could also be that many technical controls are present; i.e. the list $x_i=P(z_i)$ could be composed of a large number of transformations of elementary regressors $z_i$ such as B-splines, dummies, polynomials, and various interactions as in \citen{newey:series} or \citen{chen:Chapter}.

Having very many controls creates a challenge for estimation and inference.  A key condition that makes it possible to perform constructive estimation and inference in such cases is termed sparsity.  Sparsity is the condition that there exist approximations $ x_i'\beta_{g0}$ and $x_i'\beta_{m0}$ to $g(z_i)$ and $m(z_i)$ in  (\ref{eq: appPL1})-(\ref{eq: appPL2}) that require only a small number of non-zero coefficients to render the approximation errors $r_{gi}$ and $r_{mi}$ sufficiently small relative to estimation error.  More formally, sparsity relies on two conditions.  First, there exist $\beta_{g0}$ and $\beta_{m0}$ such  that at most $s=s_n \ll n$ elements of $\beta_{m0}$ and $\beta_{g0}$ are non-zero so that
$$
\|\beta_{m0}\|_0 \leq s \text{ and } \|\beta_{g0}\|_0 \leq s.
$$
Second, the sparsity condition requires the size of the resulting approximation errors to be small compared to the conjectured size of the estimation error:
$$
\{\barEp[r^{2}_{gi}]\}^{1/2} \lesssim \sqrt{s/n}  \text{ and } \{\barEp[r^{2}_{mi}]\}^{1/2} \lesssim  \sqrt{s/n}.
$$
Note that the size of the approximating model $s=s_n$ can grow with $n$ just as in standard series estimation.

The high-dimensional-sparse-model framework outlined above extends the standard framework in the treatment effect literature which assumes both that the identities of the relevant controls are known and that the number of such controls $s$ is much smaller than the sample size.  Instead, we  assume that there are many, $p$, potential controls of which at most $s$ controls suffice to achieve a desirable approximation to the unknown functions $g(\cdot)$ and $m(\cdot)$ and allow the identity of these controls to be unknown.  Relying on this assumed sparsity, we use selection methods to select approximately the right set of controls and then estimate the treatment effect $\alpha_0$.

\subsection{The Method: Least Squares after Double Selection}\label{Sec:DoubleSelection}

We propose the following method for estimating and performing inference about $\alpha_0$.  The most important feature of this method is that it does not rely on the highly unrealistic assumption of perfect model selection which is often invoked to justify inference after model selection. To the best of our knowledge, our result is the first of its kind in this setting. This result extends our previous results on inference under imperfect model selection in the instrumental variables model given in \citen{BellChenChernHans:nonGauss}. The problem is fundamentally more difficult in the present paper due to lack of adaptivity in estimation which we overcome by introducing additional model selection steps. The construction of our advocated procedure reflects our effort to offer a method that has attractive robustness/uniformity properties for inference.  The estimator is $\sqrt{n}$-consistent and asymptotically normal under mild conditions and provides confidence intervals that are robust to various perturbations of the data-generating process that preserve approximate sparsity.

To define the method, we first write the reduced
form corresponding to (\ref{eq: PL1})-(\ref{eq: PL2}) as:
\begin{eqnarray}\label{eq: RPL1}
& & y_{i}  = x_i'\bar \beta_0 +  \bar r_{i} + \bar \zeta_i, \\
& & d_i  = x_i'\beta_{m0} + r_{mi} +  v_i \label{eq: RPL2},
\end{eqnarray}
where $\bar \beta_0 := \alpha_0  \beta_{m0} + \beta_{g0}, \ \ \bar r_i := \alpha_0 r_{mi} + r_{gi}, \ \ \bar \zeta_i := \alpha_0 v_i + \zeta_i.$

We have two equations and hence can apply model selection methods to each equation to select control terms.  The chief method we discuss is the Lasso method described in more detail below. Given the set of selected controls from (\ref{eq: RPL1}) and (\ref{eq: RPL2}), we can estimate $\alpha_0$ by a least squares regression of $y_{i}$ on $d_i$ and the union of the selected controls.  Inference on $\alpha_0$ may then be performed using conventional methods for inference about parameters estimated by least squares.   Intuitively, this procedure works well since we are more likely to recover key controls by considering selection of controls from both equations instead
of just considering selection of controls from the single equation (\ref{eq: appPL1}) or (\ref{eq: RPL1}). In finite-sample experiments, single-selection methods essentially fail, providing poor inference relative to the double-selection method outlined above.  This performance is also supported theoretically by the fact that the double-selection method requires weaker regularity conditions for its validity and for attaining the efficiency bound\footnote{Semi-parametric efficiency is attained in the homoscedastic case.} than the single selection method.

Now we formally define the post-double-selection estimator:   Let $\hat I_1 = {\rm support }(\hat \beta_1)$
denote the control terms selected by a feasible Lasso estimator $\hat \beta_1$
computed using data $(\tilde y_i,\tilde x_i) = (d_{i}, x_i), \ i =1,...,n$.
Let $\hat I_2 = {\rm support }(\hat \beta_2)$ denote the control terms selected by a feasible Lasso estimator
$\hat \beta_2$ computed using  data $(\tilde y_i,\tilde x_i) = (y_{i}, x_i), \ i =1,...,n$.
The post-double-selection estimator  $\check \alpha$ of $\alpha_0$ is defined as the least
squares  estimator obtained by regressing $y_{i}$ on $d_i$ and the selected control terms
$x_{ij}$ with $j \in \hat I \supseteq \hat I_1 \cup \hat I_2$:
 \begin{equation}\label{eq: define pds}
(\check \alpha, \check \beta) = \underset{ \alpha \in \Bbb{R}, \beta \in \Bbb{R}^p}{\rm argmin}\{ \En[(y_{i} - d_i \alpha - x_i'\beta)^2] \ : \ \beta_j = 0, \forall j \not \in \hat I \}.
 \end{equation}
The set $\hat I$ may contain variables that were not selected in the variable selection steps with indices in $\hat I_3$ that the analyst thinks are important
for ensuring robustness.  We call $\hat I_3$ the amelioration set.  Thus, $\hat I =  \hat I_1 \cup \hat I_2 \cup \hat I_3$; let $\hat s = |\hat I|$
and $\hat s_j = |\hat I_j|$ for $j =1,2,3$.

We define feasible Lasso estimators below and note that other selection methods could be used
as well.  When a feasible Lasso  is used to construct $\hat I_1$ and $\hat I_2$, we refer to the post-double-selection estimator as the \emph{post-double-Lasso estimator}.  When other model selection devices are used to construct  $\hat I  = \hat I_1 \text{ and } \hat I_2$, we shall refer the estimator as the generic post-double-selection estimator.

The main theoretical result of the paper shows that the post-double-selection estimator
$\check \alpha$ obeys
\begin{equation}\label {normal}
([\barEp v_i^2]^{-1}\barEp[ v_i^2\zeta_i^2] [\barEp v_i^2]^{-1})^{-1/2} \sqrt{n} (\check \alpha - \alpha_0) \rightsquigarrow N(0,1)
\end{equation}
under approximate sparsity conditions,\emph{ uniformly} within a rich set of data generating processes. We also show that the standard plug-in estimator for standard errors is consistent in these settings.
All of these results imply uniform validity of confidence regions over large, interesting classes of models.   Figure \ref{fig:figure0} (right panel) illustrates the result (\ref{normal}) by showing that the finite-sample distribution of our post-double-selection estimator is very close to the normal distribution.  In contrast, Figure \ref{fig:figure0} (left panel)
illustrates the classical problem with the  traditional post-single-selection estimator based on (\ref{eq: appPL1}), showing that its distribution is bimodal and sharply deviates from the normal distribution.  Finally, it is worth noting that the estimator achieves the semi-parametric efficiency bound under homoscedasticity.

\begin{figure}
\includegraphics[width=\textwidth]{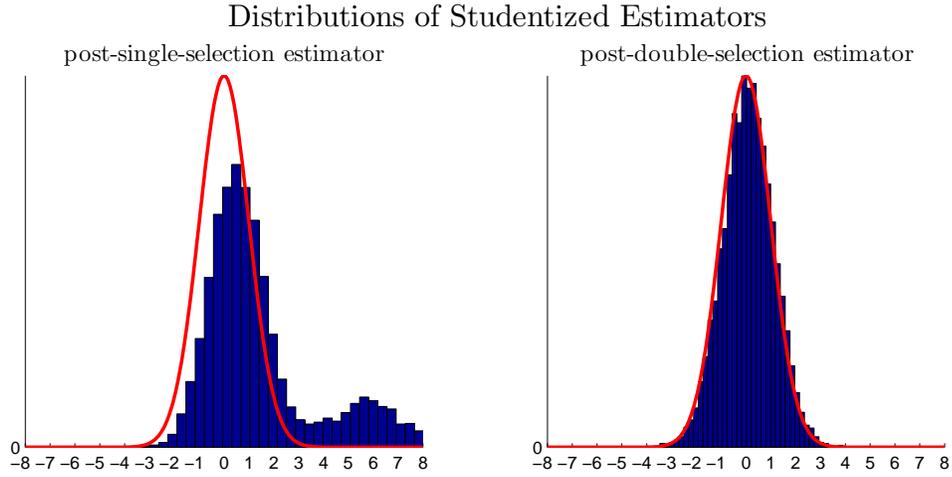}
	\label{fig:figure0}
\caption{\footnotesize The finite-sample distributions (densities) of the standard post-single selection estimator (left panel) and of our proposed post-double selection estimator (right panel).  The distributions are given for centered and studentized quantities.  The results are based on 10000 replications of Design 1 described in Section 6, with $R^2$'s in equation  (\ref{eq: RPL1}) and (\ref{eq: RPL2})  set to $0.5$.}
\end{figure}

\subsection{Selection of controls via feasible Lasso Methods}\label{Sec:FeasibleLasso} Here we describe feasible variable selection via Lasso. Note that each of the regression equations above is of the form
$$
\tilde y_i = \underbrace{\tilde x_i'\beta_0 + r_i}_{f(\tilde z_i)} + \epsilon_i,
$$
where $f(\tilde z_i)$ is the regression function, $\tilde x_i'\beta_0$ is the approximation
based on the dictionary $\tilde x_i=P(\tilde z_i)$, $r_i$ is the approximation error, and $\epsilon_i$ is the error. The Lasso estimator is defined as a solution to \begin{equation}\label{Def:LASSOmain}
 \min_{\beta \in \Bbb{R}^p} \En[(\tilde y_{i} - \tilde x_i'\beta)^2] +  \frac{\lambda}{n} \|\beta \|_{1},
\end{equation}
where $\|\beta\|_{1} = \sum_{j=1}^p | \beta_j|$; see \cite{FF:1993} and \cite{T1996}.  The kinked nature of the penalty function induces
the solution  $\widehat \beta$ to have many zeroes, and thus the Lasso solution may be used for model selection. The selected
model $\hat T = \text{support}(\widehat \beta)$ is often used for further refitting by least squares,
leading to the so called post-Lasso or Gauss-Lasso estimator, see, e.g., \citen{BC-PostLASSO}. The Lasso estimator/selector is computationally attractive because it minimizes a convex function. In the homoskedastic Gaussian case, a basic choice for penalty level suggested by \citen{BickelRitovTsybakov2009} is
 \begin{equation}\label{Def:LambdaLASSOboound0}\lambda = 2 \cdot c \sigma  \sqrt{2 n \log(2p/\conflvl)},
 \end{equation}
where $c>1$, $1-\gamma$ is a confidence level that needs to be set close to 1, and $\sigma$ is the standard deviation of the noise. The formal motivation for this penalty is that it leads to near-optimal rates of convergence of the estimator under approximate sparsity. The good behavior of the estimator of $\beta_0$ in turn implies good approximation properties of the selected model $\widehat T$, as noted in  \citen{BC-PostLASSO}. Unfortunately, even in the homoskedastic case the penalty level specified above is not feasible since it depends on the unknown $\sigma$.

\citen{BellChenChernHans:nonGauss} formulate a feasible Lasso estimator/selector $\widehat \beta$ geared for heteroscedastic, non-Gaussian cases, which solves
\begin{equation}\label{Def:LASSOmain2}
 \min_{\beta \in \Bbb{R}^p} \En[(\tilde y_{i} - \tilde x_i'\beta)^2] +  \frac{\lambda}{n} \|\hat \Psi \beta\|_1,
\end{equation} where $\hat \Psi = {\rm diag}(\hat l_1,\ldots,\hat l_p)$ is a diagonal matrix of penalty loadings. The penalty level $\lambda$
and loadings $\hat l_j$'s are set as
 \begin{equation}\label{Def:LambdaLASSOboound}\lambda = 2 \cdot c \sqrt{n}\Phi^{-1}(1-\conflvl/2p)  \mbox{ and }  \hat l_j = \ l_j + o_{P}(1) , \ \ l_j=\sqrt{\En[\tilde x_{ij}^2  \epsilon_i^2]},  \text{ uniformly in } j=1,\ldots,p,
 \end{equation}
where $c>1$  and $1-\gamma$ is a confidence level.\footnote{Practical recommendations include the choice $c=1.1$ and $\gamma=.05$.}  The $\l_j$'s are ideal penalty loadings that are not observed, and we estimate $\l_j$ by $\hat \l_j$ obtained via an iteration method given in Appendix A. We refer to the resulting feasible Lasso method as the \emph{Iterated Lasso}. The estimator $\widehat \beta$ has statistical performance that is similar to that of the (infeasible) Lasso described above in Gaussian cases and delivers similar performance in non-Gaussian, heteroscedastic cases; see \citen{BellChenChernHans:nonGauss}. In this paper, we only use $\widehat \beta$ as a model selection device.  Specifically, we only make use of
$$
\widehat T = \text{support}(\widehat \beta),
$$
the labels of the regressors with non-zero estimated coefficients. We show that the selected model $\widehat T$ has good approximation properties for the regression function $f$ under approximate sparsity in Section 3.

\citen{BCW-SqLASSO} propose another feasible variant of Lasso called the \emph{Square-root Lasso} estimator, $\widehat \beta$, defined as a solution to
\begin{equation}\label{Def:SQLASSOmain}
\min_{\beta \in \Bbb{R}^p} \sqrt{\En[(\tilde y_{i} - \tilde x_i'\beta)^2]} +  \frac{\lambda}{n}  \| \hat \Psi \beta \|_{1},
\end{equation}
with the penalty level
 \begin{equation}\label{lambda: root lasso}
 \lambda = c \cdot \sqrt{n}\Phi^{-1}(1-\conflvl/2p),
 \end{equation}
where $c>1$, $\gamma \in (0,1)$ is a confidence level, and
$\hat \Psi = {\rm diag}(\hat l_1,\ldots,\hat l_p)$ is a diagonal matrix of penalty loadings.  The main attractive feature of (\ref{Def:SQLASSOmain}) is that one can set $\hat \l_j =\{\En[\tilde x_{ij}^2]\}^{1/2}$ which depends only on observed data in the homoscedastic case.

In the heteroscedastic case, we would like to choose $\hat l_j$ so that
 \begin{equation}\label{loadings: root lasso}
l_j + o_P(1) \leq \hat l_j \lesssim_P l_j,   \text{ where } \ l_j= \{\En[\tilde x_{ij}^2\epsilon_i^2]]/\En[\epsilon_i^2]\}^{1/2}, \text{ uniformly in } j =1,...,p.
\end{equation}
As a simple bound, we could use $\hat \l_j = 2 \{\En[\tilde x_{ij}^4]\}^{1/4}$
 since $$\{\En[\tilde x_{ij}^2\epsilon_i^2]]/\En[\epsilon_i^2]\}^{1/2} \leq \{\En[\tilde x_{ij}^4]\}^{1/4}\{\En[\epsilon_i^4]\}^{1/4}/\{\En[\epsilon_i^2]\}^{1/2}.$$
This bound gives $l_j + o_P(1) \leq \hat \l_j$
if $\{\En[\epsilon_i^4]\}^{1/4}/\{\En[\epsilon_i^2]\}^{1/2} \leq 2 + o_P(1)$,
which covers a wide class of marginal distributions for error $\epsilon_i$.  For example, all $t$-distributions with degrees of freedom greater than five satisfy this condition.  As in the previous case, we can also iteratively re-estimate the penalty loadings using estimates of the $\epsilon_i$'s to approximate the ideal penalty loadings:
 \begin{equation}\label{refined loadings: root lasso}
\hat l_j = l_j + o_P(1), \text{ uniformly in } j =1,...,p.
\end{equation}
The resulting Square-root Lasso and post-Square-root Lasso estimators based on these penalty loadings achieve near optimal rates of convergence even in non-Gaussian, heteroscedastic cases. This good performance implies
good approximation properties for the selected model $\widehat T$.

In what follows, we shall use the term \emph{feasible Lasso} to refer to  either the Iterated Lasso estimator $\widehat \beta $ solving
(\ref{Def:LASSOmain2})-(\ref{Def:LambdaLASSOboound}) or the Square-root Lasso estimator $\widehat \beta$ solving (\ref{Def:SQLASSOmain})-(\ref{loadings: root lasso}) with $c > 1$ and $1-\conflvl$ set such that \begin{equation}\label{Def: conf level}
\conflvl = o(1) \textrm{ and } \log(1/\conflvl) \lesssim \log (p \vee n).\end{equation}

\section{Theory of Estimation and Inference}\label{Sec:LargeSample}

\subsection{Regularity Conditions}
In this section, we provide regularity conditions that are sufficient for validity of the main estimation and inference result. We begin
by stating our main condition, which contains the previously defined approximate sparsity as well as other more technical assumptions.
Throughout the paper, we let $c$, $C$, and $q$ be  absolute constants, and let $\ell_n \nearrow \infty, \delta_n \searrow 0$, and $\Delta_n \searrow 0$ be sequences of absolute positive constants. By absolute constants, we  mean constants that are given, and do not depend the dgp $\Pr= \Pr_n$.

We  assume that for each $n$ the following condition holds on dgp $\Pr = \Pr_n$.

\textbf{Condition ASTE ($\Pr$)}. \textit{(i)  $\{(y_{i}, d_i, z_i), i = 1,...,n\}$ are i.n.i.d. vectors on $(\Omega, \mathcal{F}, \Pr)$ that obey the  model
(\ref{eq: PL1})-(\ref{eq: PL2}), and  the vector  $x_i = P(z_i)$ is a dictionary of transformations of $z_i$, which  may depend on $n$ but not on $\Pr$. (ii)  The true parameter value $\alpha_0$, which may depend on $\Pr$,  is bounded, $|\alpha_0| \leq C$. (iii) Functions $m$ and $g$ admit
an approximately sparse form.  Namely there exists $s \geq 1$ and $\beta_{m0}$ and $\beta_{g0}$, which depend on $n$ and $\Pr$, such that
\begin{eqnarray}
&&m(z_i) = x_i' \beta_{m0} + r_{mi},  \ \ \|\beta_{m0}\|_0 \leq s, \ \  \{ \barEp [r_{mi}^2]\}^{1/2} \leq C \sqrt{s/n}, \\
&&g(z_i) = x_i' \beta_{g0} + r_{gi},  \ \ \ \ \ \|\beta_{g0}\|_0 \leq s, \ \ \   \{ \barEp[r_{gi}^2]\}^{1/2} \leq  C \sqrt{s/n}.
\end{eqnarray}
(iv) The sparsity index obeys $s^2 \log^2 (p\vee n)/n \leq \delta_n$ and
 the size of the amelioration set obeys $ \hat s_3 \leq C (1\vee \hat s_1 \vee \hat s_2)$. (v)  For $\tilde v_i = v_i + r_{mi}$ and $\tilde \zeta_i = \zeta_i + r_{gi}$ we have $|\barEp[ \tilde v_i^2\tilde \zeta_i^2 ] - \barEp[ v_i^2\zeta_i^2 ]| \leq  \delta_n$, and $\barEp[|\tilde v_i|^q+|\tilde \zeta_i|^q] \leq C$ for some $q>4$. Moreover, $\max_{i\leq n} \| x_{i}\|^2_\infty s n^{-1/2+2/q}  \leq \delta_n$ wp $1-\Delta_n$.}

\begin{remark}  The approximate sparsity (iii) and the growth condition (iv) are the main conditions for establishing the key inferential result.  We present a number of primitive examples to show that these conditions contain standard models used in empirical research as well as more flexible models.  Condition (iv) requires that the size $\hat s_3$ of the amelioration set $\hat I_3$ should not be substantially larger than the size of the set of variables selected by the Lasso method. Simply put, if we decide to include controls in addition to those selected by Lasso,  the total number of additions should not dominate the number of controls selected by Lasso.  This and other conditions will ensure that the total number $\hat s$ of controls obeys $ \hat s \lesssim_P s$, and we also require that $s^2 \log^2 (p\vee n)/n \to 0$. This condition can be relaxed using the sample-splitting method of \citen{FanGuoHao2011}, which is done in the Supplementary Appendix.  Condition (v) is simply a set of sufficient conditions for consistent estimation of the variance of the double selection estimator. If the regressors are uniformly bounded and the approximation errors are going to zero a.s., it is implied by other conditions stated below; and it can also be demonstrated under other sorts of more primitive conditions. \qed
\end{remark}

The next condition concerns the  behavior of the Gram matrix $\En [x_ix_i']$.
Whenever $p>n$, the empirical Gram matrix $\En[x_ix_i']$ does not
have full rank and in principle is not well-behaved. However, we only need good behavior of smaller submatrices. Define
the minimal and maximal $m$-sparse eigenvalue of a semi-definite matrix $M$ as
\begin{equation}\label{Def:RSE1}
\semin{m}[M] : = \min_{1\leq \|\delta \|_{0} \leqslant m} \frac{  \delta 'M \delta  }{\|\delta\|^2} \ \ \mbox{and} \ \ \hfill
 \semax{m}[M] : = \max_{1\leq \|\delta \|_{0} \leqslant m } \frac{ \delta 'M \delta }{\|\delta\|^2}.
\end{equation}
To assume that $\semin{m}[\En [x_ix_i']] >0$ requires that all empirical Gram submatrices formed by any $m$ components of $x_i$  are positive definite. We shall employ the following condition as a sufficient condition for our results.

\textbf{Condition SE ($\Pr$).} \textit{ There is an absolute sequence of constants  $\ell_n \to \infty$ such that
the maximal and minimal $\ell_n s$-sparse eigenvalues are bounded from below and away from zero, namely with probability at least $1-\Delta_n$,
$$\kappa' \leq \semin{\ell_n s}[\En [x_ix_i']] \leq \semax{\ell_n s}[\En [x_ix_i']] \leq \kappa'',$$
where $0< \kappa' <  \kappa'' < \infty$ are absolute constants.
}
\begin{remark}
It is well-known that Condition SE is quite plausible for many designs of interest. For instance, Condition SE holds if 
\begin{itemize}
\item[(a)] $ x_i$, $i = 1,\ldots,n$, are i.i.d. zero-mean sub-Gaussian random vectors that have population Gram matrix $\Ep[ x_i  x_i']$ with  minimal and maximal $s\log n$-sparse eigenvalues bounded away from zero and from above by absolute constants where $s(\log n )(\log p)/n \leq \delta_n \to 0$;
\item[(b)] $ x_i$, $i=1,\ldots,n$, are i.i.d. bounded zero-mean random vectors with $\| x_i\|_\infty \leq K_n$ a.s. that have population Gram  matrix $\Ep[ x_i  x_i']$ with minimal and maximal $s\log n$-sparse eigenvalues bounded from above and away from zero by absolute constants where $K_n^2s(\log^3 n)\{\log(p\vee n)\} /n \leq \delta_n \to 0$.
\end{itemize}
The claim (a) holds by Theorem 3.2 in \citen{RudelsonZhou2011} (see also \citen{Zhou2009a} and \citen{Baraniuketal2008}) and claim (b) holds by Lemma 1 in \citen{BC-PostLASSO}  or by Theorem 1.8 \citen{RudelsonZhou2011}. Recall that a standard assumption in econometric research is to assume that the population Gram matrix $\Ep[x_i x_i']$ has eigenvalues bounded from above and away from zero, see e.g. \citen{newey:series}. The conditions above allow for this and more general behavior, requiring only that the $s \log n$ sparse eigenvalues of the population Gram matrix $\Ep[x_i x_i']$ are bounded from below and from above. 
\qed \end{remark}

The next condition imposes moment conditions on the structural errors and regressors.

\textbf{Condition SM ($\Pr$)}.  \textit{ There are absolute constants
$0< c< C < \infty$ and $4< q < \infty$ such that for  $(\tilde y_i, \epsilon_i) = (y_i, \zeta_i) $ and $(\tilde y_i, \epsilon_i) = (d_i, v_i)$ the following conditions hold:
\begin{itemize}
\item[(i)]  $\displaystyle \barEp [|d_i|^q] \leq C, \ \  c \leq \Ep[\zeta_i^2\mid x_i, v_i]\leq C   \ \mbox{and}  \ c\leq \Ep[v_i^2\mid x_i]\leq C \  \mbox{a.s. }  1 \leq i \leq n$,
    \item[(ii)] $\displaystyle \barEp[|\epsilon_i|^q]+\barEp [\tilde y_i^2] + \max_{1\leq j\leq p}  \{ \barEp[x_{ij}^2{\tilde y}_i^2]+\barEp[|x_{ij}^3 \epsilon_i^3|]+ 1/\barEp[ x_{ij}^2] \} \leq C$,
        \item[(iii)]$ \displaystyle \log^3 p / n \leq \delta_n$, 
                \item[(iv)]$ \displaystyle \max_{1\leq j\leq p}  \{ |(\En-\barEp)[ x_{ij}^2\epsilon_i^2]|+|(\En-\barEp)[x_{ij}^2{\tilde y}_i^2]|\} + \max_{1\leq i\leq n}\| x_i\|_\infty^2 \frac{s\log(n\vee p)}{n} \leq \delta_n \text{ wp } 1-\Delta_n.$
\end{itemize}}
These conditions, which are rather mild, 
ensure good model selection performance of feasible Lasso applied to equations (\ref{eq: RPL1}) and (\ref{eq: RPL2}).  These conditions also allow us to invoke moderate deviation theorems for self-normalized sums from \citen{jing:etal} to bound some important error components.

\subsection{The Main Result} The following is the main result of this paper. It shows that
the post-double selection estimator is root-$n$ consistent and asymptotically normal. Under homoscedasticity this estimator achieves
the semi-parametric efficiency bound.  The result also verifies that plug-in estimates of the standard errors are consistent.

\begin{theorem}[Estimation and Inference on Treatment Effects]\label{theorem:inference}  Let $\{\Pr_n\}$ be a sequence of data-generating processes.  Assume conditions ASTE ($\Pr$),  SM  ($\Pr$), and SE  ($\Pr$)  hold for $\Pr = \Pr_n$ for each $n$. Then, the post-double-Lasso estimator $\check \alpha$, constructed in the previous section,  obeys as $n \to \infty$
$$
\sigma_n^{-1} \sqrt{n} (\check \alpha - \alpha_0) \rightsquigarrow N(0,1),
$$
where $\sigma^2_n=  [\barEp v_i^2]^{-1}\barEp[ v_i^2\zeta_i^2] [\barEp v_i^2]^{-1}$.  Moreover, the result continues to apply if $\sigma^2_n$ is replaced by $\hat \sigma^2_n = [\En \hat v_i^2]^{-1}\En[\hat v_i^2\hat \zeta_i^2][\En \hat v_i^2]^{-1}$, for $\hat \zeta_i := [y_i - d_i\check \alpha - x_i'\check \beta]\{n/(n - \hat s-1)\}^{1/2}$
and $\hat v_i:=d_i - x_i'\hat\beta$, $i=1,\ldots,n$ where $\hat \beta \in \arg\min_\beta \{\En[(d_i-x_i'\beta)^2]:\beta_j=0, \forall j\notin \widehat I\}$.
\end{theorem}

A consequence of this result is the following corollary.

\begin{corollary}[\textbf{Uniformly Valid Confidence Intervals}]  (i) Let $\mathbf{P}_n$ be the collection of all data-generating processes $\Pr$ for which conditions ASTE($\Pr$), SM ($\Pr$), and SE ($\Pr$) hold for given $n$. Let $c(1-\xi) = \Phi^{-1} (1-\xi/2)$. Then as $n \to \infty$, uniformly in $\Pr \in \mathbf{P}_n$
$$
\Pr \left ( \alpha_0 \in [ \check \alpha \pm c(1-\xi) \hat \sigma_n /\sqrt{n}]\right) \to 1- \xi .
$$
(ii) Let $\mathbf{P} = \cap_{n \geq n_0} \mathbf{P}_n$ be the collection of data-generating processes for which the conditions above hold for all $n \geq n_0$ for some $n_0$.  Then
as $n \to \infty$, uniformly in $\Pr \in \mathbf{P}$
$$
\Pr \left ( \alpha_0 \in [ \check \alpha \pm c(1-\xi) \hat \sigma_n /\sqrt{n}]\right) \to 1- \xi.
$$
\end{corollary}

By exploiting both equations (\ref{eq: appPL1}) and (\ref{eq: appPL2}) for model selection, the post-double-selection method creates the necessary adaptivity that makes it robust to imperfect model selection.  Robustness of the post-double selection method is reflected in the fact that Theorem \ref{theorem:inference} permits the data-generating process to change with $n$.  Thus, the conclusions of the theorem are valid for a wide variety of sequences of data-generating processes which in turn define the regions $\mathbf{P}$ of uniform validity of the resulting confidence sets. These regions appear to be substantial, as we demonstrate via a sequence of theoretical and numerical examples in Section 5 and 6.  In contrast, the standard post-selection method based on (\ref{eq: appPL1}) generates non-robust confidence intervals. 

\begin{remark} Our approach to uniformity analysis is most similar to that  of \citen{romano:uniform}, Theorem 4. It  proceeds under triangular array asymptotics, with the sequence of dgps obeying certain constraints; then these results imply uniformity over sets of dgps that obey the constraints for all sample sizes.  This approach is also  similar to the classical central limit theorems for sample means under triangular arrays, and does not require the dgps to be parametrically (or otherwise tightly) specified, which then translates  into uniformity of confidence regions.  This  approach is somewhat different in spirit to the generic uniformity analysis suggested by \citen{andrews:cheng:guggen}.
\qed\end{remark}

\begin{remark} Uniformity holds over a large class of approximately sparse models, which cover conventional models used in series estimation of partially linear models as shown in Section 5. Of course, for every interesting class of models and any inference method, one could find an even bigger class of models where the uniformity does not apply. In particular, our models do not cover models with many small coefficients.  In the series case, a model with many small coefficients corresponds to a deviation from smoothness towards highly non-smooth functions, namely functions generated as realized paths of an approximate white noise process. The fact that our results do not cover such models motivates further research work on inference procedures that have robustness properties to deviations from the given class of models that are deemed important. In the simulations in Section 6, we consider incorporating the ridge fit along the other controls to be selected over using lasso to build extra robustness against ``many small coefficients" deviations away from approximately sparse models.  \qed \end{remark}

\subsection{Auxiliary Results on Model Selection via Lasso and Post-Lasso}\label{Sec:ResultsHDSM}

The post-double-selection estimator applies the least squares estimator to the union of variables selected for equations (\ref{eq: RPL1}) and (\ref{eq: RPL2}) via feasible Lasso. Therefore, the
model selection properties of feasible Lasso as well as properties
of least squares estimates for $m$ and $g$ based on the selected
model play an important role in the derivation of the main result.  The purpose of this section is to describe these properties.  The proof of Theorem 1 relies on these properties.

Note that each of the regression models (\ref{eq: RPL1})-(\ref{eq: RPL2}) obeys the following conditions.

\textbf{Condition ASM.}\textit{  Let $\{\Pr_n\}$ be a sequence of data-generating processes.  For each $n$, we have  data $\{(\tilde y_i,\tilde z_i,\tilde x_i=P(\tilde z_i))  :  1 \leq i \leq  n\}$  defined on $(\Omega, \mathcal{A}, \Pr_n) $ consisting of i.n.i.d vectors that obey the following approximately sparse regression model for each $n$:
 \begin{eqnarray*}
 &  & \tilde y_i = f(\tilde z_i) + \epsilon_i = \tilde x_i'\beta_0 + r_i + \epsilon_i,  \\
  & & \Ep[\epsilon_i\mid \tilde x_i ]=0,    \barEp[\epsilon_i^2] = \sigma^2,\\
  &  & \|\beta_0\|_0\leq s,\  \barEp[r_i^2]\lesssim  \sigma^2 s/n.
 \end{eqnarray*}
 }
Let  $\widehat T$ denote the model
selected by the feasible Lasso estimator $\hat\beta$:
$$\widehat T = \supp( \hat \beta ) = \{ j \in \{1,\ldots,p\} \ : \ |\hat\beta_j| > 0\},$$
The Post-Lasso estimator $\widetilde \beta$ is  is ordinary least squares applied to the data after removing the regressors that were not selected by the feasible Lasso: \begin{equation}\label{Def:TwoStep} \widetilde \beta \in \arg\min_{\beta \in \mathbb{R}^p} \ \En[(\tilde y_i-\tilde x_i'\beta)^2]\ \ :  \ \ \beta_j = 0 \text{ for each } j \notin \widehat T.
\end{equation}

The following  regularity conditions are imposed to deal with  non-Gaussian, heteroscedastic errors.


\textbf{Condition RF.}  \textit{In addition to ASTE, we have
\begin{itemize}
\item[(i)] $\log^{3} p / n \to 0 \ \text{ and } \ s \log (p\vee n) /n \to 0$,
\item[(ii)] $ \barEp[\tilde y_i^2] + \max_{1\leq j\leq p} \{\barEp[\tilde x_{ij}^2\tilde y_i^2]+\barEp[|\tilde x_{ij}^3 \epsilon_i^3|]+ 1/\barEp[\tilde x_{ij}^2\epsilon_i^2]\} \lesssim 1$,
\item[(iii)] $\displaystyle \max_{1\leq j\leq p} \{|(\En-\barEp)[\tilde x_{ij}^2\epsilon_i^2]|+|(\En-\barEp)[\tilde x_{ij}^2\tilde y_i^2]|\} + \max_{1\leq i\leq n}\|\tilde x_i\|_\infty^2 \frac{s\log(n\vee p)}{n} = o_P(1)$. \end{itemize}
}

The main auxiliary result that we use in proving the main result is as follows.

\begin{lemma}[Model Selection Properties of Lasso
and Properties of Post-Lasso]\label{corollary3:postrate} Let $\{\Pr_n\}$ be a sequence of data-generating processes.  Suppose that conditions ASM and RF hold, and that Condition SE $(\Pr_n)$ holds for $\En[\tilde x_i \tilde x_i']$. Consider a feasible Lasso estimator
with penalty level and loadings specified as in Section 3.3.

(i) Then the data-dependent model $\widehat T$ selected by a feasible Lasso estimator satisfies with probability approaching 1:
\begin{equation}\label{eq: sparsity bound}
\hat s = | \widehat T | \lesssim s
 \end{equation}
and
 \begin{equation}\label{eq: model selection bound}
\min_{\beta \in \Bbb{R}^p: \ \beta_j = 0 \  \forall j \not \in \widehat T} \sqrt{\En[ f(\tilde z_i) -  \tilde x_i'\beta]^2} \lesssim \sigma \sqrt{ \frac{ s \log (p \vee n)}{n} }.
 \end{equation}
(ii) The Post-Lasso estimator obeys
$$
\sqrt{\En[ f (\tilde z_i) -  \tilde x_i'\widetilde \beta]^2} \lesssim_P \sigma \sqrt{ \frac{ s \log (p \vee n)}{n} }.
$$
and
\begin{equation}
\|\widetilde \beta - \beta_0\| \lesssim_P \sqrt{\En[\{\tilde x_i'\widetilde \beta - \tilde x_i'\beta_0\}^2]} \lesssim_P \sigma \sqrt{ \frac{ s \log (p \vee n)}{n} }.
\end{equation}
\end{lemma}

Lemma \ref{corollary3:postrate} was derived in \citen{BellChenChernHans:nonGauss} for Iterated Lasso and by \citen{BCW-SqLASSO2} for Square-root Lasso.  These analyses build on the rate analysis of infeasible Lasso by \citen{BickelRitovTsybakov2009} and on sparsity analysis and rate analysis of Post-Lasso by \citen{BC-PostLASSO}. Lemma \ref{corollary3:postrate} shows that feasible Lasso methods select a model $\widehat T$ that provides a high-quality approximation to the regression function $f(\tilde z_i)$; i.e. they find a sparse model that can  approximate the function at the ``near-oracle" rate $\sqrt{s/n} \sqrt{\log (p\vee n)}$. If we knew the ``best" approximating model $ T= \supp(\beta_0)$, we could achieve the ``oracle" rate of $\sqrt{s/n}$.  Note that Lasso methods generally will not recover $T$ perfectly. Moreover, no method can recover $T$ perfectly in general, except under the restrictive condition that all non-zero coefficients in $\beta_0$ are bounded away from zero by a factor that exceeds estimation error.  We do not require this condition to hold in our results.  All that we need is that the selected model $\widehat T$ can approximate the regression function well
and that the size of the selected model, $\hat s = |\widehat T|$, is of the same stochastic order as $s = |T|$.  This condition holds in many cases in which some non-zero coefficients are close to zero.

The lemma above also shows that feasible Post-Lasso achieves the same near-oracle rate as feasible Lasso.  The coincidence in rates occurs despite the fact that feasible Lasso will in general fail to correctly select the best-approximating model $T$ as a subset of the variables selected; that is, $T \not \subseteq \widehat T$.  The intuition for this result is that any components of $T$ that feasible Lasso misses are  unlikely to be important; otherwise, (\ref{eq: model selection bound}) would be impossible.  This result was first derived in the context of median regression by \citen{BC-SparseQR} and extended to least squares in reference cited above.


\section{Generalization: Inference after Double Selection by a Generic  Selection Method}

The conditions provided so far are simply a set sufficient conditions
that are tied to the use of Lasso as the model selector.  The purpose
of this section is to prove that the main results apply to any other model selection method  that  is able to select a sparse model with good approximation properties. As in the case of Lasso, we
allow for imperfect model selection. Next we state a high-level condition that summarizes a sufficient condition on the performance of a model selection method that allows the post-double selection estimator to attain good inferential properties.

%
%


\textbf{Condition HLMS ($\Pr$)}.\textit{ A model selector provides
possibly data-dependent sets  $ \widehat I_1 \cup \widehat I_2 \subseteq \widehat I  \subset \{1,...,p\}$ of covariate names such that, with probability $1-\Delta_n$,   $|\widehat I| \leq C s$ and  $$\displaystyle \min_{\beta: \beta_j=0, j\not\in \widehat I_1} \sqrt{\En[ (m(z_i)-x_i'\beta)^2]} \leq \delta_n  n^{-1/4} \textrm{ and } \min_{\beta: \beta_j=0, j\not\in \widehat I_2} \sqrt{\En[ (g(z_i)-x_i'\beta)^2]} \leq \delta_n  n^{-1/4}.$$}

Condition HLMS requires that with high probability the selected models are sparse and generates a good approximation for the functions $g$ and $m$. Examples of  methods producing such models include the
Dantzig selector \cite{CandesTao2007}, feasible Dantzig selector \cite{gautier:tsybakov},
Bridge estimator \cite{HHS2008},
SCAD penalized least squares \cite{FanLi2001},
and thresholded Lasso \cite{BC-PostLASSO},
to name a few. We emphasize that, similarly to the previous arguments, we allow for imperfect model selection.

The following result establishes the inferential properties of a generic post-double-selection estimator.

\begin{theorem}[Estimation and Inference on Treatment Effects under High-Level Model Selection]\label{theorem:inferenceHLMS}
Let $\{\Pr_n\}$ be a sequence of data-generating processes
 and the model selection device be such that conditions ASTE ($\Pr$),  SM  ($\Pr$), SE  ($\Pr$), and  HLSM($\Pr$) hold for $\Pr = \Pr_n$  for each $n$.  Then the generic post-double-selection estimator
$\check \alpha$ based on $\widehat I$, as defined in (\ref{eq: define pds}), obeys
$$
([\barEp v_i^2]^{-1}\barEp[ v_i^2\zeta_i^2] [\barEp v_i^2]^{-1})^{-1/2} \sqrt{n} (\check \alpha - \alpha_0) \rightsquigarrow N(0,1).
$$
Moreover, the result continues to apply if $\barEp[v_i^2]$ and $\barEp[v_i^2\zeta_i^2]$ are replaced by  $\En[\hat v_i^2]$ and $\En[\hat v_i^2\hat \zeta_i^2]$ for $\hat \zeta_i := [y_i - d_i\check \alpha - x_i'\check \beta]\{n/(n - \hat s-1)\}^{1/2}$
and $\hat v_i:=d_i - x_i'\hat\beta$, $i=1,\ldots,n$ where $\hat \beta \in \arg\min_\beta \{\En[(d_i-x_i'\beta)^2]:\beta_j=0, \forall j\notin \widehat I\}$.
\end{theorem}

Theorem \ref{theorem:inferenceHLMS} can also be used to establish uniformly valid confidence intervals as shown is the following corollary.

\begin{corollary}[\textbf{Uniformly Valid Confidence Intervals}]  (i) Let $\mathbf{P}_n$ be the collection of all data-generating processes $\Pr$ for which conditions ASTE($\Pr$), SM ($\Pr$), SE ($\Pr$), and HLSM ($\Pr$) hold for given $n$.  Let $c(1-\xi) = \Phi^{-1} (1-\xi/2)$. Then as $n \to \infty$, uniformly in $\Pr \in \mathbf{P}_n$
$$
\Pr \left ( \alpha_0 \in [ \check \alpha \pm c(1-\xi) \hat \sigma_n /\sqrt{n}]\right) \to 1- \xi .
$$
(ii) Let $\mathbf{P} = \cap_{n \geq n_0} \mathbf{P}_n$ be the collection of data-generating processes for which the conditions above hold for all $n \geq n_0$ for some $n_0$.  Then
as $n \to \infty$, uniformly in $\Pr \in \mathbf{P}$
$$
\Pr \left ( \alpha_0 \in [ \check \alpha \pm c(1-\xi) \hat \sigma_n /\sqrt{n}]\right) \to 1- \xi.
$$
\end{corollary}

\section{Theoretical Examples}\label{Sec:Examples}

The purpose of this section is to give a sequence of examples -- progressing from simple to somewhat involved -- that highlight the range of the applicability and robustness of the proposed method.  In these examples, we specify primitive conditions which cover a broad range of applications including nonparametric models and high-dimensional parametric models. We emphasize that our main regularity conditions cover even more general models which combine various features of these examples such as models with both nonparametric and high-dimensional parametric components.

In all examples, the model is
\begin{equation}\label{model}
\begin{array}{ll}
 y_i = d_i \alpha_0 +  g(z_i) + \zeta_i, & \Ep[\zeta_i \mid z_i, v_i] = 0, \\
 d_i = m(z_i) + v_i, &  \Ep [v_i \mid z_i] = 0, \\
\end{array}
\end{equation}
however, the structure for $g$ and $m$ will vary across examples, and so will the assumptions on the error terms $\zeta_i$ and $v_i$.

We start out with a simple example, in which  the dimension $p$ of the regressors is fixed.  In practical terms this example approximates cases with  $p$ small compared to $n$. This simple example is important since standard post-single-selection methods fail even in this simple case.  Specifically, they produce confidence intervals that are \emph{not} valid uniformly in the underlying data-generating process; see \citen{leeb:potscher:pms}.  In contrast, the post-double-selection method produces confidence intervals that are valid uniformly in the underlying data-generating process.

\textbf{Example 1.} (Parametric Model with Fixed $p$.)
Consider  $(\Omega, \mathcal{A}, \Pr)$ as the probability space, on which we have $(y_i, z_i, d_i)$
as i.i.d. vectors for $i=1,...,n$ obeying the model (\ref{model}) with
\begin{equation}
\begin{array}{ll}
 g(z_i) =  & \sum_{j=1}^p \beta_{g0j}z_{ij} , \\
m(z_i) = & \sum_{j=1}^p \beta_{m0j}z_{ij}.   \\
\end{array}
\end{equation}
For estimation we use $x_i = (z_{ij}, j=1,...,p)'$. We assume that  there are some absolute constants $0< b < B<\infty$, $q_x \geq q > 4$,  with $4/q_x + 4/q < 1$,  such that
\begin{equation}
 \begin{array}{ll}
& \Ep[\|x_i\|^{ q_x}] \leq B, \ \  \|\alpha_0\| + \| \beta_{g0}\|+ \| \beta_{m0}\| \leq B, \ \  b \leq \lambda_{\text{min}} (\Ep[x_i x_i']),\\
&  b \leq \Ep[\zeta_i^2 \mid x_i, v_i], \ \ \Ep[|\zeta_i^q| \mid x_i, v_i] \leq B, \ \ b\leq \Ep[v_i^2 \mid x_i], \ \  \Ep[|v_i^q| \mid x_i] \leq B.

\end{array}
\end{equation}

Let $\mathbf{P}$ be the collection of all regression models
$\Pr$ that obey the conditions set forth above for all $n$ for the given constants $(p, b, B, q_x, q)$.  Then, as established in Appendix \ref{Sec:verify},  any $\Pr \in \mathbf{P}$ obeys Conditions ASTE ($\Pr$) with $s=p$,  SE ($\Pr$), and SM ($\Pr$) for all $n \geq n_0$, with the constants $n_0$ and $( \kappa', \kappa'', c, C)$ and sequences $\Delta_n$ and $\delta_n$ in those conditions depending
only on $(p, b, B, q_x, q)$. Therefore, the conclusions of Theorem 1  hold for any sequence $\Pr_n \in \mathbf{P}$, and the conclusions of Corollary 1 on the uniform validity of confidence intervals apply
uniformly in $\Pr \in \mathbf{P}$.  \qed

The next  examples are more substantial
and include infinite-dimensional models which we approximate
with linear functional forms with potentially
very many regressors, $p \gg n$.  The key to estimation
in these models is a smoothness condition which requires
regression coefficients to decay at some rates.  In series
estimation, this condition is often directly connected to smoothness
of the regression function.

Let $\aa$ and $A$ be positive constants. We shall say that a sequence of coefficients $$\theta = \{ \theta_j, j=1,2,...\}$$ is $\aa$-smooth with constant $A$ if
$$
| \theta_j| \leq A j^{-\aa},  \  j=1,2,... ,
$$
which will be denoted as $\theta \in S^{\aa}_A$. We shall say that a sequence of coefficients $\theta = \{ \theta_j, j=1,2,...\}$ is $\aa$-smooth with constant $A$ after $p$-rearrangement if
$$
| \theta_{(j)}| \leq A j^{-\aa},  \  j=1,2,..., p, \ \  | \theta_j| \leq A j^{-\aa},  \  j=p+1, p+2,...,
$$
which will be denoted as $\theta \in S^{\aa}_A(p)$, where
$\{| \theta_{(j)}|, j=1,...,p\}$ denotes the decreasing rearrangement of the numbers  $\{ | \theta_j|, j =1,...,p\}$.   Since $S^{\aa}_A \subset S^{\aa}_A(p)$, the second kind
of smoothness is strictly more general than the first kind.

Here we use the term ``smoothness" motivated by Fourier series analysis where smoothness of functions often translates into smoothness of  the Fourier coefficients in the sense that is stated above; see, e.g., \citen{kerk:picard}. 
For example, if a function $h: [0,1]^d \mapsto \Bbb{R}$
possesses $r>0$ continuous derivatives uniformly bounded by a constant $M$ and the terms $P_j$ are compactly supported
Daubechies wavelets, then  $h$ can be represented
as $h(z) = \sum_{j=1}^\infty P_j(z) \theta_{hj}$, with $|\theta_{hj}| \leq A j^{-r/d-1/2}$ for some constant $A$; see \citen{kerk:picard}.
We also note that the second kind of smoothness is considerably more general than the first since it allows relatively large coefficients to appear anywhere in the series of the first $p$ coefficients. In contrast, the first kind of smoothness only allows relatively large coefficients among the early terms in the series.  Lasso-type methods are specifically designed to deal with the generalized smoothness of the second kind and
perform equally well under both kinds of smoothness.   In the context
of series applications, smoothness of the second kind allows one to approximate functions that exhibit oscillatory phenomena or spikes, which are associated with ``high order" series terms.  An example
of this is  the wage function example given in \citen{BCH2011:InferenceGauss}.

Before we proceed to other examples we discuss a way
 to generate sparse approximations in infinite-dimensional examples. Consider, for example, a function $h$ that can be represented a.s. as $h(z_i) = \sum_{j=1}^{\infty}\theta_{hj}P_j(z_i) $ with coefficients $\theta_h \in S_A^\aa(p)$.  In this case we can construct sparse approximations by simply thresholding to zero all coefficients smaller than $1/\sqrt{n}$ and with indices $j \geq p$. This generates a sparsity index $ s \leq A^{\frac{1}{a}} n^{\frac{1}{2a}}$. The non-zero coefficient could be further reoptimized by using the least squares projection.  More formally,   given a sparsity index $s>0$, a target function $h(z_i)$, and terms  $x_i =( P_j(z_i) : j=1,\ldots,p)'\in \RR^p$, we let
\begin{equation}\label{DefBetah}\begin{array}{rl}
\beta_{h0} := & \displaystyle \arg\min_{\|\beta\|_0\leq s } \Ep[ (h(z_i)-x_i'\beta)^2],\end{array}\end{equation}
and define $x_i'\beta_{h0}$ as the best $s$-sparse approximation to $h(z_i)$.

\textbf{Example 2.} (Gaussian Model with Very Large $p$.) Consider  $(\Omega, \mathcal{A}, \Pr)$ as the probability space on which we have $(y_i, z_i, d_i)$
as i.i.d. vectors for $i=1,...,n$ obeying the model (\ref{model}) with
\begin{equation}\label{gauss}
\begin{array}{rl}
g(z_i) = & \sum_{j=1}^\infty \theta_{gj}z_{ij},  \\
m(z_i) = & \sum_{j=1}^\infty \theta_{mj}z_{ij}.
\end{array}
\end{equation}
Assume that the infinite dimensional vector $w_i = (z_i', \zeta_i, v_i)'$  is jointly Gaussian with  minimal and maximal eigenvalues of the matrix (operator) $\Ep[w_i w_i']$  bounded below by an absolute constant $\underline{\kappa}>0$ and above by an absolute constant $\overline{\kappa}< \infty$. 

The main assumption that guarantees approximate sparsity is the smoothness condition on the coefficients. Let $\aa>1$
and $0<A<\infty$ be some absolute constants.  We require that the coefficients of the expansions in (\ref{gauss}) are $\aa$-smooth with constant $A$ after $p$-rearrangement, namely
$$
\theta_m = (\theta_{mj}, j=1,2,...) \in S^{\aa}_A(p), \ \ \theta_g = (\theta_{gj}, j=1,2,...) \in S^{\aa}_A(p).
$$
For estimation purposes we shall use $x_i = (z_{ij}, j=1,...,p)',$
and assume that $\|\alpha_0\|\leq B$ and $p = p_n$ obeys
$$
n^{[(1-\aa)/\aa]+\chi}\log^2(p\vee n) \leq\bar\delta_n,  \ \  A^{1/\aa}n^{\frac{1}{2\aa}}  \leq p \bar\delta_n,  \ \text{ and } \  \log^3 p /n \leq \bar \delta_n, $$
 for some absolute sequence $\bar\delta_n \searrow 0$ and absolute constants $B$ and $\chi > 0$.

Let $\mathbf{P}_n$ be the collection of all dgp
$\Pr$  that obey the conditions set forth in this example for a given $n$ and for the given constants  $(\underline{\kappa}, \overline{\kappa}, \aa, A, B,\chi)$ and sequences $p=p_n$ and $\bar\delta_n$. Then, as established in Appendix \ref{Sec:verify},   any $\Pr \in \mathbf{P}_n$ obeys Conditions ASTE ($\Pr$) with $s=A^{1/\aa} n^{\frac{1}{2\aa}}$, SE ($\Pr$), and SM ($\Pr$) for all $n \geq n_0$, with constants $n_0$ and $( \kappa', \kappa'', c, C)$ and sequences $\Delta_n$ and $\delta_n$ in those conditions depending
only on  $(\underline{\kappa}, \bar\kappa, \aa, A, B,\chi)$, $p$, and $\bar\delta_n$.  Therefore, the conclusions of Theorem 1  hold for any sequence $\Pr_n \in \mathbf{P}_n$, and the conclusions of Corollary 1 on the uniform validity of confidence intervals apply
uniformly for any $\Pr \in \mathbf{P}_n$.  In particular, these conclusions apply uniformly in  $  \Pr \in \mathbf{P} = \cap_{n \geq n_0} \mathbf{P}_n$.  \qed

\textbf{Example 3.} (Series Model with Very Large $p$.) Consider  $(\Omega, \mathcal{A}, \Pr)$ as the probability space, on which we have $(y_i, z_i, d_i)$
as i.i.d. vectors for $i=1,..., n$ obeying the model:
\begin{equation}\label{series}
\begin{array}{rl}
g(z_i) = & \sum_{j=1}^\infty \theta_{gj}P_{j}(z_i),\\
m(z_i) = & \sum_{j=1}^\infty \theta_{mj}P_{j}(z_i),
\end{array}
\end{equation}
where $z_i$ has support $[0,1]^d$ with density bounded from below by constant $\underline{f}>0$ and
above by constant $\bar f$, and $\{ P_j, j =1,2,..\}$ is an orthonormal
basis on $L^2[0,1]^d$ with bounded elements, i.e. $ \max_{z \in [0,1]^d}|P_j(z)| \leq B$ for all $j =1,2,...$.
Here all constants are taken to be absolute.  Examples of such orthonormal bases include 
canonical trigonometric bases. 

Let $\aa>1$
and $0<A<\infty$ be some absolute constants.  We require that the coefficients of the expansions in (\ref{series}) are $\aa$-smooth with constant $A$ after $p$-rearrangement, namely
$$
\theta_m = (\theta_{mj}, j=1,2,...) \in S^{\aa}_A(p), \ \ \theta_g = (\theta_{gj}, j=1,2,...) \in S^{\aa}_A(p).
$$

For estimation purposes we shall use
$x_i = (P_j(z_{i}), j=1,...,p)',$ and assume that $p = p_n$ obeys
$$
n^{(1-\aa)/\aa}\log^2(p\vee n) \leq\bar\delta_n, \ \ A^{1/\aa}n^{\frac{1}{2\aa}}  \leq p \bar\delta_n  \ \text{ and } \  \log^3 p /n \leq  \bar\delta_n, $$
 for some sequence of absolute constants $\bar\delta_n \searrow 0$. We assume that  there are some absolute constants $b>0$, $B<\infty$, $q > 4$, with $(1-\aa)/\aa+ 4/q <0$, such that
\begin{equation}
 \begin{array}{ll}
& \ \  \|\alpha_0\| \leq B,   \ b \leq \Ep[\zeta_i^2 \mid x_i,v_i], \ \ \Ep[|\zeta_i^q| \mid x_i,v_i] \leq B, \ \ b \leq \Ep[v_i^2 \mid x_i], \ \ \Ep[|v_i^q| \mid x_i] \leq B.
\end{array}
\end{equation}

Let $\mathbf{P}_n$ be the collection of all regression models
$\Pr$  that obey the conditions set forth above for a given $n$. Then, as established in Appendix \ref{Sec:verify}, any $\Pr \in \mathbf{P}_n$ obeys Conditions ASTE ($\Pr$) with $s=A^{1/\aa} n^{\frac{1}{2\aa}}$,  SE ($\Pr$), and SM ($\Pr$) for all $n \geq n_0$, with absolute constants in those conditions depending
only on $(\underline{f}, \bar f, \aa, A, b, B, q)$ and $\bar\delta_n$.  Therefore, the conclusions of Theorem 1  hold for any sequence $\Pr_n \in \mathbf{P}_n$, and the conclusions of Corollary 1 on the uniform validity of confidence intervals apply
uniformly for any $\Pr \in \mathbf{P}_n$.  In particular, as a special case, the same conclusion
applies uniformly in  $  \Pr \in \mathbf{P} = \cap_{n \geq n_0} \mathbf{P}_n$.  \qed

\section{Monte-Carlo Examples}\label{Sec:MonteCarlo}

 In this section, we examine the finite-sample properties of the post- double-selection method through a series of simulation exercises and compare its performance to that the standard post-single-selection method.

 All of the simulation results are based on the structural model
\begin{equation}\label{ModelMCPLMy}
y_i = d_i'\alpha_0 +  x_i'\theta_g + \sigma_{y}(d_i,x_i) \zeta_i, \ \  \zeta_i \sim N(0,1)
\end{equation}
where $p = \dim(x_i) = 200$, the covariates $ x_i \sim N(0,\Sigma)$ with $\Sigma_{kj} = (0.5)^{|j-k|}$, $\alpha_0 = .5$, and
the sample size $n$ is set to $100$.  In each design, we generate
\begin{equation}\label{ModelMCPLMd}
d_i =  x_i'\theta_m  + \sigma_{d}(x_i) v_i, \ \  v_i \sim N(0,1)
\end{equation}
with E[$\zeta_i v_i] = 0 $.
Inference results for all designs are based on conventional t-tests with standard errors calculated using the heteroscedasticity consistent jackknife variance estimator discussed in \citen{mackinnon:white}. Another option would be to use the standard error estimator recently proposed in \citen{CJN:PLMStandardError}.  

We report results from three different dgp's.  In the first two dgp's, we set $\theta_{g,j} =  c_y\beta_{0,j}$ and $\theta_{m,j} = c_d\beta_{0,j}$ with $\beta_{0,j} = (1/j)^2$ for $j = 1,...,200$.  The first dgp, which we label ``Design 1,'' uses homoscedastic innovations with $\sigma_y = \sigma_d = 1$.  The second dgp, ``Design 2,'' is heteroscedastic with
$\sigma_{d,i} = \sqrt{\frac{(1+x_i'\beta_0)^2}{\En(1+x_i'\beta_0)^2}}$ and $\sigma_{y,i} = \sqrt{\frac{(1+\alpha_0 d_i + x_i'\beta_0)^2}{\En(1+\alpha_0 d_i+x_i'\beta_0)^2}}$.  The constants $c_y$ and $c_d$ are chosen to generate desired population values for the reduced form $R^2$'s, i.e. the $R^2$'s for equations (\ref{eq: RPL1}) and (\ref{eq: RPL2}).  For each equation, we choose $c_y$ and $c_d$ to generate $R^2 = 0, .2, .4, .6,$ and $.8$.  In the heteroscedastic design, we choose $c_y$ and $c_d$ based on $R^2$ as if (\ref{ModelMCPLMy}) and (\ref{ModelMCPLMd}) held with $v_i$ and $\zeta_i$ homoscedastic and label the results by $R^2$ as in Design 1.  In the third design (``Design 3''), we use a combination of deterministic and random coefficients.  For the deterministic coefficients, we set $\theta_{g,j} = c_y(1/j)^2$ for $j \le 5$ and $\theta_{m,j} = c_d (1/j)^2$ for $j \le 5$.  We then generate the remaining coefficients as iid draws from $(\theta_{g,j},\theta_{m,j})' \sim N(0_{2 \times 1},(1/p) I_2)$.  For each equation, we choose $c_y$ and $c_d$ to generate $R^2 = 0, .2, .4, .6,$ and $.8$ in the case that all of the random coefficients were exactly equal to 0 and label the results by $R^2$ as in Design 1.
We draw new $x$'s, $\zeta$'s, and $v$'s at every simulation replication, and we also generate new $\theta$'s at every simulation replication in Design 3.

We consider Designs 1 and 2 to be baseline designs.  These designs do not have exact sparse representations but have coefficients that decay quickly so that approximately sparse representations are available. Design 3 is meant to introduce a modest deviation from the approximately sparse model towards a model
with many small, uncorrelated coefficients.  Using this we shall document that our proposed procedure still performs reasonably well, although it could be improved by incorporation of a ridge fit as one of regressors over which selection occurs.  In a working paper version of this paper \citen{BCH2011:SuppMatTE}, we present results for 26 additional designs.  The results presented in this section are sufficient to illustrate the general patterns from the larger set of results.\footnote{ In particular, the post-double-Lasso performed very well across all simulations designs where approximate sparsity provides a reasonable description of the dgp.  Unsurprisingly, the performance deteriorates as one deviates from the smooth/approximately sparse case.  However, in no design was the post-double-Lasso outperformed by other feasible procedures.  In extensive initial simulations, we also found that Square-Root Lasso and Iterated Lasso performed very similarly and thus only report Lasso results. }

We report results for five different procedures.  Two of the procedures are infeasible benchmarks: Oracle and Double-Selection Oracle estimators, which use of knowledge of the true coefficient structures $\theta_g$ and $\theta_m$ and are thus unavailable in practice.  The Oracle estimator is the ordinary least squares of $y_i - x_i'\theta_g$ on $d_i$, and the Double-Selection Oracle is the ordinary least squares of $y - x_i'\theta_g$ on $d_i - x_i'\theta_m$.   The other procedures we consider are feasible.  In all of them, we rely on Lasso and set  $\lambda$ according to the algorithm outlined in Appendix A with $1-\conflvl = .95$.  One procedure is the standard post-single selection estimator -- the  Post-Lasso -- which applies Lasso to equation (\ref{ModelMCPLMy}) without penalizing $\alpha$, the coefficient on $d$, to select additional control variables from among $x$.  Estimates of $\alpha_0$ are then obtained by OLS regression of $y$ on $d$ and the set of additional controls selected in the Lasso step and inference using the Post-Lasso estimator proceeds using conventional heteroscedasticity robust OLS inference from this regression.  Post-Double-Selection or Post-Double-Lasso is the feasible procedure advocated in this paper.  We run Lasso of $y$ on $x$ to select a set of predictors for $y$ and run Lasso of $d$ on $x$ to select a set of predictors for $d$.  $\alpha_0$ is then estimated by running OLS regression of $y$ on $d$ and the union of the sets of regressors selected in the two Lasso runs, and inference is simply the usual heteroscedasticity robust OLS inference from this regression.  Post-Double-Selection $+$ Ridge is an \textit{ad hoc} variant of Post-Double-Selection in which we add the ridge fit from equation (\ref{ModelMCPLMd}) as an additional potential regressor that may be selected by Lasso.  The ridge fit is obtained with a single ridge penalty parameter that is chosen using 10-fold cross-validation.  This procedure is motivated by a desire to add further robustness in the case that many small coefficients are suspected.  Further exploration of procedures that perform well, both theoretically and in simulations, in the presence of many small coefficients is an interesting avenue for additional research.

We start by summarizing results in Table 1 for $(R^2_y,R^2_d) = (0,.2), (0,.8), (.8,.2),$ and $(.8,.8)$ where $R^2_y$ is the population $R^2$ from regressing $y$ on $x$ (Structure $R^2$) and $R^2_d$ is the population $R^2$ from regressing $d$ on $x$ (First Stage $R^2$).  We report root-mean-square-error (RMSE) for estimating $\alpha_0$ and size of 5\% level tests (Rej. Rate).  As should be the case, the Oracle and Double-Selection Oracle, which are reported to provide the performance of an infeasible benchmark, perform well relative to the feasible procedures across the three designs.  We do see that the feasible Post-Double-Selection procedures perform similarly to the Double-Selection Oracle without relying on \textit{ex ante} knowledge of the coefficients that go in to the control functions, $\theta_g$ and $\theta_m$.  On the other hand, the Post-Lasso procedure generally does not perform as well as Post-Double-Selection and is very sensitive to the value of $R^2_d$.  While Post-Lasso performs adequately when $R^2_d$ is small, its performance deteriorates quickly as $R^2_d$ increases.  This lack of robustness of traditional variable selection methods such as Lasso which were designed with forecasting, not inference about treatment effects, in mind is the chief motivation for our advocating the Post-Double-Selection procedure when trying to infer structural or treatment parameters.

We provide further details about the performance of the feasible estimators in Figures 1, 2, and 3 which plot size of 5\% level tests, bias, and standard deviation for the Post-Lasso, Double-Selection (DS), and Double-Selection Oracle (DS Oracle) estimators of the treatment effect across the full set of $R^2$ values considered.  Figure 1, 2, and 3 respectively report the results from Design 1, 2, and 3.  The figures are plotted with the same scale to aid comparability and for readability rejection frequencies for Post-Lasso were censored at .5.  Perhaps the most striking feature of the figures is the poor performance of the Post-Lasso estimator.  The Post-Lasso estimator performs poorly in terms of size of tests across many different $R^2$ combinations and can have an order of magnitude more bias than the corresponding Post-Double-Selection estimator.  The behavior of Post-Lasso is quite non-uniform across $R^2$ combinations, and Post-Lasso does not reliably control size distortions or bias except in the case where the controls are uncorrelated with the treatment (where First-Stage $R^2$ equals 0) and thus ignorable.  In contrast, the Post-Double-Selection estimator performs relatively well across the full range of $R^2$ combinations considered. The Post-Double-Selection estimator's performance is also quite similar to that of the infeasible Double-Selection Oracle across the majority of $R^2$ values considered.  Comparing across Figures 1 and 2, we see that size distortions for both the Post-Double-Selection estimator and the Double-Selection Oracle are somewhat larger in the presence of heteroscedasticity but that the basic patterns are more-or-less the same across the two figures.  Looking at Figure 3, we also see that the addition of small independent random coefficients results in somewhat larger size distortions for the Post-Double-Selection estimator than in the other homoscedastic design, Design 1, though the procedure still performs relatively well.

In the final figure, Figure 4, we compare the performance of the Post-Double-Selection procedure to the \textit{ad hoc} Post-Double-Selection procedure which selects among the original set of variables augmented with the ridge fit obtained from equation (\ref{ModelMCPLMd}).  We see that the addition of this variable does add robustness relative to Post-Double-Selection using only the raw controls in the sense of producing tests that tend to have size closer to the nominal level.  This additional robustness is a good feature,  though it comes at the cost of increased RMSE which is especially prominent for small values of the first-stage $R^2$.

The simulation results are favorable to the Post-Double-Selection estimator.  In the simulations, we see that the Post-Double-Selection procedure provides an estimator of a treatment effect in the presence of a large number of potential confounding variables that performs similarly to the infeasible estimator that knows the values of the coefficients on all of the confounding variables.  Overall, the simulation evidence supports our theoretical results and suggests that the proposed Post-Double-Selection procedure can be a useful tool to researchers doing structural estimation in the presence of many potential confounding variables. It also shows, as a contrast,  that  the standard  Post-Single-Selection procedure provides poor inference and therefore can not be a reliable tool to these researchers.

\section{Empirical Example: Estimating the Effect of Abortion on Crime}\label{Sec:Empirical}

In the preceding sections, we have provided results demonstrating how variable selection methods, focusing on the case of Lasso-based methods, can be used to estimate treatment effects in models in which we believe the  variable of interest is exogenous conditional on observables.  We further illustrate the use of these methods in this section by reexamining Donohue III and Levitt's \citeyear{levitt:abortion} study of the impact of abortion on crime rates.  In the following, we briefly review \citen{levitt:abortion} and then present estimates obtained using the methods developed in this paper.

\citen{levitt:abortion} discuss two key arguments for a causal channel relating abortion to crime.  The first is simply that more abortion among a cohort results in an otherwise smaller cohort and so crime 15 to 25 years later, when this cohort is in the period when its members are most at risk for committing crimes, will be otherwise lower given the smaller cohort size.  The second argument is that abortion gives women more control over the timing of their fertility allowing them to more easily assure that childbirth occurs at a time when a more favorable environment is available during a child's life.  For example, access to abortion may make it easier to ensure that a child is born at a time when the family environment is stable, the mother is more well-educated, or household income is stable.  This second channel would mean that more access to abortion could lead to lower crime rates even if fertility rates remained constant.

The basic problem in estimating the causal impact of abortion on crime is that state-level abortion rates are not randomly assigned, and it seems likely that there will be factors that are associated to both abortion rates and crime rates. It is clear that any association between the current abortion rate and the current crime rate is likely to be spurious.  However, even if one looks at say the relationship between the abortion rate 18 years in the past and the crime rate among current 18 year olds, the lack of random assignment makes establishing a causal link difficult without adequate controls.  An obvious confounding factor is the existence of persistent state-to-state differences in policies, attitudes, and demographics that are likely related to the overall state level abortion and crime rates.  It is also important to control flexibly for aggregate trends.  For example, it could be the case that national crime rates were falling over this period while national abortion rates were rising but that these trends were driven by completely different factors.  Without controlling for these trends, one would mistakenly associate the reduction in crime to the increase in abortion.  In addition to these overall differences across states and times, there are other time varying characteristics such as state-level income, policing, or drug-use to name a few that could be associated with current crime and past abortion.

To address these confounds, \citen{levitt:abortion} estimate a model for state-level crime rates running from 1985 to 1997 in which they condition on a number of these factors.  Their basic specification is
\begin{align}\label{LevittModel}
y_{cit} = \alpha a_{cit} + w_{it}'\beta + \delta_i + \gamma_t + \varepsilon_{it}
\end{align}
where $i$ indexes states, $t$ indexes times, $c \in \{\textnormal{violent, property, murder}\}$ indexes type of crime, $\delta_i$ are state-specific effects that control for any time-invariant state-specific characteristics, $\gamma_t$ are time-specific effects that control flexibly for any aggregate trends, $w_{it}$ are a set of control variables to control for time-varying confounding state-level factors, $a_{cit}$ is a measure of the abortion rate relevant for type of crime $c$,\footnote{This variable is constructed as weighted average of abortion rates where weights are determined by the fraction of the type of crime committed by various age groups.  For example, if 60\% of violent crime were committed by 18 year olds and 40\% were committed by 19 year olds in state $i$, the abortion rate for violent crime at time $t$ in state $i$ would be constructed as .6 times the abortion rate in state $i$ at time $t-18$ plus .4 times the abortion rate in state $i$ at time $t-19$.  See \citen{levitt:abortion} for further detail and exact construction methods.} and $y_{cit}$ is the crime-rate for crime type $c$.  \citen{levitt:abortion} use the log of lagged prisoners per capita, the log of lagged police per capita, the unemployment rate, per-capita income, the poverty rate, AFDC generosity at time $t - 15$, a dummy for concealed weapons law, and beer consumption per capita for $w_{it}$, the set of time-varying state-specific controls.  Tables IV and V in \citen{levitt:abortion} present baseline estimation results based on (\ref{LevittModel}) as well as results from different models which vary the sample and set of controls to show that the baseline estimates are robust to small deviations from (\ref{LevittModel}).  We refer the reader to the original paper for additional details, data definitions, and institutional background.

For our analysis, we take the argument that the abortion rates defined above may be taken as exogenous relative to crime rates once observables have been conditioned on from \citen{levitt:abortion} as given.  Given the seemingly obvious importance of controlling for state and time effects, we account for these in all models we estimate.  We choose to eliminate the state effects via differencing rather than including a full set of state dummies but include a full set of time dummies in every model.  Thus, we will estimate models of the form
\begin{align}\label{LassoLevittModel}
y_{cit}-y_{cit-1} = \alpha (a_{cit}-a_{cit-1}) + z_{it}'\kappa + \gamma_t + \eta_{it}.
\end{align}
We use the same state-level data as \citen{levitt:abortion} but delete Alaska, Hawaii, and Washington, D.C. which gives a sample with 48 cross-sectional observations and 12 time series observations for a total of 576 observations.  With these deletions, our baseline estimates using the same controls as in (\ref{LevittModel}) are quite similar to those reported in \citen{levitt:abortion}.  Baseline estimates from Table IV of \citen{levitt:abortion} and our baseline estimates based on the differenced version of (\ref{LevittModel}) are given in the first and second row of Table 2 respectively.

Our main point of departure from \citen{levitt:abortion} is that we allow for a much richer set $z_{it}$ than allowed for in $w_{it}$ in model (\ref{LevittModel}).  Our $z_{it}$ includes higher-order terms and interactions of the control variables defined above.  In addition, we put initial conditions and initial differences of $w_{it}$ and $a_{it}$ into our vector of controls $z_{it}$.  This addition allows for the possibility that there may be some feature of a state that is associated both with its growth rate in abortion and its growth rate in crime.  For example, having an initially high-levels of abortion could be associated with having high-growth rates in abortion and low growth rates in crime.  Failure to control for this factor could then lead to misattributing the effect of this initial factor, perhaps driven by policy or state-level demographics, to the effect of abortion.  Finally, we allow for more general trends by allowing for an aggregate quadratic trend in $z_{it}$ as well as interactions of this quadratic trend with control variables.
This gives us a set of 251 control variables to select among in addition to the 12 time effects that we include in every model.\footnote{The exact identities of the 251 potential controls is available upon request.  It consists of linear and quadratic terms of each continuous variable in $w_{it}$, interactions of every variable in $w_{it}$, initial levels and initial differences of $w_{it}$ and $a_{it}$, and interactions of these variables with a quadratic trend.}

Note that interpreting estimates of the effect of abortion from model (\ref{LevittModel}) as causal relies on the belief that there are no higher-order terms of the control variables, no interaction terms, and no additional excluded variables that are associated both to crime rates and the associated abortion rate.  Thus, controlling for a large set of variables as described above is desirable from the standpoint of making this belief more plausible.  At the same time, naively controlling lessens our ability to identify the effect of interest and thus tends to make estimates far less precise.  The effect of estimating the abortion effect conditioning on the full set of 251 potential controls described above is given in the third row of Table 2.  As expected, all coefficients are estimated very imprecisely.  Of course, very few researchers would consider using 251 controls with only 576 observations due to exactly this issue.

We are faced with a tradeoff between controlling for very few variables which may leave us wondering whether we have included sufficient controls for the exogeneity of the treatment and controlling for so many variables that we are essentially mechanically unable to learn about the effect of the treatment.  The variable selection methods developed in this paper offer one resolution to this tension.  The assumed sparse structure maintains that there is a small enough set of variables that one could potentially learn about the treatment but adds substantial flexibility to the usual case where a researcher considers only a few control variables by allowing this set to be found by the data from among a large set of controls.  Thus, the approach should complement the usual careful specification analysis by providing a researcher an efficient, data-driven way to search for a small set of influential confounds from among a sensibly chosen broad set of potential confounding variables.

In the abortion example, we use the post-double-selection estimator defined in Section \ref{Sec:DoubleSelection} for each of our dependent variables.  For violent crime, ten variables are selected in the abortion equation,\footnote{The selected variables are AFDC generosity squared, beer consumption squared, the initial poverty change, initial income, initial income squared, the initial change in prisoners per capita squared interacted with the trend, initial income interacted with the trend, the initial change in the abortion rate, the initial change in the abortion rate interacted with the trend, and the initial level of the abortion rate.} and one is selected in the crime equation.\footnote{The initial level of the abortion rate interacted with time is selected.}
For property crime, eight variables are selected in the abortion equation,\footnote{The selected variables are income, the initial poverty change, the initial change in prisoners per capita squared, the initial level of prisoners per capita, initial income, the initial change in the abortion rate, the initial change in the abortion rate interacted with the trend, and the initial level of the abortion rate.} and six are selected in the crime equation.\footnote{The six variables are the initial level of AFDF generosity, the initial level of income interacted with the trend and the trend squared, the initial level of income squared interacted with the trend and the trend squared, and the initial level of the abortion rate interacted with the trend.}
For murder, eight variables are selected in the abortion equation,\footnote{The selected variables are AFDC generosity, beer consumption squared, the change in beer consumption squared, the change in beer consumption squared times the trend and the trend squared, initial income times the trend, the initial change in the abortion rate interacted with the trend, and the initial level of the abortion rate.} and none were selected in the crime equation.

Estimates of the causal effect of abortion on crime obtained by searching for confounding factors among our set of 251 potential controls are given in the fourth row of Table 2.  Each of these estimates is obtained from the least squares regression of the crime rate on the abortion rate and the 11, 14, and eight controls selected by the double-post-Lasso procedure for violent crime, property crime, and murder respectively.  The estimates for the effect of abortion on violent crime and the effect of abortion on murder are quite imprecise, producing 95\% confidence intervals that encompass large positive and negative values.  The estimated effect for property crime is roughly in line with the previous estimates though it is no longer significant at the 5\% level but is significant at the 10\% level.  Note that the double-post-Lasso produces models that are not of vastly different size than the ``intuitive'' model (\ref{LevittModel}).  As a final check, we also report results that include all of the original variables from (\ref{LevittModel}) in the amelioration set in the fifth row of the table.  These results show that the conclusions made from using only the variable selection procedure do not qualitatively change when the variables used in the original \citen{levitt:abortion} are added to the equation.  For a quick benchmark relative to the simulation examples, we note that the $R^2$ obtained by regressing the crime rate on the selected variables are .0395,  .1185, and  .0044 for violent crime, property crime, and the murder rate respectively and that the $R^2$'s from regressing the abortion rate on the selected variables are .9447, .9013, and  .9144 for violent crime, property crime, and the murder rate respectively.  These values correspond to regions of the $R^2$ space considered in the simulation where the double selection procedure substantially outperformed simple Lasso procedures.

It is very interesting that one would draw qualitatively different conclusions from the estimates obtained using formal variable selection than from the estimates obtained using a small set of intuitively selected controls.  Looking at the set of selected control variables, we see that initial conditions and interactions with trends are selected across all dependent variables.  The selection of this set of variables suggests that there are initial factors which are associated with the change in the abortion rate.  We also see that we cannot precisely determine the effect of the abortion rate on crime rates once one accounts for initial conditions.  Of course, this does not mean that the effects of the abortion rate provided in the first two rows of Table 2 are not representative of the true causal effects.  It does, however, imply that this conclusion is strongly predicated on the belief that there are not other unobserved state-level factors that are correlated to both initial values of the controls and abortion rates, abortion rate changes, and crime rate changes.  Interestingly, a similar conclusion is given in \citen{FooteGoetzAbortion} based on an intuitive argument.

We believe that the example in this section illustrates how one may use modern variable selection techniques to complement causal analysis in economics.  In the abortion example, we are able to search among a large set of controls and transformations of variables when trying to estimate the effect of abortion on crime.  Considering a large set of controls makes the underlying assumption of exogeneity of the abortion rate conditional on observables more plausible, while the methods we develop allow us to produce an end-model which is of manageable dimension.  Interestingly, we see that one would draw quite different conclusions from the estimates obtained using formal variable selection.  Looking at the variables selected, we can also see that this change in interpretation is being driven by the variable selection method's selecting different variables, specifically initial values of the abortion rate and controls, than are usually considered.  Thus, it appears that the usual interpretation hinges on the prior belief that initial values should be excluded from the structural equation.

\section{Conclusion}

In this paper, we consider estimation of treatment effects or structural parameters in an environment where the treatment is believed to be exogenous conditional on observables.  We do not impose the conventional assumption that the identities of the relevant conditioning variables and the functional form with which they enter the model are known.  Rather, we assume that the researcher believes there is a relatively small number of important factors whose identities are unknown within a much larger known set of potential variables and transformations.  This sparsity assumption allows the researcher to estimate the desired treatment effect and infer a set of important variables upon which one needs to condition by using modern variable selection techniques without \textit{ex ante} knowledge of which are the important conditioning variables.  Since naive application of variable selection methods in this context may result in very poor properties for inferring the treatment effect of interest, we propose a ``double-selection'' estimator of the treatment effect, provide a formal demonstration of its properties for estimating the treatment effect, and provide its approximate distribution under technical regularity conditions and the assumed sparsity in the model.

In addition to the theoretical development, we illustrate the potential usefulness of our proposal through a number of simulation studies and an empirical example.   In Monte Carlo simulations, our procedure outperforms simple variable selection strategies for estimating the treatment effect across the designs considered and does relatively well compared to an infeasible estimator that uses the identities of the relevant conditioning variables.  We then apply our estimator to attempt to estimate the causal impact of abortion on crime following \citen{levitt:abortion}.  We find that our procedure selects a small number of conditioning variables.  After conditioning on these selected variables, one would draw qualitatively different inference about the effect of abortion on crime than would be drawn if one assumed that the correct set of conditioning variables was known and the same as those variables used in \citen{levitt:abortion}.   Taken together, the empirical and simulation examples demonstrate that the proposed method may provide a useful complement to other sorts of specification analysis done in applied research.

\appendix

\section{Iterated Estimation of Penalty Loadings}\label{Sec:EstSigma}

In the case of Lasso under heteroscedasticity, we must specify for the penalty  loadings (\ref{Def:LambdaLASSOboound}). Here we state
algorithms for estimating these loadings.

Let  $I_0$ be an initial set of regressors with bounded
 number of elements, including for example intercept.  Let $\bar\beta(I_0)$ be the least squares estimator of the coefficients on the covariates associated with $I_0$, and define $ \hat l_{j0} := \sqrt{\En[ x_{ij}^2(y_i-x_i'\bar\beta(I_0))^2]}.$


An algorithm for estimating the penalty loadings using Post-Lasso is as follows:
\begin{algorithm}[Estimation of Lasso loadings using Post-Lasso iterations] Set $\hat l_{j,0} := \hat l_{jI_{0}}$, $j=1,\ldots,p$.   Set $k = 0$, and specify a small constant $\nu \geq 0$ as a tolerance level and
a constant $K>1$ as an upper bound on the number of iterations. (1)  Compute the Post-Lasso estimator $\widetilde \beta$ based on the loadings $\hat l_{j,k}$. (2) For $\widehat s = \|\widetilde \beta\|_0 = |\widehat T|$ set  $l_{j,k+1} := \sqrt{\En[ x_{ij}^2(y_i-x_i'\widehat \beta)^2]}\sqrt{n/(n-\widehat s)}.$ (3) If $\max_{1\leq j\leq p}| \hat l_{j,k} - \hat l_{j,k+1}| \leqslant \nu$ or $k> K$, set the loadings to $\hat l_{j,k+1}$, $j=1,\ldots,p$ and stop; otherwise, set $k \leftarrow k+1 $ and go to (1).
\end{algorithm}

A similar algorithm can be defined for using with Post-Square-root Lasso instead of Post-Lasso.\footnote{The algorithms can also be modified in the obvious manner for Lasso or Square-root Lasso.}



\begin{algorithm}[Estimation of Square-root Lasso loadings using Post-Square-root Lasso iterations] Set $k = 0$, and specify a small constant $\nu \geq 0$ as a tolerance level and
a constant $K>1$ as an upper bound on the number of iterations.   (1)  Compute the Post-Square-root Lasso estimator $\widetilde \beta$ based on the loadings $\hat l_{j,k}$.
(2) Set $\hat l_{j,k+1} := \sqrt{\En[ x_{ij}^2(y_i-x_i'\widetilde \beta)^2]}/\sqrt{\En[(y_i-x_i'\widetilde \beta)^2]}.$ (3) If $\max_{1\leq j\leq p}|\hat l_{j,k} - \hat l_{j,k+1}| \leqslant \nu$ or $k> K$, set the loadings to $\hat l_{j,k+1}$, $j=1,\ldots,p$, and stop;
otherwise set $k \leftarrow k+1 $ and go to (1).
\end{algorithm}

\section{Proof of Theorem \ref{theorem:inference}}

The proof proceeds under given sequence of probability
measures $\{\Pr_n\}$, as $n \to \infty$.

Let
$ Y =[y_{1},...,y_{n}]'$, $X= [x_{1},...,x_{n}]'$, $D=[d_1,...,d_n]'$, $V= [v_1,...,v_n]'$,
$\zeta =[\zeta_1,...,\zeta_n]'$, $m = [m_1,...,m_n]'$, $R_m = [r_{m1},...,r_{mn}]'$,
$g = [g_1,...,g_n]'$, $R_g = [r_{g1},...,r_{gn}]'$,
and so on. For $A \subset \{1,...,p\}$,
let $X[A] = \{X_j, j \in A\}$,
where $\{X_j, j=1,...,p\}$ are the columns of $X$.  Let
$$\mathcal{P}_A  = X[A](X[A]'X[A])^{-}X[A]'$$ be the projection operator
sending vectors in $\Bbb{R}^n$ onto ${\rm span}[X[A]]$, and let $\mathcal{M}_A = {\rm I}_n - \mathcal{P}_A$ be the
projection onto the subspace that is orthogonal to  ${\rm span}[X[A]]$.  For a vector $Z \in \Bbb{R}^n$, let
$$
\tilde \beta_Z(A) := \arg\min_{b \in \Bbb{R}^p} \|Z- X'b\|^2: \ b_j = 0, \ \forall j \not \in A,
$$
be the coefficient of linear projection of $Z$ onto ${\rm span}[X[A]]$.   If $A = \varnothing$,
interpret $\mathcal{P}_A = 0_n$, and $\tilde \beta_Z = 0_p$.

Finally, denote $\semin{m} = \semin{m}[\En[x_ix_i']]$ and  $\semax{m} = \semax{m}[\En[x_ix_i']]$.

Step 1.(Main) Write
$
\check \alpha = \[D'\MX D/n\]^{-1}[D'\MX Y/n]
$
so that
$$
\sqrt{n}(\check \alpha - \alpha_0) = \[D'\MX D/n\]^{-1}[D'\MX (g + \zeta)/\sqrt{n}] =: ii^{-1} \cdot i.
$$
By Steps 2 and 3, $$ii = V'V/n + o_P(1) \text{ and } i = V'\zeta/\sqrt{n} + o_P(1).$$ Next note that
$V'V/n = \Ep[V'V/n] + o_P(1)$ by Chebyshev, and because $\Ep[V'V/n]$ is bounded away from zero and from above uniformly in $n$ by Condition SM, we have $ii^{-1} = \Ep[V'V/n]^{-1} + o_P(1)$.


By Condition SM $\sigma_n^2 = \barEp[v_i^2]^{-1}\barEp[\zeta_i^2v_i^2] \barEp[v_i^2]^{-1}$ is bounded away from zero
and from above, uniformly in $n$.
Hence
$$
Z_n = \sigma_n^{-1} \sqrt{n}(\check \alpha - \alpha_0)  = n^{-1/2} \sum_{i=1}^n  z_{i,n} + o_{P}(1),$$
\noindent where $z_{i,n} := \sigma_n^{-1}v_i\zeta_i $ are i.n.i.d. with mean zero. For $\delta>0$ such that $4 + 2 \delta \leq q$
$$
\barEp|z_{i,n}|^{2+\delta} \lesssim \barEp\left[ | v_i|^{2+\delta}|\zeta_i |^{2+\delta}  \right] \lesssim  \sqrt{\barEp |v_i|^{4+2\delta}} \sqrt{\barEp |\zeta_i|^{4 + 2\delta}} \lesssim 1,
$$
by Condition SM. This condition verifies the Lyapunov condition and thus application of the Lyapunov CLT for i.n.i.d. triangular arrays implies that $$Z_n \rightsquigarrow N(0,1).$$

Step 2. (Behavior of $i$.) Decompose, using $ D= m+ V$,
\begin{eqnarray*}
i = V'\zeta/\sqrt{n} + \underset{=:i_{a}}{m' \MX g/\sqrt{n}} + \underset{=:i_{b}}{m'\MX \zeta/\sqrt{n}} + \underset{=:i_{c}}{V'\MX g/\sqrt{n}} - \underset{=:i_{d}}{ V'\PX \zeta/\sqrt{n}}.
\end{eqnarray*}

First, by Step 5 and 6 below we have
$$
|i_{a}| = |m'\MX g/\sqrt{n}| \leq \sqrt{n}\|\MX g/\sqrt{n}\|\|\MX m/\sqrt{n}\| \lesssim_P \sqrt{ [s \log (p\vee n)]^2/n } = o(1),
$$ where the last bound follows from the assumed growth condition $s^2\log^2(p\vee n) = o(n)$.

Second, using that $m = X\beta_{m0} + R_m$ and $m'\MX\zeta = R_m'\zeta-(\tilde \beta_m(\hat I) - \beta_{m0})'X'\zeta$ , conclude
$$
|i_b| \leq |R_m'\zeta/\sqrt{n}| + |(\tilde \beta_m(\hat I) - \beta_{m0})'X'\zeta/\sqrt{n}| \lesssim_P \sqrt{ [s \log (p\vee n)]^2/n } = o_P(1).
$$
This follows since $$
|R_m'\zeta/\sqrt{n}| \lesssim_P  \sqrt{R_m'R_m/n} \lesssim_P \sqrt{s/n},
$$
holding by Chebyshev inequality  and Conditions SM and ASTE(iii), and
$$
|(\tilde \beta_m(\hat I) - \beta_{m0})'X'\zeta/\sqrt{n}| \leq \|\tilde \beta_m (\hat I) - \beta_{m0}\|_1 \|X'\zeta/\sqrt{n}\|_{\infty} \lesssim_P \sqrt{ [s^2 \log (p\vee n)]/n} \sqrt{\log (p\vee n)}.
$$
The latter bound follows by (a)  $$\|\tilde \beta_m (\hat I) - \beta_{m0}\|_1
\leq   \sqrt{\hat s + s} \|\tilde \beta_m (\hat I) - \beta_{m0}\| \lesssim_P
\sqrt{ [s^2 \log (p\vee n)]/n}$$ holding by Step 5 and by $\hat s \lesssim_P s$ implied by Lemma \ref{corollary3:postrate},  and (b)  by $$\|X'\zeta/\sqrt{n}\|_{\infty}\lesssim_P \sqrt{ \log (p\vee n) }$$
holding by Step 4 under Condition SM.

Third, using similar reasoning, decomposition $g = X\beta_{g0} + R_g$, and Steps 4 and 6, conclude
$$
|i_c| \leq |R_g'V/\sqrt{n}| + |(\tilde \beta_g(\hat I) - \beta_{g0})'X'V/\sqrt{n}| \lesssim_P   \sqrt{ [s \log (p\vee n)]^2/n } = o_P(1).
$$

Fourth, we have
$$
|i_d| \leq |\tilde \beta_V(\hat I)' X'\zeta/\sqrt{n}| \leq \| \tilde \beta_V(\hat I)\|_1 \|X'\zeta/\sqrt{n}\|_{\infty} \lesssim_P   \sqrt{ [s \log (p\vee n)]^2/n } = o_P(1),
$$
since by Step 4 below $\|X'\zeta/\sqrt{n}\|_{\infty}\lesssim_P \sqrt{ \log (p\vee n) }$, and
\begin{eqnarray*}
 \| \tilde \beta_V(\hat I)\|_1 & \leq &  \sqrt{\hat s} \|\tilde \beta_V(\hat I)\| \leq \sqrt{\hat s} \| (X[\hat I]'X[\hat I]/n)^{-1} X[\hat I]'V/n\| \\
& \leq &  \sqrt{\hat s} \phi_{\min}^{-1}(\hat s) \sqrt{\hat s}  \| X'V/\sqrt{n}\|_{\infty}/\sqrt{n} \lesssim_P s \sqrt{[\log (p\vee n)]/n}.
 \end{eqnarray*}
The latter bound follows  from $\hat s \lesssim_P s$, holding by Lemma \ref{corollary3:postrate}, so that  $ \phi^{-1}_{\min}(\hat s) \lesssim_P 1$ by Condition SE, and from $\|X'V/\sqrt{n}\|_{\infty}\lesssim_P \sqrt{ \log (p\vee n) }$ holding by Step 4.

Step 3. (Behavior of $ii$.) Decompose
$$
ii = (m+V)'\MX (m+V)/n = V'V/n + \underset{=:ii_a}{m' \MX m/n} +
\underset{=:ii_b}{ 2 m'\MX V/n } - \underset{=:ii_c}{V'\PX V/n}.
$$
Then $|ii_a| \lesssim_P [s \log (p\vee n)]/n= o_P(1)$ by Step 5,
$|ii_b| \lesssim_P [s \log (p\vee n)]/n= o_P(1)$ by reasoning similar to deriving
the bound for $|i_b|$, and $|ii_c| \lesssim_P [s \log (p\vee n)]/n= o_P(1)$
by reasoning similar to deriving the bound for $|i_d|$.

Step 4. (Auxiliary: Bounds on $\|X'\zeta/\sqrt{n}\|_{\infty}$ and $\|X'V/\sqrt{n}\|_{\infty}$) Here we show that
$$\text{(a)} \ \|X'\zeta/\sqrt{n}\|_{\infty}\lesssim_P \sqrt{ \log (p\vee n) } \text{ \  and \ }  \text{(b)} \|X'V/\sqrt{n}\|_{\infty}\lesssim_P \sqrt{ \log (p\vee n) }.$$
To show (a), we use Lemma \ref{Lemma:SNMD} stated in Appendix F on the tail bound for self-normalized deviations to deduce the bound. Indeed, we have that wp $\to 1$
for some $\ell_n \to \infty$ but so slowly that $1/\gamma = \ell_n \lesssim \log n$, with probability $1-o(1)$
\begin{equation}\label{eq: bound sup}
\max_{1 \leq j \leq p}  \left |\frac{ n^{-1/2}  \sum_{i=1}^n x_{ij}\zeta_i}
 {\sqrt{\En [x_{ij}^2 \zeta_i^2]}}  \right |  \leq \Phi^{-1}\left (1- \frac{1}{2\ell_np} \right)\lesssim \sqrt{ 2 \log (2\ell_np) }\lesssim \sqrt{ \log (p\vee n) }.
\end{equation}
By  Lemma \ref{Lemma:SNMD}  the first inequality in (\ref{eq: bound sup})  holds, provided that for all $n$ sufficiently large the following holds,
$$
 \Phi^{-1}\left (1- \frac{1}{2\ell_np} \right) \leq \frac{n^{1/6}}{\ell_n} \min_{1\leq j \leq p} M^2_{j}-1, \ \  M_j := \frac{\barEp[x_{ij}^2 \zeta_i^2]^{1/2}}{\barEp[|x_{ij}^3| |\zeta_i^3|]^{1/3}}.
  $$
 Since we can choose $\ell_n$ to grow as slowly as needed,
a sufficient condition for this are the conditions:
$$\log p = o(n^{1/3}) \text{ and }  \min_{1 \leq j \leq p}  M_j \gtrsim 1,$$
which both hold by Condition SM. Finally,
\begin{equation}\label{eq:bound sup 2}
\max_{1 \leq j \leq p} \En [x_{ij}^2 \zeta_i^2]  \lesssim_{\mathrm{P}} 1,
\end{equation}
by Condition SM.   Therefore (a) follows from the bounds (\ref{eq: bound sup}) and (\ref{eq:bound sup 2}).   Claim (b) follows similarly.

Step 5. (Auxiliary: Bound on $\|\MX m\|$ and related quantities.) This step shows that
$$
\text{(a)}   \  \| \MX m/\sqrt{n}\|  \lesssim_{P} \sqrt{ [s \log (p\vee n)]/n} \text{ and }  \text{(b)} \  \| \tilde \beta_m(\hat I) - \beta_{m0}\|  \lesssim_{P} \sqrt{ [s \log (p\vee n)]/n}.
$$
Observe that
\begin{eqnarray*}
\sqrt{ [s \log (p\vee n)]/n}  \underset{(1)}{\gtrsim_P} \| \MXd m/\sqrt{n}\|    \underset{(2)}{\gtrsim_P}    \| \MX m/\sqrt{n}\| \\
 \end{eqnarray*}
where inequality (1) holds  since  by Lemma 1 $\| \MXd m/\sqrt{n}\| \leq \|(X\tilde \beta_D(\hat I_1) - m)/\sqrt{n}\|
\lesssim_P \sqrt{ [s \log (p\vee n)]/n}$,  and (2) holds by $\hat I_1 \subseteq \hat I$ by construction. This shows claim (a).
To show claim (b) note that
$$
  \| \MX m/\sqrt{n}\|   \underset{(3)}{\gtrsim_P}   |\| X(\tilde \beta_m(\hat I) - \beta_{m0})/\sqrt{n}\| - \| R_m/\sqrt{n}\||
$$
where (3)  holds by the triangle
inequality.  Since $\|R_m/\sqrt{n}\| \lesssim_P \sqrt{s/n}$ by Chebyshev and Condition ASTE(iii), conclude
that 
\begin{eqnarray*}
\sqrt{ [s\log (p\vee n)]/n }  & \gtrsim_P&  \| X(\tilde \beta_m(\hat I) -
\beta_{m0})/\sqrt{n}\|  \\
 &\geq &  \sqrt{\semin{\hat s+s} } \| \tilde \beta_m(\hat I) - \beta_{m0}\|
\gtrsim_P \| \tilde \beta_m(\hat I) - \beta_{m0}\|,
 \end{eqnarray*}
since $\hat s \lesssim_P s$ by Lemma \ref{corollary3:postrate} so that $1/\semin{\hat s+s} \lesssim_P 1$
by condition SE.  This shows claim (b).

Step 6. (Auxiliary: Bound on $\|\MX g\|$ and related quantities.)
This step shows that
$$
\text{(a)}   \  \| \MX g/\sqrt{n}\|  \lesssim_{P} \sqrt{ [s \log (p\vee n)]/n} \text{ and }  \text{(b)} \  \| \tilde \beta_g(\hat I) - \beta_{g0}\|  \lesssim_{P} \sqrt{ [s \log (p\vee n)]/n}.
$$
Observe that
\begin{eqnarray*}
\sqrt{ [s \log (p\vee n)]/n} & \underset{(1)}{\gtrsim_P} &  \| \MXy (\alpha_0 m + g)/\sqrt{n}\| \\
  & \underset{(2)}{\gtrsim_P}&   \| \MX (\alpha_0 m + g)/\sqrt{n}\| \\
 &  \underset{(3)}{\gtrsim_P} & |\|  \MX g /\sqrt{n}\|- \|\MX \alpha_0 m /\sqrt{n}\||
\end{eqnarray*}
where inequality (1) holds  since by Lemma \ref{corollary3:postrate} $\| \MXy (\alpha_0 m + g)/\sqrt{n}\| \leq \|(X \tilde \beta_{Y_1}(\hat I_2) - \alpha_0 m - g)/\sqrt{n}\|
\lesssim_P \sqrt{ [s \log (p\vee n)]/n}$, (2) holds by $\hat I_2 \subseteq \hat I$, and (3) by the triangle
inequality.  Since $\|\alpha_0\|$ is bounded uniformly in $n$ by assumption,  by Step 5,
$\|\MX \alpha_0 m /\sqrt{n}\| \lesssim_P  \sqrt{ [s \log (p\vee n)]/n}$. Hence claim (a) follows by the triangle inequality:
$$
\sqrt{ [s \log (p\vee n)]/n}  \gtrsim_P \| \MX g /\sqrt{n}\|
$$
To show claim (b) we note that
$$
\| \MX g /\sqrt{n}\| \geq  |\| X(\tilde \beta_g(\hat I) - \beta_{g0})/\sqrt{n}\| - \| R_g/\sqrt{n}\||
$$
where $\| R_g/\sqrt{n}\| \lesssim_P \sqrt{s/n}$ by Condition ASTE(iii). Then conclude similarly to Step 5 that 
 \begin{eqnarray*}
\sqrt{ [s\log (p\vee n)]/n } & \gtrsim_P&  \| X(\tilde \beta_g(\hat I) -
\beta_{g0})/\sqrt{n}\|  \\
&\geq &  \sqrt{\semin{\hat s+s}} \| \tilde \beta_g(\hat I) - \beta_{g0}\|
\gtrsim_P \| \tilde \beta_g(\hat I) - \beta_{g0}\|.
 \end{eqnarray*}

Step 7. (Variance Estimation.)
Since $\hat s \lesssim_P s = o(n)$, $(n-\hat s - 1)/n = o_P(1)$,
and since $\barEp[ v_i^2 \zeta_i^2 ] $ and $\barEp[v_i^2]$
are bounded away from zero and from above uniformly in $n$ by
Condition SM, it suffices to show that
$$\En[ \hat v_i^2\hat \zeta_i^2] - \barEp[ v_i^2 \zeta_i^2 ] \to_P 0,  \  \    \En[ \hat v_i^2 ] - \barEp[v_i^2] \to_P 0,$$
The second relation was shown in Step 3, so it remains to show
the first relation.

Let $\tilde v_i = v_i + r_{mi}$ and $\tilde \zeta_i = \zeta_i + r_{gi}$. Recall that by Condition ASTE(v) we have $\barEp[ \tilde v_i^2\tilde \zeta_i^2 ] - \barEp[ v_i^2\zeta_i^2 ] \to 0$, and $\En[ \tilde v_i^2\tilde \zeta_i^2 ]  - \barEp[ \tilde v_i^2\tilde \zeta_i^2 ] \to_P 0$ by Vonbahr-Esseen's inequality in \citen{vonbahr:esseen} since $\barEp[ |\tilde v_i\tilde \zeta_i|^{2+\delta}]\leq (\barEp[ |\tilde v_i|^{4+2\delta}]\barEp[ |\tilde \zeta_i|^{4+2\delta}])^{1/2}$ is uniformly bounded for $4+2\delta\leq q$. Thus it suffices to show that $\En[ \hat v_i^2\hat \zeta_i^2] - \En[ \tilde v_i^2\tilde \zeta_i^2 ] \to_P 0$.

By the triangular inequality
$$| \En[ \hat v_i^2\hat \zeta_i^2 - \tilde v_i^2\tilde \zeta_i^2 ] | \leq \underset{=:iv } {| \En[ (\hat v_i^2 - \tilde v_i^2)\tilde \zeta_i^2 ] |} + \underset{=:iii }{| \En[  \hat v_i^2(\hat \zeta_i^2 - \tilde \zeta_i^2) ]|}.$$
Then, expanding $\hat \zeta_i^2 - \tilde \zeta_i^2 $ we have 
$$\begin{array}{rl}
iii & \leq 2\En[\{d_i(\alpha_0-\check\alpha)\}^2 \hat v_i^2]+2\En[\{x_i'(\check\beta - \beta_{g0})\}^2 \hat v_i^2]\\
& +|2\En[\tilde\zeta_id_i(\alpha_0-\check\alpha) \hat v_i^2]|+|2\En[\tilde\zeta_ix_i'(\check\beta - \beta_{g0})\hat v_i^2]|  \\ & =:  iii_a + iii_b + iii_c + iii_d =  o_P(1)\end{array}$$
where the last bound follows by the relations derived below.

First, we note
\begin{eqnarray}
 iii_a & \leq & 2 \max_{i\leq n} d_i^2|\alpha_0-\check\alpha|^2\En[\hat v_i^2] \lesssim_P  n^{(2/q) - 1} = o(1) \label{rel1} \\
iii_c & \leq &  2 \max_{i\leq n}\{|\tilde \zeta_i| |d_i|\} \En[\hat v_i^2] |\alpha_0-\check\alpha| \lesssim_P  n^{(2/q) - (1/2)}  = o(1) \label{rel2}
\end{eqnarray}
which holds by the following argument. Condition SM assumes that $\Ep[|d_i|^q]$ which in turn implies that $\Ep[\max_{i\leq n}d_i^2] \lesssim n^{2/q}$. Similarly Condition ASTE implies that $\Ep[\max_{i\leq n}\tilde \zeta_i^2]\lesssim n^{2/q}$ and $\Ep[\max_{i\leq n}\tilde v_i^2] \lesssim n^{2/q}$. Thus by Markov inequality
\begin{equation}\label{bound max}
\max_{i\leq n} |d_i| + |\tilde \zeta_i| + |\tilde v_i|\lesssim_P n^{1/q}.
 \end{equation}
Moreover,  $\En[\hat v_i^2] \lesssim_P 1$ and $|\check\alpha - \alpha_0|\lesssim_P n^{-1/2}$ by the previous steps. These bounds and $q>4$ imposed in Condition SM imply (\ref{rel1})-(\ref{rel2}).

Next we bound,
\begin{eqnarray}
iii_d & \leq & 2 \max_{i\leq n} |\tilde \zeta_i|  \max_{i \leq n}
|x_i'(\check\beta - \beta_{g0})| \En[\hat v_i^2]  \nonumber \\
  & \lesssim_P &
 n^{1/q}  \max_{i\leq n} \| x_{i}\|_\infty \sqrt{\frac{s}{\sqrt{n}} \frac{s\log (p\vee n)}{\sqrt{n}}} = o_P(1),  \label{rel4}
\end{eqnarray}
using (\ref{bound max}) and  that for  $\widehat T_g = \supp(\beta_{g0})\cup \widehat I$, we have
$$
 \max_{i\leq n}\{x_i'(\check\beta - \beta_{g0})\}^2 \leq \max_{i\leq n} \| x_{i \widehat T_g}\|^2  \| \check\beta - \beta_{g0}\|^2,
$$
where  $$ \max_{i\leq n} \| x_{i\widehat T_g}\|^2 \leq |\widehat T_g|\max_{i\leq n} \| x_{i}\|^2_\infty \lesssim_P s \max_{i\leq n} \| x_{i}\|^2_\infty$$ by the sparsity assumption in ASTE and the sparsity bound in Lemma \ref{corollary3:postrate}, and
since $\check\beta[\hat I] = (X[\hat I]'X[\hat I])^-X[\hat I]'(\zeta + g -(\check \alpha - \alpha_0)D)$ we have
$$  \|\check\beta - \beta_{g0}\| \leq \|\tilde \beta_g(\hat I) - \beta_{g0}\| + \|\tilde \beta_\zeta(\hat I)\| + |\check \alpha - \alpha_0|\cdot\|\tilde \beta_D(\hat I)\| \lesssim_P \sqrt{ s \log (p \vee n)/ n}$$
by Step 6(b),  by $$\|\tilde \beta_\zeta(\hat I)\| \leq \sqrt{\hat s} \phi_{\min}^{-1}(\hat s)\|X'\zeta/n\|_\infty\lesssim_P \sqrt{s \log (p\vee n)/n}$$ holding by Condition SE and by $\hat s \lesssim_P s$ from Lemma \ref{corollary3:postrate}, and by Step 4, $|\check \alpha - \alpha_0|\lesssim_P 1/\sqrt{n}$ by Step 1, and $$\|\tilde \beta_D(\hat I)\|\leq \phi_{\min}^{-1}(\hat s) \sqrt{\hat s} \max_{1\leq j \leq p}|\En[x_{ij}d_i]| \leq
\phi_{\min}^{-1}(\hat s) \sqrt{\hat s} \max_{1\leq j \leq p}\sqrt{\En[x_{ij}^2d_i^2]} \lesssim_P \sqrt{s}$$
by Condition SE, $\hat s \lesssim_P s$ by the sparsity bound in Lemma \ref{corollary3:postrate}, and Condition SM.

The final conclusion in (\ref{rel4}) then follows by condition ASTE (iv) and (v).

Next, using the relations above and condition ASTE (iv) and (v), we also conclude that
\begin{eqnarray}
  iii_b & \leq &  2  \max_{i\leq n}\{x_i'(\check\beta - \beta_{g0})\}^2\En[\hat v_i^2]  \nonumber \\
 & \lesssim_P &  \max_{i\leq n} \| x_{i}\|^2_\infty \frac{s}{\sqrt{n}} \frac{s\log (p\vee n)}{\sqrt{n}} = o_P(1) .  \label{rel3}
\end{eqnarray}

Finally, the argument for $iv=o_P(1)$ follows similarly to
the argument for $iii=o_P(1)$ and the result follows. \qed

\section{Proof of Corollary 1}
Let $\mathbf{P}_n$ be a collection
of probability measures $\Pr$ for which conditions ASTE ($\Pr$), SM ($\Pr$), SE ($\Pr$), and R ($\Pr$) hold for the given $n$.
Consider any sequence $\{\Pr_n\}$, with index $n \in \{n_0, n_0+1,...\}$, with $\Pr_n \in \mathbf{P}_n$ for each $n\in \{n_0, n_0+1,...\}$.  By Theorem 1 we have that, for $c = \Phi^{-1}(1- \gamma/2)$, $
\lim_{n \to \infty} \Pr_n \left ( \alpha_0 \in [ \check \alpha  \pm  c \hat \sigma_n /\sqrt{n}]\right) =  \Phi(c ) - \Phi(-c) = 1- \gamma.
$
This means that for every further subsequence  $\{ \Pr_{n_k} \}$
with $\Pr_{n_k} \in \mathbf{P}_{n_k} $ for each $k \in \{1,2,...\}$
\begin{equation}\label{eq:Corollary}
\lim_{k \to \infty} \Pr_{n_k} \left ( \alpha_0 \in [ \check \alpha  \pm  c \hat \sigma_{n_k} /\sqrt{n_k}]\right) = 1- \gamma.
\end{equation}

Suppose that the claim of corollary does not hold, i.e.
$$
\limsup_{n \to \infty} \sup_{\Pr \in \mathbf{P}_n} \Big |\Pr \left ( \alpha_0 \in [ \check \alpha \pm  c\hat \sigma_n /\sqrt{n}]\right) - (1- \gamma) \Big | > 0.
$$
Hence there is a subsequence $\{\Pr_{n_k}\}$ with $\Pr_{n_k} \in \mathbf{P}_{n_k} $ for each $k \in \{1,2,...\} $ such that:
$$
\lim_{k \to \infty} \Pr_{n_k} \left ( \alpha_0 \in [ \check \alpha  \pm  c \hat \sigma_{n_k} /\sqrt{n_k}]\right) \neq 1- \gamma.
$$
This gives a contradiction to (\ref{eq:Corollary}).  The claim (i) follows. Claim (ii) follows from claim (i), since $\mathbf{P} \subseteq \mathbf{P}_n$ for all $n \geq n_0$.
\qed

\section{Proof of Theorem \ref{theorem:inferenceHLMS}}

We use the same notation as in Theorem \ref{theorem:inference}. Using that notation the approximations bounds stated in Condition HLMS are equivalent to $ \|\MX g\| \leq \delta_n n^{1/4}$ and $\|\MX m\| \leq \delta_n n^{1/4}$.

Step 1. It follows the same reasoning as Step 1 in the proof of Theorem \ref{theorem:inference}.

Step 2. (Behavior of $i$.) Decompose, using $D = m+V$
\begin{eqnarray*}
i = V'\zeta/\sqrt{n} + \underset{=:i_{a}}{m' \MX g/\sqrt{n}} + \underset{=:i_{b}}{m'\MX \zeta/\sqrt{n}} + \underset{=:i_{c}}{V'\MX g/\sqrt{n}} - \underset{=:i_{d}}{ V'\PX \zeta/\sqrt{n}}.
\end{eqnarray*}
First, by Condition HLMS we have $\|\MX g\|=o_P(n^{1/4})$ and $\|\MX m\| = o_P(n^{1/4})$. Therefore
$$
|i_{a}| = |m'\MX g/\sqrt{n}| \leq \sqrt{n}\|\MX g/\sqrt{n}\|\|\MX m/\sqrt{n}\| \lesssim_P o(1).
$$
Second, using that $m = X\beta_{m0} + R_m$ and $m'\MX\zeta = R_m'\zeta-(\tilde \beta_m(\widehat I)-\beta_{m0})'X'\zeta$, we have
$$
\begin{array}{rl}|i_b| & \leq |R_m'\zeta/\sqrt{n}| + |(\tilde \beta_m(\hat I) - \beta_{m0})'X'\zeta/\sqrt{n}| \\
& \leq |R_m'\zeta/\sqrt{n}| +\|\tilde \beta_m(\hat I) - \beta_{m0}\|_1\|X'\zeta/\sqrt{n}\|_\infty \\
& \lesssim_P  \sqrt{s/n}+ \sqrt{s} \ \{o(n^{-1/4})+\sqrt{s/n}\}\sqrt{\log(p\vee n)} = o(1).
\end{array}
$$
This follows because $$|R_m'\zeta/\sqrt{n}|\lesssim_P\sqrt{R_m'R_m/n}\lesssim_P \sqrt{s/n}, $$
by Chebyshev inequality and Conditions SM and ASTE(iii),  $$ \|\tilde \beta_m(\hat I) - \beta_{m0}\|_1\leq \sqrt{\hat s + s}\|\tilde \beta_m(\hat I) - \beta_{m0}\| \lesssim_P \sqrt{s} \  \{o(n^{-1/4})+\sqrt{s/n}\},$$
by Step 4 and $\hat s = |\widehat I|\lesssim_P s$ by Condition HLMS, and $$\|X'\zeta/\sqrt{n}\|_{\infty}\lesssim_P \sqrt{ \log (p\vee n) }$$ holding by Step 4 in the proof of Theorem 1.

Third, using similar reasoning and the decomposition $g = X\beta_{g0} + R_g$  conclude
\begin{eqnarray*}
|i_c| & \leq &  |R_g'V/\sqrt{n}| + |(\tilde \beta_g(\hat I) - \beta_{g0})'X'V/\sqrt{n}| \\
& \lesssim_P & \sqrt{s/n} + \sqrt{s} \ \{o(n^{-1/4})+\sqrt{s/n}\} \ \sqrt{\log(p\vee n)} =o_P(1).
\end{eqnarray*}


Fourth, we have
$$
|i_d| \leq |\tilde \beta_V(\hat I)' X'\zeta/\sqrt{n}| \leq \| \tilde \beta_V(\hat I)\|_1 \|X'\zeta/\sqrt{n}\|_{\infty} \lesssim_P  \sqrt{ [s \log (p\vee n)]^2/n } = o_P(1),
$$
since $\|X'\zeta/\sqrt{n}\|_\infty \lesssim_P \sqrt{\log(p\vee n)}$ by Step 4 of the proof of Theorem 1, and
$$\begin{array}{rl}
 \| \tilde \beta_V(\hat I)\|_1 & \leq \sqrt{\hat s} \|\tilde \beta_V(\hat I)\| \leq
\sqrt{\hat s} \| (X[\hat I]'X[\hat I]/n)^{-1} X[\hat I]'V/n\|\\
&\leq \sqrt{\hat s} \phi^{-1}_{\min}(\hat s) \sqrt{\hat s}  \| X'V/\sqrt{n}\|_{\infty}/\sqrt{n} \lesssim_P s \sqrt{[\log (p\vee n)]/n}.\end{array}$$
The latter bound follows from $\hat s \lesssim_P s$ by condition HLMS so that $ \phi^{-1}_{\min}(\hat s) \lesssim_P 1$ by condition SE, and again invoking Step 4 of the proof of Theorem 1 to establish $\|X'V/\sqrt{n}\|_\infty \lesssim_P \sqrt{\log(p\vee n)}$.

Step 3. (Behavior of $ii$.) Decompose
$$
ii = (m+V)'\MX (m+V)/n = V'V/n + \underset{=:ii_a}{m' \MX m/n} +
\underset{=:ii_b}{ 2 m'\MX V/n } - \underset{=:ii_c}{V'\PX V/n}.
$$
Then $|ii_a| \lesssim_P o(n^{1/2})/n= o_P(n^{-1/2})$ by condition HLMS,
$|ii_b| = o(n^{-1/2})$ by reasoning similar to deriving
the bound for $|i_b|$, and $|ii_c| \lesssim_P [s \log (p\vee n)]/n= o_P(1)$
by reasoning similar to deriving the bound for $|i_d|$.

Step 4. (Auxiliary: Bounds on $\|\tilde\beta_m(\widehat I)-\beta_{m0}\|$ and $\|\tilde\beta_g(\widehat I)-\beta_{g0}\|$.)
To establish a bound on $\|\tilde\beta_g(\widehat I)-\beta_{g0}\|$ note that
$$ \|\MX g/\sqrt{n}\| \geq | \ \|X(\tilde \beta_g(\widehat I) - \beta_{g0})/\sqrt{n}\| - \|R_g/\sqrt{n}\| \ |$$ where $\|R_g/\sqrt{n}\|\lesssim_P\sqrt{s/n}$ holds by Chebyshev inequality and Condition ASTE(iii). Moreover, by Condition HLMS we have $ \|\MX g/\sqrt{n}\| = o_P(n^{-1/4})$ and $\hat s = |\widehat I|\lesssim_P s$. Thus
$$\begin{array}{rl}
o(n^{-1/4}) + \sqrt{s/n} & \gtrsim_P \|X(\tilde\beta_g(\widehat I)-\beta_{g0})/\sqrt{n} \|\\
& \geq \sqrt{\semin{s+\hat s}}\|\tilde\beta_g(\widehat I)-\beta_{g0} \|\\
& \gtrsim_P \|\tilde\beta_g(\widehat I)-\beta_{g0} \|\\
\end{array}
$$ since $\sqrt{\semin{s+\hat s}} \gtrsim_P 1$ by Condition SE.

The same logic yields $\|\tilde\beta_m(\widehat I)-\beta_{m0}\|\lesssim_P \sqrt{s/n} + o(n^{-1/4})$.

Step 5. (Variance Estimation.)  It follows similarly to Step 7 in the proof of Theorem \ref{theorem:inference} but using Condition HLMS instead of Lemma \ref{corollary3:postrate}.

\qed

\section{Proof of Corollary 2}

The proof is similar to the proof of Corollary 1.

\section{Verification of Conditions for the Examples}\label{Sec:verify}

\subsection{\bf Verification for Example 1.}
Let $\mathbf{P}$ be the collection of all regression models
$\Pr$ that obey the conditions set forth above for all $n$ for the given constants $(p, b, B, q_x, q)$. Below we provide explicit bounds for $\kappa'$, $\kappa''$, $c$, $C$, $\delta_n$ and $\Delta_n$ that appear in Conditions ASTE, SE and SM that depend only on $(p, b, B, q_x, q)$ and $n$ which in turn establish these conditions for any $\Pr \in \mathbf{P}$.

Condition ASTE(i) is assumed. Condition ASTE(ii) holds with $\|\alpha_0\| \leq C_1^{ASTE} = B$.
Condition ASTE(iii) holds with $s=p$ and $r_{gi}=r_{mi}=0$.

Condition ASTE(iv) holds with $\delta_{1n}^{ASTE} := p^2 \log^2(p\vee n) / n \to 0$ since $s=p$ is fixed. Finally, we verify ASTE(v). Because $\tilde v_i = v_i$, $\tilde \zeta_i = \zeta_i$ and the moment condition $\Ep[|v_i^q|]+\Ep[|\zeta_i^q|]\leq C_2^{ASTE} = 2B$  with $q>4$, the first two requirements follow. To show the last requirement, note that because $\Ep[\|x_i\|^{q_x}]\leq B$ we have \begin{equation}\label{EqMax}
\Pr\left(\max_{1\leq i\leq n} \|x_i\|_\infty > t_{1n} \right) \leq \Pr\left( \left[\sum_{i=1}^n \|x_i\|^{q_x}\right]^{1/q_x} > t_{1n} \right) \leq n \Ep[ \|x_i\|^{q_x}]/t_{1n}^{q_x} \leq nB/t_{1n}^{q_x}=: \Delta_{1n}^{ASTE}.\end{equation}  Let $t_{1n} = (n \log n )^{1/q_x}B^{1/q_x}$ so that $\Delta_{1n}^{ASTE} = 1/\log n$. Thus we have  with probability $1-\Delta_{1n}^{ASTE}$ $$\max_{1\leq i\leq n}\|x_i\|_\infty^2 s n^{-1/2 + 2/q} \leq (n \log n )^{2/q_x}B^{2/q_x} p n ^{-1/2+2/q}=:\delta_{2n}^{ASTE}.$$ It follows that $\delta_{2n}^{ASTE} \to 0$ by the assumption that $4/q_x+4/q<1$.

To verify Condition SE note that
$$\begin{array}{rl}
 \Pr( \|\En[x_ix_i']-\Ep[x_ix_i']\| > t_{2n} ) &\displaystyle  \leq \sum_{k=1}^p\sum_{j=1}^p \frac{\Ep[x_{ij}^2x_{ik}^2]}{nt_{2n}^2} \leq \sum_{k=1}^p\sum_{j=1}^p \frac{\Ep[x_{ij}^4]+\Ep[x_{ik}^4]}{2nt_{2n}^2} \\
 &\displaystyle  \leq  \frac{p\Ep[\|x_{i}\|^4]}{nt_{2n}^2} \leq \frac{p B^{4/q_x}}{nt_{2n}^2} =: \Delta_{1n}^{SE}.\end{array}$$
Setting $t_{2n} := b/2$ we have $\Delta_{1n}^{SE} = (2/b)^2B^{4/q_x}p/n\to 0$ since $p$ is fixed.
Then, with probability $1-\Delta_{1n}^{SE}$ we have $$\begin{array}{rl}
\lambda_{\min}(\En[x_ix_i']) & \geq \lambda_{\min}(\Ep[x_ix_i']) - \|\En[x_ix_i']-\Ep[x_ix_i']\| \geq b/2=:\kappa',\\
\lambda_{\max}(\En[x_ix_i']) & \leq \lambda_{\max}(\Ep[x_ix_i']) + \|\En[x_ix_i']-\Ep[x_ix_i']\| \leq  \Ep[\|x_i\|^2] + b/2 \leq  2B^{2/q_x}=:\kappa''.\end{array}$$

In the verification of Condition SM note that the second and third requirements in Condition SM(i) hold with $c_{1}^{SM} = b$ and $C_1^{SM} = B^{2/q}$. Condition SM(iii) holds with $\delta_{1n}^{SM} := \log^3 p / n \to 0$ since $p$ is fixed.

The first requirement in Condition SM(i) and Condition SM(ii) hold by the stated moment assumptions, for $\epsilon_i = v_i$ and $\epsilon_i=\zeta_i$, $\tilde y_i = d_i$ and $\tilde y_i=y_i$,
$$\begin{array}{rl}
\displaystyle \Ep[|\epsilon_i^q|]& \leq B =: A_1\\
\displaystyle \Ep[|d_i^q|] & \leq 2^{q-1}\Ep[|x_i'\beta_{m0}|^q] +2^{q-1}\Ep[|v_i^q|]\leq 2^{q-1}\Ep[\|x_i\|^q]\|\beta_{m0}\|^q + 2^{q-1}\Ep[|v_i^q|] \\
& \leq 2^{q-1}(B^{q/q_x} B^q + B) =: A_2 \\
\Ep[d_i^4]& \leq 2^{3}(B^{4/q_x} B^4 + B) =: A_2'\\
\Ep[y_i^4]& \leq 3^3\|\alpha_0\|^4 \Ep[d_i^4] + 3^3\|\beta_{g0}\|^4\Ep[\|x_i\|^4] + 3^3\Ep[\zeta_i^4] \\ & \leq 3^3B^42^{3}A_2' + 3^3B^4B^{4/q_x}+3^3B^{4/q} =: A_3\\
\displaystyle \max_{1\leq j\leq p}\Ep[x_{ij}^2\tilde y_i^2] &\displaystyle \leq \max_{1\leq j\leq p} (\Ep[x_{ij}^4])^{1/2}(\Ep[\tilde y_i^4])^{1/2}\leq B^{2/q_x}(\Ep[\tilde y_i^4])^{1/2} \leq B^{2/q_x} (A_2'\vee A_3)^{1/2} =: A_4\\
\displaystyle \max_{1\leq j\leq p}\Ep[|x_{ij}\epsilon_i|^3] &\displaystyle = \max_{1\leq j\leq p} \Ep[|x_{ij}^3|\Ep[|\epsilon_i^3| \mid x_i]] \leq B^{3/q} \max_{1\leq j\leq p}\Ep[|x_{ij}^3|] \leq B^{3/q + 3/q_x} =: A_5\\
\displaystyle  \max_{1\leq j\leq p} 1/\Ep[x_{ij}^2] &\leq 1/ \lambda_{\min}(\Ep[x_ix_i'])\leq 1/b =: A_6\\
\end{array}$$  since $4< q\leq q_x$. Thus these conditions hold with $C_2^{SM} = A_2 \vee (A_1 + (A_2'\vee A_3)^{1/2} + A_4 + A_5  + A_6 )$.

Next we show Condition SM(iv). By (\ref{EqMax}) we have $\max_{1\leq i\leq n}\|x_i\|^2_\infty \leq (n \log n)^{2/q_x}B^{2/q_x}$ with probability $1-\Delta_{1n}^{ASTE}$, thus with the same probability $$\max_{i\leq n} \|x_i\|_\infty^2 \frac{s \log (n\vee p)}{n} \leq (B\log n)^{2/q_x} \frac{n^{2/q_x} p \log(p\vee n)}{n} =: \delta^{SM}_{1n}\to 0$$ since $q_x>4$ and $s=p$ is fixed.

Next for $\epsilon_i = v_i$ and $\epsilon_i = \zeta_i$ we have
$$ \Pr\left( \max_{1\leq j\leq p} | (\En-\Ep)[x_{ij}^2\epsilon_i^2]| > \delta_{2n}^{SM} \right) \leq \sum_{j=1}^p \frac{\Ep[x_{ij}^4\epsilon_i^4]}{n(\delta_{2n}^{SM})^2} \leq \frac{p B^{4/q+4/q_x}}{n(\delta_{2n}^{SM})^2} =: \Delta_{1n}^{SM}$$
by the union bound, Chebyshev inequality and by $\Ep[x_{ij}^4 \epsilon_i^4] = \Ep[x_{ij}^4\Ep[\epsilon_i^4 \mid x_i]] \leq B^{4/q+4/q_x}$. Letting $\delta_{2n}^{SM} = B^{2/q+2/q_x}n^{-1/4}\to 0$ we have $\Delta_{1n}^{SM} = p/n^{1/2} \to 0$ since $p$, $B$, $q$ and $q_x$ are fixed.

Next for $\tilde y_i = d_i$ and $\tilde y_i = y_i$ we have
$$ \Pr\left( \max_{1\leq j\leq p} | (\En-\Ep)[x_{ij}^2\tilde y_i^2]| > \delta_{3n}^{SM} \right) \leq \sum_{j=1}^p \frac{\Ep[x_{ij}^4\tilde y_i^4]}{n(\delta_{3n}^{SM})^2} \leq \frac{p B^{4/q_x} A_8^{4/q}}{n(\delta_{3n}^{SM})^2} =: \Delta_{2n}^{SM}$$
by the union bound, Chebyshev inequality and by $$\Ep[x_{ij}^4 \tilde y_i^4] \leq \Ep[x_{ij}^{\tilde q}]^{4/\tilde q} \Ep[\tilde y_i^q]^{4/q} \leq \Ep[x_{ij}^{q_x}]^{4/q_x} \Ep[\tilde y_i^q]^{4/q} \leq B^{4/q_x} A_8^{4/q}$$ holding by H\"older inequality where $4<\tilde q \leq q_x$ such that $4/q + 4/\tilde q = 1$, and
$$\begin{array}{rl}
 \Ep[\tilde y_i^q] & \leq  (1 + 3^{q-1}\|\alpha_0\|^q) \Ep[d_i^q] + 3^{q-1}\|\beta_{g0}\|^q\Ep[\|x_i\|^q] + 3^{q-1}\Ep[\zeta_i^q] \\
  &\leq 3^q( A_2+B^qA_2+ B^qB^{q/q_x}+B)=:A_8.\end{array}$$
Letting $\delta_{3n}^{SM} = B^{4/q_x} A_8^{4/q}n^{-1/4}\to 0$ we have $\Delta_{2n}^{SM} = p/n^{1/2} \to 0$ since $p$, $B$, $q$ and $q_x$ are fixed.

Finally,  we set $ c = c_1^{SM}, \ \  C = \max\{C_1^{ASTE}, C_2^{ASTE}, C_1^{SM}, C_2^{SM} \}$, $\delta_n = \max\{ \delta_{1n}^{ASTE}, \delta_{2n}^{ASTE},$ $\delta_{1n}^{SM}+\delta_{2n}^{SM}+\delta_{3n}^{SM}\} \to 0,$ and $\Delta_n = \max\{ \Delta_{1n}^{ASTE} + \Delta_{1n}^{SM} + \Delta_{2n}^{SM}, \Delta_{1n}^{SE}\}\to 0.$ $\qed$

We will make use of the following technical lemma in the verification of examples 2, 3, and 4.

\begin{lemma}[Uniform Approximation]\label{LemmaUnifApp}
Let $h_i=x_i'\theta_h+\rho_i$ be a function whose coefficients $\theta_h\in S_A^\aa(p)$, and $\underline{\kappa} \leq \lambda_{\min}(\Ep[x_ix_i'])\leq \lambda_{\max}(\Ep[x_ix_i'])\leq \bar{\kappa}$. For $s=A^{1/\aa} n^{1/2\aa}$, $\aa>1$, define $\beta_{h0}$ as in (\ref{DefBetah}), $r_{hi} = h_i-x_i'\beta_{h0}$, for $i=1,\ldots,n$. Then we have
$$ |r_{hi}| \leq \|x_i\|_\infty(\bar{\kappa}/\underline{\kappa})^{3/2}\left\{\frac{2\aa-1}{\aa-1}\sqrt{s^2/n} + 5\sqrt{s  \Ep[\rho^2_i]/\underline{\kappa}}\right\} +|\rho_i|.$$
\end{lemma}
\begin{proof}
Let $T_h$ denote the support of $\beta_{h0}$ and $S$ denote the support of the $s$ largest components of $\theta_h$. Note that $|T_h| = |S| = s$. First we establish some auxiliary bounds on the $\|\theta_{h}[T_h^c]\|$ and $\|\theta_{h}[T_h^c]\|_1$. By the optimality of $T_h$ and $\beta_{h0}$ we have that
$$ \sqrt{\Ep[ (h_i - x_i'\beta_{h0})^2]} \leq \sqrt{\Ep[ (x_i[S^c]'\theta_h[S^c] + \rho_i)^2]} \leq \sqrt{\bar{\kappa}}\|\theta_h[S^c]\| + \sqrt{\Ep[\rho_i^2]} \ \ \mbox{and} $$
$$ \sqrt{\Ep[ (h_i - x_i'\beta_{h0})^2]} = \sqrt{\Ep[ \{x_i'(\theta_h - \beta_{h0}) + \rho_i\}^2]} \geq  \sqrt{\underline{\kappa}}\|\theta_h[T^c_h]\| - \sqrt{\Ep[\rho^2_i]}. $$
Thus we have $\|\theta_h[T^c_h]\| \leq \sqrt{\bar{\kappa}/\underline{\kappa}}\|\theta_h[S^c]\| + 2\sqrt{\Ep[\rho^2_i]/\underline{\kappa}}$. Moreover, since $\theta_h \in S_A^\aa(p)$, we have
$$ \|\theta_h[S^c]\|^2 = \sum_{j=s+1}^\infty \theta_{h(j)}^2 \leq A^2 \sum_{j=s+1}^\infty j^{-2\aa} \leq A^2 s^{-2\aa+1}/[2\aa-1] \leq A^2 s^{-2\aa+1}$$
since $a>1$. Combining these relations we have
$$\begin{array}{rl}
\|\theta_h[T^c_h]\|  &\displaystyle \leq \sqrt{\bar{\kappa}/\underline{\kappa}} A s^{-\aa + 1/2} + 2\sqrt{\Ep[\rho^2_i]/\underline{\kappa}} \\
&\displaystyle = \sqrt{\bar{\kappa}/\underline{\kappa}} \sqrt{s/n} + 2\sqrt{\Ep[\rho^2_i]/\underline{\kappa}}. \end{array}$$

The second bound follows by observing that $$\begin{array}{rl} \|\theta_{h}[T_h^c]\|_1 & \displaystyle  \leq \sqrt{s}\|\theta_{h}[T_h^c \cap S]\| + \|\theta_{h}[S^c]\|_1 \leq \sqrt{s}\|\theta_h[T^c_h]\| + A s^{-\aa+1}/[\aa-1] \\
& \leq \sqrt{s^2/n}\sqrt{\bar{\kappa}/\underline{\kappa}} + 2\sqrt{s\Ep[\rho^2_i]/\underline{\kappa}} + (s/\sqrt{n})/[\aa-1]\\
&\leq \sqrt{s^2/n}\sqrt{\bar{\kappa}/\underline{\kappa}} \ a/[a-1] + 2\sqrt{s\Ep[\rho^2_i]/\underline{\kappa}}.\end{array}$$

By the first-order optimality condition of the problem (\ref{DefBetah}) that defines $\beta_{h0}$, we have $$\Ep[x_i[T_h]x_i[T_h]'](\beta_{h0}[T_h]-\theta_{h}[T_h]) = \Ep[x_i[T_h]x_i[T^c_h]']\theta_{h}[T^c_h]+\Ep[x_i[T_h]\rho_i].$$ Thus, since $\|\Ep[x_i[T_h]\rho_i]\| = \sup_{\|\eta\|=1}\Ep[\eta'x_i[T_h]\rho_i] \leq \sup_{\|\eta\|=1}\sqrt{\Ep[(\eta'x_i[T_h])^2]}\sqrt{\Ep[\rho_i^2]}$
we have $$\begin{array}{rl}
\underline{\kappa}\|\beta_{h0} - \theta_{h}[T_h]\| & \leq \bar{\kappa}\|\theta_{h}[T^c_h]\| + \sqrt{\bar{\kappa}\Ep[\rho^2_i]} \\
& \leq \sqrt{s/n} \ (\bar{\kappa}^{3/2}/\sqrt{\underline{\kappa}})+\sqrt{\Ep[\rho^2_i]} \ \sqrt{\bar{\kappa}}(1 + 2\sqrt{\bar{\kappa}/\underline{\kappa}})\end{array}$$
where the last inequality follows from the definition of $s=A^{1/\aa}n^{1/2\aa}$. Therefore
$$\begin{array}{rl}
|r_{hi}| & = |h_i - x_i'\beta_{h0}| = |x_i'(\theta_h - \beta_{h0})|+|\rho_i|\\
& \leq \|x_i\|_\infty\|\theta_h - \beta_{h0}\|_1+|\rho_i|\\
& \leq \sqrt{s}\|x_i\|_\infty\| \theta_{hT_h}-\beta_{h0}\| + \|x_i\|_\infty\|\theta_{hT_h^c}\|_1 +|\rho_i| \\
& \leq \|x_i\|_\infty \{\sqrt{s^2/n} \ (\bar{\kappa}/\underline{\kappa})^{3/2}+\sqrt{s\Ep[\rho^2_i]/\underline{\kappa}} \ \sqrt{\bar{\kappa}/\underline{\kappa}}(1 + 2\sqrt{\bar{\kappa}/\underline{\kappa}}) \}+\\  & + \|x_i\|_\infty(\sqrt{s^2/n}\sqrt{\bar{\kappa}/\underline{\kappa}} \ \aa/[\aa-1] + 2\sqrt{s\Ep[\rho^2_i]/\underline{\kappa}}) + |\rho_i| \\
& \leq \|x_i\|_\infty(\bar{\kappa}/\underline{\kappa})^{3/2}\{\frac{2\aa-1}{\aa-1}\sqrt{s^2/n} + 5\sqrt{s  \Ep[\rho^2_i]/\underline{\kappa}}\} +|\rho_i|.
\end{array}$$\end{proof}

\subsection{\bf Verification for Example 2.}
Let $\mathbf{P}$ be the collection of all regression models
$\Pr$ that obey the conditions set forth above for all $n$ for the given constants $(\underline{\kappa}, \bar{\kappa}, a, A, B,\chi)$ and sequences $p_n$ and $\bar\delta_n$. Below we provide explicit bounds for $\kappa'$, $\kappa''$, $ c$, $ C$, $ \delta_n$ and $ \Delta_n$ that appear in Conditions ASTE, SE and SM that depend only on $(\underline{\kappa}, \bar{\kappa}, a, A, B,\chi)$, $p$, $\bar\delta_n$ and $n$ which in turn establish these conditions for any $\Pr \in \mathbf{P}$. In what follows we exploit Gaussianity of $w_i$ and use that $(\Ep[|\eta'w_i|^k])^{1/k} \leq G_k(\Ep[|\eta'w_i|^2])^{1/2}$ for any vector $\eta$, $\|\eta\|<\infty$, where the constant $G_k$ depends on $k$ only.

Conditions ASTE(i) is assumed. Condition ASTE(ii) holds with $\|\alpha_0\| \leq B =: C_1^{ASTE}$.
Because $\theta_m, \theta_g \in S^\aa_A(p)$, Condition ASTE(iii) holds
with $$s=A^{1/\aa} n^{1/2\aa}, \ \ \ r_{mi} = m(z_i) - \sum_{j=1}^p z_{ij}\beta_{m0j}, \ \ \mbox{and} \ \ r_{gi} = g(z_i) - \sum_{j=1}^p z_{ij}\beta_{g0j}$$ where $\|\beta_{m0}\|_0\leq s$ and $\|\beta_{g0}\|_0\leq s$. Indeed, we have $$\Ep[r_{mi}^2] \leq \Ep\left[\left(\sum_{j\geq s+1} \theta_{m(j)}z_{i(j)}\right)^2\right] \leq \bar{\kappa} \sum_{j\geq s+1} \theta_{m(j)}^2 \leq \bar{\kappa} A^2 s^{-2\aa+1}/[2\aa-1] \leq \bar{\kappa} s/n$$ where the first inequality follows by the definition of $\beta_{m0}$ in (\ref{DefBetah}), the  second inequality follows from $\theta_m \in S^\aa_A(p)$, and the last inequality because $s=A^{1/\aa} n^{1/2\aa}$. Similarly we have $\Ep[r_{gi}^2] \leq \Ep[(\sum_{j\geq s+1} \theta_{g(j)}z_{i(j)})^2] \leq  \bar{\kappa} A^2 s^{-2\aa+1}/[2\aa-1] \leq \bar{\kappa} s/n$. Thus let $C_2^{ASTE} :=\sqrt{\bar f}$.

Condition ASTE(iv) holds with $\delta_{1n}^{ASTE} := A^{2/a}n^{1/a - 1}\log^2(p\vee n) \to 0$ since $s = A^{1/a}n^{1/2\aa}$, $A$ is fixed, and the assumed condition $n^{(1-a)/a}\log^2(p\vee n) \log^2 n \leq \bar\delta_n \to 0$.

The moment restrictions in Condition ASTE(v) are satisfied by the Gaussianity. Indeed, we have for $q=4/\chi$ (where $\chi < 1$ by assumption)
$$\begin{array}{rl}
\Ep[|\tilde \zeta_i|^q] & \leq 2^{q-1}\Ep[|\zeta_i^q|]+2^{q-1}\Ep[|r_{gi}^q|] \leq 2^{q-1}G_q^q (\Ep[\zeta_i^2]^{q/2} + \Ep[r_{gi}^2]^{q/2})\\
& \leq 2^{q-1}G_q^q \{\bar{\kappa}^{q/2} + \bar{\kappa}^{q/2} (s/n)^{q/2}\}\\
& \leq 2^qG_q^q\bar{\kappa}^{q/2}=:C_{3}^{ASTE}\end{array}$$ for $s\leq n$, i.e., $n\geq n_{01}^{ASTE} := A^{2/[2\aa-1]}$. Similarly, $\Ep[|\tilde v_i|^q] \leq C_{3}^{ASTE}$. Moreover, $$\begin{array}{rl}
|\Ep[\tilde \zeta_i^2 \tilde v_i^2] - \Ep[ \zeta_i^2v_i^2]|& \leq \Ep[\zeta_i^2r_{mi}^2] + \Ep[r_{gi}^2v_i^2] + \Ep[r_{mi}^2r_{gi}^2] \\
& \leq \sqrt{\Ep[\zeta_i^4]\Ep[r_{mi}^4]} + \sqrt{\Ep[r_{gi}^4]\Ep[v_i^4]} + \sqrt{\Ep[r_{mi}^4]\Ep[r_{gi}^4]} \\
& \leq G_4^2\bar\kappa \Ep[r_{mi}^2] + G_4^2\bar\kappa \Ep[r_{gi}^2] + G_4^2\Ep[r_{mi}^2] \Ep[r_{gi}^2] \\
& \leq G_4^2\bar{\kappa}^2 \{ 2 + \bar{\kappa}s/n\}s/n =: \delta_{2n}^{ASTE} \to 0.\end{array}$$

Next note that by Gaussian tail bounds and $\lambda_{\max}(\Ep[w_iw_i'])\leq \bar{\kappa}$ we have 
\begin{equation}\label{MaxGauss}\begin{array}{rl}
\max_{i\leq n}\|x_i\|_\infty & \leq \|\Ep[x_i]\|_\infty + \max_{i\leq n}\|x_i-\Ep[x_i]\|_\infty    \\
& \leq \sqrt{\bar\kappa} + \sqrt{2\bar{\kappa}\log(pn)} \ \ \mbox{with probability at least} \ 1-\Delta_{1n}^{ASTE}\end{array}\end{equation} where $\Delta_{1n}^{ASTE} = 1/\sqrt{2\bar{\kappa}\log(pn)}$. 
 The last requirement in Condition ASTE(v) holds with $q=4/\chi$ $$
 \max_{i\leq n}\|x_i\|_\infty^2 s n^{-1/2+2/q}   \leq  6\bar{\kappa}\log(pn) A^{1/\aa} n^{\frac{1}{2\aa} -\frac{1}{2}+\chi/2} =: \delta_{3n}^{ASTE}$$ with probability $1-\Delta_{1n}^{ASTE}$. By the assumption on $\aa$, $p$, $\chi$, and $n$, $\delta_{3n}^{ASTE} \to 0$.

To verify Condition SE with $\ell_n = \log n$  note that the minimal and maximal eigenvalues of $\Ep[x_ix_i']$ are bounded away from zero by $\underline{\kappa}>0$ and from above by $\bar \kappa<\infty$ uniformly in $n$. Also, let $\mu = \Ep[x_i]$ so that $x_i = \tilde x_i + \mu$ where $\tilde x_i$ is zero mean. By constriction $\Ep[x_ix_i'] = \Ep[\tilde x_i\tilde x_i'] + \mu \mu'$ and $\|\mu\|\leq \sqrt{\bar\kappa}$.

For any $\eta \in \RR^p$, $\|\eta\|_0\leq k:=s\log n$ and $\|\eta\|=1$, we have that
$$ \En[(\eta'x_i)^2] - \Ep[(\eta'x_i)^2] =  \En[(\eta'\tilde x_i)^2] - \Ep[(\eta'\tilde x_i)^2] + 2\eta'\En[\tilde x_i]\cdot \eta'\mu.$$
Moreover, by Gaussianity of $x_i$, with probability $1-\Delta_{1n}^{SE}$, where $\Delta_{1n}^{SE} = 1/\sqrt{2\bar\kappa \log(pn)}$,
$$\begin{array}{rl}
|\eta'\En[\tilde x_i]| & \leq \|\eta\|_1 \|\En[\tilde x_i]\|_\infty \leq \sqrt{k}\sqrt{2\bar\kappa \log(pn)}/\sqrt{n}\\
|\eta'\mu| & \leq \|\eta\| \ \|\mu\| \leq \sqrt{\bar\kappa}.
\end{array} $$

By the sub-Gaussianity of $\tilde x_i=(\Ep[x_ix_i']-\mu \mu')^{-1/2}\Psi_i$, where $\Psi_i \sim N(0,I_p)$,  by Theorem 3.2 in \citen{RudelsonZhou2011} (restated in Lemma \ref{thm:RZ32} in Appendix G) with $\tau = 1/6$, $k=s\log n$, $\alpha = \sqrt{8/3}$, provided that  $$n\geq N_n := 80(\alpha^4/\tau^2)(s\log n) \ \log(12ep/[\tau s\log n]),$$ we have
$$ (1-\tau)^2 \Ep[(\eta'\tilde x_i)^2] \leq \En[(\eta'\tilde x_i)^2] \leq (1+\tau)^2 \Ep[(\eta'\tilde x_i)^2]$$
with probability $1-\Delta_{1n}^{SE}$, where $\Delta_{1n}^{SE} = 2{\rm exp}(-\tau^2n/80\alpha^4)$. Note that under ASTE(iv) we have $\Delta_{1n}^{SE} \to 0$ and $$n_{01}^{SE}:=\max \{ n : n\leq N_n\}\leq \max\{(12e/\tau)^{2\aa}A^{-2}, \  80^2(\alpha^8/\tau^4)A^{2/\aa}, n^* \}$$ where $n^*$ is the smallest $n$ such that $\bar\delta_n < 1$.

Therefore, with probability $1-\Delta_{1n}^{SE}$ and $n\geq n_{01}^{SE}$, we have for any $\eta \in \RR^p$, $\|\eta\|_0\leq k$ and $\|\eta\|=1$,
$$
\begin{array}{rl}
 \En[(\eta' x_i)^2]  & \geq \Ep[(\eta'  x_i)^2] - | \En[(\eta' x_i)^2] -\Ep[(\eta'  x_i)^2]| \\
 & \geq \Ep[(\eta'  x_i)^2] - | \En[(\eta' \tilde x_i)^2] -\Ep[(\eta' \tilde x_i)^2]| - 2|\eta'\En[\tilde x_i]|\cdot |\eta'\mu|\\
 & \geq \Ep[(\eta'  x_i)^2] \{ 1 - 2\tau - \tau^2 \} - 2 \bar\kappa \sqrt{2k\log(pn)}/\sqrt{n} \\
 & \geq \Ep[(\eta'  x_i)^2]/2 - 2 \bar\kappa \sqrt{2k\log(pn)}/\sqrt{n} \\
\end{array}
$$ since $\tau = 1/6$ and $\Ep[(\eta' \tilde x_i)^2] \leq \Ep[(\eta' x_i)^2]$. So for $n \geq n_{02}^{SE} := 288k (\bar\kappa/\underline{\kappa})^2 \log(pn)$ we have
$$\semin{s\log n}[\En[x_ix_i']] \geq \underline{\kappa} / 3 =: \kappa'. $$

Similarly, we have
$$
\begin{array}{rl}
 \En[(\eta' x_i)^2]  & \leq \Ep[(\eta'  x_i)^2] + | \En[(\eta' x_i)^2] -\Ep[(\eta'  x_i)^2]| \\
 & \leq \Ep[(\eta'  x_i)^2] + | \En[(\eta' \tilde x_i)^2] -\Ep[(\eta' \tilde x_i)^2]| + 2|\eta'\En[\tilde x_i]|\cdot |\eta'\mu|\\
 & \leq \Ep[(\eta'  x_i)^2] \{ 1 + 2\tau + \tau^2 \} + 2 \bar\kappa \sqrt{2k\log(pn)}/\sqrt{n} \\
 & \leq 2\Ep[(\eta'  x_i)^2] + 2 \bar\kappa \sqrt{2k\log(pn)}/\sqrt{n} \\
\end{array}
$$ since $\tau = 1/6$ and $\Ep[(\eta' \tilde x_i)^2] \leq \Ep[(\eta' x_i)^2]$. So for $n \geq n_{03}^{SE} := 2k \log(pn)$ we have
$$\semax{s\log n}[\En[x_ix_i']] \leq 4\bar{\kappa}  =: \kappa''. $$



The second and third requirements in Conditions SM(i) holds by the Gaussianity of $w_i$, $\Ep[\zeta_i\mid x_i, v_i ] = 0$, $\Ep[v_i\mid x_i ] = 0$, and the assumption that the minimal and maximum eigenvalues of the covariance matrix (operator) $\Ep[w_iw_i']$ are bounded below and above by positive absolute constants.

The first requirement in Condition SM(i) and Condition SM(ii) also hold by Gaussianity. Indeed, we have for $\epsilon_i=v_i$ and $\epsilon_i=\zeta_i$, $\tilde y_i=d_i$ and $\tilde y_i=y_i$ $$
\begin{array}{rl}
\Ep[|v_i^q|]+ \Ep[|\zeta_i^q|] & \leq 2^{q-1}G_q^q\{ (\Ep[v_i^2])^{q/2} + (\Ep[\zeta_i^2])^{q/2}\} \leq 2^{q}G_q^q\bar\kappa^{q/2}=:A_1\\
\Ep[|d_i^q|] & \leq 2^{q-1}\Ep[|\theta_m'z|^q] + 2^{q-1}\Ep[|v_i^q|] \leq 2^{q-1}G_q^q(\Ep[|\theta_m'z|^2])^{q/2} +2^{q-1}G_q^q(\Ep[v_i^2])^{q/2}\\
& \leq 2^{q-1}G_q^q\|\theta_m\|^q\bar\kappa^{q/2} + 2^{q-1}G_q^q\bar\kappa^{q/2} \leq 2^qG_q^q\bar\kappa^{q/2} ( 1 + (2A)^q) =: A_2\\
\Ep[d_i^2] & \leq 2\Ep[|\theta_m'z_i|^2] + 2\Ep[v_i^2] \leq 2\bar\kappa \|\theta_m\|^2+2\bar\kappa \leq 2\bar\kappa(4A^2+1)=:A_2'\\
\Ep[y_i^2] & \leq 3|\alpha_0|^2\Ep[d_i^2] + 3\Ep[ |\theta_m'z|^2] + 3\Ep[\zeta_i^2]\leq 3B^2A_2'+ 3A_2'+3\bar\kappa=:A_3\\
\max_{1\leq j\leq p}\Ep[x_{ij}^2\tilde y_i^2] &\leq \max_{1\leq j\leq p} (\Ep[x_{ij}^4])^{1/2}(\Ep[\tilde y_i^4])^{1/2} \leq G_4^4\max_{1\leq j\leq p}\Ep[x_{ij}^2]\Ep[\tilde y_i^2]\\
& \leq G_4^4 \bar\kappa (A_2'\vee A_3)=:A_4\\
\max_{1\leq j\leq p}\Ep[|x_{ij}\epsilon_i|^3] &\leq \max_{1\leq j\leq p} (\Ep[x_{ij}^6])^{1/2}(\Ep[\epsilon_i^6])^{1/2} \leq G_6^6 \max_{1\leq j\leq p}(\Ep[x_{ij}^2])^{3/2}(\Ep[\epsilon_i^2])^{3/2}\\
& \leq G_6^6 \bar\kappa^3 =:A_5\\
\max_{1\leq j\leq p} 1/\Ep[x_{ij}^2] &\leq 1/ \lambda_{\min}(\Ep[w_iw_i'])\leq 1/\underline{\kappa}=:A_6
\end{array}$$ because $\|\theta_m\| \leq 2A$ and $\|\theta_g\|\leq 2A$ since $\theta_m,\theta_g \in S^\aa_A(p)$.
Thus the first requirement in Condition SM(i) holds with $C_2^{SM} = A_2$. Condition SM(ii) holds with $C_3^{SM} = A_1 + (A_2'\vee A_3) + A_4 + A_5+ A_6$.

Condition SM(iii) is assumed. 

To verify Condition SM(iv) note that for $\epsilon_i = v_i$ and $\epsilon_i=\zeta_i$, by (\ref{MaxGauss}), with probability $1-\Delta_{1n}^{ASTE}$,
\begin{equation}\label{FourthBound}\begin{array}{rl}
\max_{j\leq p} \sqrt{\En[ x_{ij}^4\epsilon_i^4 ]} & \leq
  \max_{j\leq p} \sqrt[4]{\En[ x_{ij}^8]}\sqrt[4]{\En[\epsilon_i^8 ]} \\
  & \leq \{\sqrt{\bar\kappa} + \sqrt{2\bar\kappa\log(p n)} \} \ \max_{j\leq p}\sqrt[4]{\En[ x_{ij}^4]}\sqrt[4]{\En[\epsilon_i^8]}.\end{array}\end{equation}
By Lemma \ref{Lemma:GaussConc} with $k=4$ we have with probability $1 - \Delta_{1n}^{SM}$, where $\Delta_{1n}^{SM} = 1/n$
\begin{equation}\label{Boundxij}\begin{array}{rl}
\max_{j\leq p} \sqrt[4]{\En[x_{ij}^4]} & \leq \|\Ep[x_i]\|_\infty + \max_{j\leq p} \sqrt[4]{\En[(x_{ij} -\Ep[x_{ij}])^4]} \\
& \leq \sqrt{\bar\kappa} + \sqrt{\bar\kappa} 2\bar C + \sqrt{\bar\kappa} n^{-1/4}\sqrt{2\log(2pn)} \leq 4\bar C\sqrt{\bar\kappa}\end{array}\end{equation} for $n \geq n_{01}^{SM} = 4\log^2(2pn)$.  Also, Lemma \ref{Lemma:GaussConc} with $k=8$ and $p=1$ we have
with probability $1-\Delta_{1n}^{SM}$ that
\begin{equation}\label{BoundEpsilon}\sqrt[4]{\En[\epsilon_i^8]} \leq 2\bar\kappa 8\bar C^2 + 2\bar\kappa n^{-1/4}2\log(2n)\leq 20\bar C^2\bar\kappa\end{equation}
for $n \geq n_{02}^{SM} = 16\log^4(2n)$. Moreover, we have
$$ \max_{1\leq j\leq p} \sqrt{\Ep[x_{ij}^4\epsilon_i^4]} \leq \max_{1\leq j\leq p}\sqrt[4]{\Ep[x_{ij}^8]}\sqrt[4]{\Ep[\epsilon_{i}^8]}\leq G_8^4\bar\kappa^2.$$

Applying Lemma \ref{Lemma:ProcessSecond}, for $\tau = 2\Delta_{1n}^{ASTE}+\Delta_{1n}^{SM}$, with probability $1 - 8\tau$ we have
$$ \max_{j\leq p}|(\En-\barEp)[ x_{ij}^2\epsilon_i^2]| \leq 4\sqrt{\frac{2\log(2p/\tau)}{n}}\sqrt{Q(\max_{1\leq j\leq p}\En[x_{ij}^4\epsilon_i^4],1-\tau)} \vee \frac{2\sqrt{2}G_8^4\bar\kappa^2}{\sqrt{n}} $$
where by (\ref{FourthBound}), (\ref{Boundxij}) and (\ref{BoundEpsilon}) we have
$$\begin{array}{rl}
Q(\max_{1\leq j\leq p}\sqrt{\En[x_{ij}^4\epsilon_i^4]},1-\tau) & \leq \bar\kappa^2\sqrt{2\log(p n)} 80\bar C^3.\\
 \end{array}$$
So we let $\delta_{1n}^{SM} = 640\bar C^3\bar\kappa^2\sqrt{\frac{\log(2p/\tau)}{n}} \sqrt{\log(p n)} \vee 2\sqrt{2}\frac{G_8^4\bar\kappa^2}{\sqrt{n}} \to 0$ under the condition that $\log^2(p\vee n)/ n \leq \bar\delta_n$.

Similarly for $\tilde y_i=d_i$ and $\tilde y_i=y_i$, by Lemma \ref{Lemma:GaussConc}, we have with probability $1-\Delta_{1n}^{SM}$, for $n\geq n_{02}^{SM}$ we have
\begin{equation}\label{Boundtildey}\begin{array}{rl}
\sqrt[8]{\En[\tilde y_i^8]} & \leq |\Ep[\tilde y_i]| + \sqrt[8]{\En[(\tilde y_i-\Ep[\tilde y_i])^8]} \\
& \leq  [A_2'\vee A_3]^{1/2} + (20\bar C^2 \Ep[\tilde y_i^2])^{1/2} \leq 6\bar C[A_2'\vee A_3]^{1/2}.\end{array}\end{equation}
Moreover, $\sqrt[4]{\Ep[\tilde y_i^8]} \leq G_8^2\Ep[\tilde y_i^2] \leq G_8^2[A_2'\vee A_3]$. Therefore by Lemma \ref{Lemma:ProcessSecond}, for $\tau = 2\Delta_{1n}^{ASTE}+\Delta_{2n}^{SM}$, with probability $1 - 8\tau$ we have by the arguments  in (\ref{FourthBound}), (\ref{Boundxij}), and (\ref{Boundtildey})
{\small $$\max_{j\leq p}|(\En-\barEp)[ x_{ij}^2\tilde y_i^2]| \leq  4\sqrt{\frac{2\log(2p/\tau)}{n}} \sqrt{6\bar\kappa \log(p n)} 4\bar C \sqrt{\bar\kappa}(36\bar C^2[A_2'\vee A_3])\vee \frac{2\sqrt{2}G_8^4\bar\kappa[A_2'\vee A_3]}{\sqrt{n}}=:\delta_{2n}^{SM}$$}
where $\delta_{2n}^{SM} \to 0$ under the condition $\log^2(p\vee n)/ n \leq \bar\delta_n \to 0$.

We have that the last term in Condition SM(iv) satisfies with probability $1-\Delta_{1n}^{ASTE}$
$$ \max \|x_i\|_\infty^2 \frac{s\log(p\vee n)}{n} \leq 6\bar\kappa\log (pn) A^{1/\aa}n^{-1+1/2\aa} \log(p\vee n) =:\delta_{3n}^{SM}.$$
Under ASTE(iv) and $s=A^{1/\aa}n^{1/2\aa}$ we have $\delta_{3n}^{SM}\to 0$.

Finally, we set $n_0 = \max\{ n_{01}^{ASTE}, n_{01}^{SE}, n_{02}^{SE}, n_{03}^{SE}, n_{01}^{SM}, n_{02}^{SM}\}$, $C = \max\{C_1^{ASTE}, C_2^{ASTE},$ $ 2C_3^{ASTE}, C_1^{SM},$ $ C_2^{SM} \}$, $\delta_n = \max\{\bar \delta_n, \delta_{1n}^{ASTE}, \delta_{2n}^{ASTE}, \delta_{1n}^{SM}+\delta_{2n}^{SM}+\delta_{3n}^{SM}\} \to 0,$ and $ \Delta_n = \max\{ 33\Delta_{1n}^{ASTE} + 16\Delta_{1n}^{SM}, \Delta_{1n}^{SE}\}\to 0.$

$ \ \ \qed$

\begin{lemma}\label{Lemma:GaussConc}
Let $f_{ij}\sim N(0,\sigma^2_{j})$, $\sigma_{j}\leq \sigma$, independent across $i=1,\ldots,n$, where $j=1,\ldots,p$. Then, for some universal constant $\bar C \geq 1$, we have that for any $k\geq 2$ and $\gamma \in (0,1)$
$$P\left( \max_{1\leq j\leq p} \{\En[|f_{ij}^k|]\}^{1/k} \geq \sigma \bar C \sqrt{k}  + \sigma n^{-1/k}\sqrt{2\log (2p/\gamma)} \right) \leq \gamma. $$
\end{lemma}
\begin{proof}
Note that $P( \En[|f_{ij}^k|] > M ) = P( \|f_{\cdot j}\|_k^k > Mn ) = P(\|f_{\cdot j}\|_k > (Mn)^{1/k})$.

Since $| \|f\|_k - \|g\|_k | \leq \|f-g\|_k \leq \|f-g\|$, we have that $\|\cdot\|_k$ is 1-Lipschitz for $k\geq 2$. Moreover, $$\Ep[\|f_{\cdot j}\|_k] \leq (\Ep[\|f_{\cdot j}\|_k^k])^{1/k} = (\sum_{i=1}^n\Ep[|f_{ij}^k|])^{1/k}= n^{1/k}(\Ep[|f_{1j}^k|])^{1/k}$$
$$  = n^{1/k}\{\sigma^k_j 2^{k/2}\Gamma((k+1)/2)/\Gamma(1/2)\}^{1/k} \leq n^{1/k}\sigma \sqrt{k} \bar C.$$

By \citen{LedouxTalagrandBook}, page 21 equation (1.6), we have
$$P( \|f_{\cdot j}\|_k  > (Mn)^{1/k}  ) \leq 2\exp(-\{(Mn)^{1/k} -  \Ep[\|f_{\cdot j}\|_k] \}^2/2\sigma^2_j).$$
Setting $M := \{\sigma \sqrt{k} \bar C  + \sigma n^{-1/k}\sqrt{2\log (2p/\gamma)}  \}^k$, so that $(Mn)^{1/k} = n^{1/k}\sigma \sqrt{k} \bar C + \sigma\sqrt{2\log (2p/\gamma)}$ we have by the union bound and $\sigma \geq \sigma_j$
$$P( \max_{1\leq j\leq p} \En[|f_{ij}^k|] \geq M ) \leq p \max_{1\leq j\leq p} P(\En[|f_{ij}^k|] \geq M ) \leq  \gamma.$$
\end{proof}

\subsection{\bf Verification for Example 3.}
Let $\mathbf{P}$ be the collection of all regression models
$\Pr$ that obey the conditions set forth above for all $n$ for the given constants $(\underline{f}, \bar f, \aa, A, b, B, q)$ and the sequence $\bar\delta_n$. Below we provide explicit bounds for $\kappa'$, $\kappa''$, $ c$, $ C$, $\delta_n$ and $\Delta_n$ that appear in Conditions ASTE, SE and SM that depend only on $(\underline{f}, \bar f, \aa, A, b, B, q)$ and  $\bar\delta_n$ which in turn establish these conditions for all $\Pr \in \mathbf{P}$.

Conditions ASTE(i) is assumed. Condition ASTE(ii) holds with $\|\alpha_0\| \leq B =: C_1^{ASTE}$.
Because $\theta_m, \theta_g \in S^\aa_A(p)$, Condition ASTE(iii) holds
with $$s=A^{1/\aa} n^{\frac{1}{2\aa}}, \ \ r_{mi} = m(z_i) - \sum_{j =1}^p \beta_{m0j}P_{j}(z_i) \ \ \mbox{and} \ \ r_{gi} = g(z_i)-\sum_{j=1}^p \beta_{g0j}P_{j}(z_i)$$ where $\|\beta_{m0}\|_0\leq s$ and $\|\beta_{g0}\|_0\leq s$. Indeed, we have $$\Ep[r_{mi}^2] \leq \Ep\left[\left(\sum_{j\geq s+1} \theta_{m(j)}P_{(j)}(z_i)\right)^2\right] \leq \bar f \sum_{j\geq s+1} \theta_{m(j)}^2 \leq \bar f  A^2 s^{-2\aa+1}/[2\aa-1] = \bar f s/n$$ where the first inequality follows by the definition of $\beta_{m0}$ in (\ref{DefBetah}), the  second inequality follows from the upper bound on the density and orthogonality of the basis, the third inequality follows from $\theta_m \in S^\aa_A(p)$, and the last inequality because $s=A^{1/\aa} n^{1/2\aa}$. Similarly we have $\Ep[r_{gi}^2] \leq \Ep[(\sum_{j\geq s+1} \theta_{g(j)}z_{i(j)})^2] \leq  \bar{f} A^2 s^{-2\aa+1}/[2\aa-1] = \bar{f} s/n$. Let $C_2^{ASTE} = \sqrt{\bar f}$.

Condition ASTE(iv) holds with $\delta_{1n}^{ASTE} := A^{2/a}n^{1/a - 1}\log^2(p\vee n) \to 0$ since $s = A^{1/a}n^{1/2\aa}$, $A$ is fixed, and the assumed condition $n^{(1-a)/a}\log^2(p\vee n)\leq \bar\delta_n \to 0$.

Next we establish the moment restrictions in Condition ASTE(v).
Because $\underline{f} \leq \lambda_{\min}(\Ep[x_ix_i'])\leq \lambda_{\max}(\Ep[x_ix_i'])\leq \bar f$, by the assumption on the density and orthonormal basis, and $\max_{i\leq n}\|x_i\|_\infty \leq B$, by Lemma \ref{LemmaUnifApp} with $\rho_i=0$ we have $$ \max_{1\leq i\leq n} |r_{mi}|\vee |r_{gi}| \leq \max_{1\leq i\leq n} \|x_i\|_\infty  (\bar{f}/\underline{f})^{3/2} \frac{2\aa-1}{\aa-1}\sqrt{s^2/n} \leq B (\bar{f}/\underline{f})^{3/2}\frac{2\aa-1}{\aa-1} \sqrt{s^2/n} =: \delta_{2n}^{ASTE}$$where $\delta_{2n}^{ASTE} \to 0$ under  $s=A^{1/\aa}n^{1/2\aa}$ and $\aa> 1$.

Thus we have
$$\begin{array}{rl}
\Ep[|\tilde \zeta_i|^q] & \leq 2^{q-1}\Ep[|\zeta_i^q|]+2^{q-1}\Ep[|r_{gi}^q|] \leq 2^{q-1}B + 2^{q-1}(\delta_{2n}^{ASTE})^q \\
& \leq 2^{q-1}B +2^{q-1}(\delta_{2n_0}^{ASTE})^q=: C_3^{ASTE}.\\
\end{array}$$ Similarly, $\Ep[|\tilde v_i|^q] \leq C_3^{ASTE}$. Moreover, since $\delta_{2n}^{ASTE} \to 0$ we have $$\begin{array}{rl}
|\Ep[\tilde \zeta_i^2 \tilde v_i^2] - \Ep[ \zeta_i^2v_i^2]|& \leq \Ep[\zeta_i^2r_{mi}^2] + \Ep[r_{gi}^2v_i^2] + \Ep[r_{mi}^2r_{gi}^2] \\
& \leq \sqrt{\Ep[\zeta_i^4]\Ep[r_{mi}^4]} + \sqrt{\Ep[r_{gi}^4]\Ep[v_i^4]} + \sqrt{\Ep[r_{mi}^4]\Ep[r_{gi}^4]} \\
& \leq 2B^{2/q} (\delta_{2n}^{ASTE})^2 + (\delta_{2n}^{ASTE})^4 =: \delta_{3n}^{ASTE} \to 0.\end{array}$$
Finally, the last requirement holds because $(1-\aa)/\aa + 4/q<0$ implies $$\max_{i\leq n}\|x_i\|_\infty^2 sn^{-1/2+2/q} \leq B^2 A^{1/\aa}n^{1/2\aa-1/2+2/q} =:\delta_{4n}^{ASTE} \to 0,$$ since $s=A^{1/\aa} n^{1/2\aa}$ and $\max_{i\leq n}\|x_i\|_\infty\leq B$.

To show Condition SE with $\ell_n = \log n$ note that regressors are uniformly bounded, and minimal and maximal eigenvalues of $\Ep[x_ix_i']$ are bounded below by $\underline{f}$ and above by $\bar f$ uniformly in $n$. Thus Condition SE follows by Corollary 4 in the supplementary material in \citen{BC-PostLASSO} (restated in Lemma \ref{thm:RV34} in Appendix G) which is based on \citen{RudelsonVershynin2008}. Let $$\delta_{1n}^{SE} := 2\bar CB\sqrt{s \log n}\log(1+s\log n) \sqrt{\log(p\vee n)} \sqrt{\log n} / \sqrt{n}$$ and $\Delta_{1n}^{SE}  := (2/\underline{f})(\delta_{1n}^{SE})^2 + \delta_{1n}^{SE} (2\bar f/\underline{f}) $, where $\bar C$ is an universal constant. By this result and the Markov inequality, we have
with probability $1 - \Delta_{1n}^{SE}$
$$ \kappa':= \underline f /2 \leq \semin{s\log n}[\En[x_ix_i']] \leq \semax{s\log n}[\En[x_ix_i']] \leq 2 \bar f =:\kappa''. $$
We need to show that $\Delta_{1n}^{SE} \to 0$ which follows from $\delta_{1n}^{SE} \to 0$.
We have that $$\delta_{1n}^{SE} 
\leq \frac{2\bar CB(1 + A)^2\sqrt{n^{1/2\aa}}\log^2(n) \sqrt{\log(p\vee n)} }{\sqrt{n}} = 2\bar CB(1 + A)^2 \sqrt{\frac{n^{1/2\aa}\log^4n}{n^{2/3}}} \sqrt{\frac{\log(p\vee n)}{n^{1/3}}}. $$
By assumption we have $\log^3 p / n \leq \bar\delta_n \to 0$ and $\aa > 1$ we have $\delta_{1n}^{SE} \to 0$.

The second and third requirements in Condition SM(i) hold with $C_1^{SM} = B^{2/q}$ and $c_1^{SM} = b$ by assumption. Condition SM(iii) is assumed.

The first requirement in Condition SM(i) and Condition SM(ii) follow by, for $\epsilon_i = v_i$ and $\epsilon_i=\zeta_i$, $\tilde y_i = d_i$ and $\tilde y_i=y_i$ $$
\begin{array}{rl}
\Ep[|v_i^q| ]+ \Ep[|\zeta_i^q| ] & \leq 2B =: A_1\\
\Ep[|d_i^q|] & \leq 2^{q-1}\Ep[|\theta_m'x_i|^q] + 2^{q-1}\Ep[|v_i^q|]\leq 2^{q-1}\|\theta_m\|_1^q\Ep[\|x_i\|_\infty^q] +2^{q-1}B\\
&  \leq 2^{q-1}(2A)^qB^q + 2^{q-1}B=:A_2\\
\Ep[d_i^2] & \leq 2\bar f \|\theta_m\|^2 + 2\Ep[v_i^2] \leq 8\bar fA^2 + 2B^{2/q} =: A_2'\\
\Ep[y_i^2]  & \leq 3|\alpha_0|^2\Ep[d_i^2]+3\|\theta_g\|_1^2\Ep[\|x_i\|_\infty^2] + 3\Ep[\zeta_i^2]\\
&\leq 3B^2A_2'+12A^2B^2+3B^{2/q} =: A_3\\
\max_{1\leq j\leq p}\Ep[x_{ij}^2\tilde y_i^2] & \leq B^2\Ep[\tilde y_i^2] \leq B^2(A_2'\vee A_3)=:A_4\\
\max_{1\leq j\leq p}\Ep[|x_{ij}\epsilon_i|^3] &\leq B^3 \Ep[|\epsilon_i^3|] \leq B^3B^{3/q}=:A_5\\
\max_{1\leq j\leq p} 1/\Ep[x_{ij}^2] & \leq 1/ \lambda_{\min}(\Ep[x_ix_i'])\leq 1/\underline{f}=:A_6
\end{array}$$ where we used that $\max_{i\leq n}\|x_i\|_\infty \leq B$, the moment assumptions of the disturbances, $\|\theta_m\|\leq \|\theta_m\|_1\leq 2A$, $\|\theta_g\|_1\leq 2A$ since $\theta_m,\theta_g\in S_A^{\aa}(p)$ for $\aa > 1$. 
Thus the first requirement in Condition SM(i) holds with $C_2^{SM} = A_2$. Condition SM(ii) holds with $C_3^{SM} := A_1 + (A_2'\vee A_3) + A_4 + A_5 + A_6$.

To verify Condition SM(iv) note that for $\epsilon_i=v_i$ and $\epsilon_i = \zeta_i$ we have by Lemma \ref{Lemma:ProcessSecond} with probability $1-8\tau$, where $\tau = 1/\log n$,
$$\begin{array}{rl}
\displaystyle \max_{1\leq j\leq p} |(\En-\barEp)[ x_{ij}^2\epsilon_i^2]| & \leq 4\sqrt{\frac{2\log(2p/\tau)}{n}}Q({\displaystyle \max_{1\leq j\leq p}} \sqrt{\En[x_{ij}^4\epsilon_i^4]},1-\tau) \vee \frac{2{\displaystyle \max_{1\leq j\leq p}}\sqrt{2\Ep[x_{ij}^4\epsilon_i^4]}}{\sqrt{n}} \\
&\leq 4\sqrt{\frac{2\log(2p/\tau)}{n}} B^2Q(\sqrt{\En[\epsilon_i^4]},1-\tau)\vee \frac{2B^2\sqrt{2\Ep[\epsilon_i^4]}}{\sqrt{n}}\\
&\leq 4\sqrt{\frac{2\log(2p\log n)}{n}} B^2 B^{2/q}\log n =:\delta_{1n}^{SM}  \end{array} $$
where we used $\Ep[\epsilon_i^4] \leq B^{4/q}$ and the Markov inequality. By the definition of $\tau$ and the assumed rate $\log^3(p\vee n)/n \leq \bar\delta_n\to 0$, we have $\delta_{1n}^{SM}\to 0$.

Similarly, we have for $\tilde y_i = d_i$ and $\tilde y_i = y_i$, with probability $1-8\tau$
$$\begin{array}{rl}
{\displaystyle\max_{1\leq j\leq p}} |(\En-\barEp)[ x_{ij}^2\tilde y_i^2]| & \leq 4\sqrt{\frac{2\log(2p/\tau)}{n}}Q({\displaystyle\max_{1\leq j\leq p}} \sqrt{\En[x_{ij}^4\tilde y_i^4]},1-\tau)\vee \frac{2{\displaystyle\max_{1\leq j\leq p}}\sqrt{2\Ep[x_{ij}^4\tilde y_i^4]}}{\sqrt{n}} \\
&\leq 4\sqrt{\frac{2\log(2p/\tau)}{n}} B^2Q(\sqrt{\En[\tilde y_i^4]},1-\tau) \vee \frac{2B^2\sqrt{2\Ep[\tilde y_i^4]}}{\sqrt{n}}\\
&\leq 4\sqrt{\frac{2\log(2p\log n)}{n}} B^2 A_7\log n =: \delta_{2n}^{SM}  \end{array} $$
where we used the Markov inequality and $$\begin{array}{rl}
\Ep[\tilde y_i^4] & \leq \Ep[d_i^4]+ 3^3|\alpha_0|^4\Ep[d_i^4]+3^3\|\theta_g\|_1^4\Ep[\|x_i\|_\infty^4] + 3^3\Ep[\zeta_i^4] \\
& \leq A_2^{4/q} + 3^3B^4A_2^{4/q}+3^3(2A)^4B^4+3^3B^{4/q}=:A_7.\end{array}$$ By the definition of $\tau$
and the assumed rate $\log^3(p\vee n)/n \leq \bar\delta_n \to 0$, we have $\delta_{2n}^{SM}\to 0$.

The last term in the requirement of Condition SM(iv), because $\max_{i\leq n}\|x_i\|_\infty \leq B$ and Condition ASTE(iv) holds, is bounded by $\delta_{3n}^{SM} := B^2 A^{1/\aa}n^{1/2\aa}\log(p\vee n) / n \to 0$.

Finally, we set $c = c_1^{SM}$, $ C = \max\{C_1^{ASTE}, C_2^{ASTE}, 2C_3^{ASTE}, C_1^{SM}, C_2^{SM}, C_3^{SM} \}$,
$ \delta_n = \max\{\bar\delta_n, \delta_{1n}^{ASTE}, \delta_{2n}^{ASTE}, \delta_{3n}^{ASTE}, \delta_{4n}^{ASTE}, \delta_{1n}^{SM}+\delta_{2n}^{SM}+\delta_{3n}^{SM}\} \to 0,$ $ \Delta_n = \max\{ 16/\log n, \Delta_{1n}^{SE}\}\to 0.$ $\qed$

\section{Tools}

\subsection{Moderate Deviations  for a Maximum of Self-Normalized Averages}

We shall be using the following result, which is  based on  Theorem 7.4 in \cite{delapena}.

\begin{lemma}[Moderate Deviation Inequality for Maximum of a Vector]\label{Lemma:SNMD} Suppose that
$$ \mathcal{S}_{j} =  \frac{\sum_{i=1}^n U_{ij}}{\sqrt{ \sum_{i=1}^n U^2_{ij}}},$$
where $U_{ij}$ are independent variables across $i$ with mean zero.  We have that
$$
\Pr \left( \max_{1 \leq j\leq p }|\mathcal{S}_{j}|  >  \Phi^{-1}(1- \gamma/2p)  \right) \leq \gamma \(1 + \frac{A}{\ell^3_n}\),
$$
where $A$ is an absolute constant, provided for $\ell_n > 0$
$$
0 \leq \Phi^{-1}(1- \gamma/(2p))  \leq \frac{n^{1/6}}{\ell_n} \min_{1\leq j \leq p} M^2_{j}-1, \ \  \ M_{j} := \frac{\left( \frac{1}{n} \sum_{i=1}^n \Ep [U_{ij}^2]\right)^{1/2}}{\left(\frac{1}{n} \sum_{i=1}^n \Ep[|U_{ij}^3|] \right)^{1/3}}.
$$
\end{lemma}
The proof of this result, given in \citen{BellChenChernHans:nonGauss}, follows from a simple combination  of union bounds with the bounds in  Theorem 7.4 in \citen{delapena}.

\subsection{Inequalities based on Symmetrization}

Next we proceed to use symmetrization arguments to bound the empirical process. In what follows for a random variable $Z$ let $Q(Z,1-\tau)$ denote its $(1-\tau)$-quantile.

\begin{lemma}[Maximal inequality via symmetrization]\label{Thm:masterSym}
Let $Z_1,\ldots, Z_n$ be arbitrary independent stochastic processes and $\mathcal{F}$ a finite set of measurable  functions. For any $\tau \in (0,1/2)$, and $\delta\in (0,1)$ we have
that with probability at least $1-4\tau-4\delta$
{\small $$
\max_{f \in \mathcal{F}} | \mathbb{G}_n(f(Z_i))| \leq    \left\{4 \sqrt{2\log(2|\mathcal{F}|/\delta)} \  Q\left(\max_{f\in\mathcal{F}}\sqrt{\En[ f(Z_i)^2 ]},1-\tau\right)\right\} \vee 2\max_{f\in\mathcal{F}} Q\left(| \mathbb{G}_n(f(Z_i))|, \frac{1}{2}\right).
$$}
\end{lemma}
\begin{proof}
Let
$$e_{1n} = \sqrt{2\log(2|\mathcal{F}|/\delta)} \  Q\left(\max_{f\in\mathcal{F}}\sqrt{\En[ f(Z_i)^2 ]},1-\tau\right), \ \ \ e_{2n} = \max_{f\in\mathcal{F}} Q\left(| \mathbb{G}_n(f(Z_i))|, \frac{1}{2}\right)$$
and the event $\mathcal{E}= \{\max_{f\in
\mathcal{F}}\sqrt{\En\[f^2(Z_i)\]} \leq Q\left(\max_{f\in\mathcal{F}}\sqrt{\En[ f^2(Z_i) ]},1-\tau\right)\}$ which satisfies $P(\mathcal{E}) \geq 1-\tau$. By the symmetrization
Lemma 2.3.7 of \citen{vdV-W} (by definition of $e_{2n}$ we have $\beta_n(x)\geq 1/2$ in Lemma 2.3.7) we obtain
$$
\begin{array}{rl}
\mathbb{P}\left\{\max_{f \in \mathcal{F}} |\mathbb{G}_n(f(Z_i))|> 4 e_{1n} \vee 2e_{2n}  \right\} & \leq 4
\mathbb{P}\left\{\max_{f \in \mathcal{F}} | \mathbb{G}_n(\varepsilon_if(Z_i))| > e_{1n}  \right\}\\
 & \leq 4\mathbb{P}\left\{\max_{f \in \mathcal{F}} | \mathbb{G}_n(\varepsilon_i f(Z_i))| > e_{1n}  | \mathcal{E} \right\} + 4\tau
\end{array}$$ where $\varepsilon_i$ are independent Rademacher random variables, $P(\varepsilon_i=1) = P(\varepsilon_i=-1) = 1/2$.

Thus a union bound yields
\begin{equation}\label{Eq:afterUB} \mathbb{P}\left\{\max_{f \in \mathcal{F}}
|\mathbb{G}_n(f(Z_i))|> 4e_{1n}\vee 2e_{2n}\right\} \leq 4\tau +
4|\mathcal{F}| \max_{f \in \mathcal{F}}
\mathbb{P}\left\{ | \mathbb{G}_n(\varepsilon_if(Z_i))| > e_{1n}  | \mathcal{E}\right\} .
\end{equation}
 We then condition on the values of $Z_1,\ldots,Z_n$ and $\mathcal{E}$, denoting the conditional
 probability measure as $\mathbb{P}_{\varepsilon}$.
 Conditional on $Z_1,\ldots,Z_n$, by the Hoeffding inequality the symmetrized
 process $\mathbb{G}_n(\varepsilon_if(Z_i))$ is sub-Gaussian for the $L_2(\mathbb{P}_n)$ norm,
 namely, for $f \in \mathcal{F}$,
$ \mathbb{P}_{\varepsilon}\{|\mathbb{G}_n(\varepsilon_if(Z_i))| >x \} \leq 2 \exp(
- x^2/\{2\En[f^2(Z_i)]\}).$ Hence, under the event $\mathcal{E}$, we can bound
$$
\begin{array}{rcl}
\mathbb{P}_{\varepsilon}\left\{
|\mathbb{G}_n(\varepsilon_if(Z_i))| > e_{1n} | Z_1,\ldots,Z_n, \mathcal{E}\right\}
& \leq &  2\exp(-e_{1n}^2/[2\En[f^2(Z_i)])\\
& \leq & 2\exp(-\log (2|\mathcal{F}|/\delta)).\\
\end{array}
$$
Taking the expectation over $Z_1,\ldots,Z_n$ does not affect the
right hand side bound. Plugging in this bound  yields the result.
\end{proof}

The following specialization will be convenient.

\begin{lemma}\label{Lemma:ProcessSecond} Let $\tau \in (0,1)$ and $\{(x_i',\epsilon_i)' \in \RR^p\times \RR, i=1,\ldots,n\}$ be random vectors that are independent across $i$. Then with probability at least $1-8\tau$
$$\max_{1\leq j\leq p}|\En[x_{ij}^2\epsilon_i^2]-\barEp[x_{ij}^2\epsilon_i^2]| \leq 4\sqrt{\frac{2\log(2p/\tau)}{n} \  Q\left(\max_{1\leq j\leq p}\En[x_{ij}^4\epsilon_i^4], 1-\tau\right)} \vee 2\max_{1\leq j\leq p}\sqrt{\frac{2\barEp[x_{ij}^4\epsilon_i^4]}{n}}  $$
\end{lemma}
\begin{proof}
 Let $Z_i = x_i\epsilon_i$, $f_j(Z_i)  = x_{ij}^2\epsilon_i^2$, $\mathcal{F} = \{ f_1,\ldots,f_p\}$, so that $n^{-1/2}\Gn(f_j(Z_i)) = \En[x_{ij}^2\epsilon_i^2]-\barEp[x_{ij}^2\epsilon_i^2]$. Also, for $\tau_1 \in(0,1/2)$ and $\tau_2 \in (0,1)$, let
{\small $$  e_{1n} = \sqrt{2\log(2p/\tau_1)} \sqrt{Q\left(\max_{1\leq j\leq p}\En[x_{ij}^4\epsilon_i^4], 1-\tau_2\right)} \ \ \mbox{and} \ \ e_{2n} = \max_{1\leq j \leq p} Q(|\Gn(x_{ij}^2\epsilon_i^2)|,1/2) $$} where we have $e_{2n}\leq \max_{1\leq j\leq p} \sqrt{2\barEp[x_{ij}^4\epsilon_i^4]}$ by Chebyshev.

By Lemma \ref{Thm:masterSym} we have
$$P\(\max_{1\leq j\leq p}|\En[x_{ij}^2\epsilon_i^2]-\barEp[x_{ij}^2\epsilon_i^2]| > \frac{4e_{1n} \vee 2e_{2n}}{\sqrt{n}} \) \leq 4\tau_1+4\tau_2.$$
The result follows by setting $\tau_1=\tau_2=\tau<1/2$. Note that for $\tau\geq 1/2$ the result is trivial.
\end{proof}

\subsection{Moment Inequality}

We shall be using the following result, which is  based on  Markov inequality and \cite{vonbahr:esseen}.

\begin{lemma}[Vonbahr-Esseen's LLN]   Let $r \in [1,2]$, and independent zero-mean random variables $X_i$ with $\barEp[|X_i|^r]\leq C$.
Then for any $\ell_n > 0$
$$
Pr \left (\frac{ \left|\sum_{i=1}^n X_i\right|}{n} >  \ell_n n^{-(1-1/r)} \right )  \leq \frac{2C}{\ell_n^r}.
$$
\end{lemma}

\subsection{Matrices Deviation Bounds}

In this section we collect matrices deviation bounds. We begin with a bound due to \citen{Rudelson1999} for the case that $p<n$.

\begin{lemma}[Essentially in \citen{Rudelson1999}]
Let $x_i$, $i=1,\ldots,n$, be independent random vectors in $\RR^p$ and set
$$\delta_n:=\bar C \frac{\sqrt{\log (n\wedge p)}}{\sqrt{n}}\sqrt{\Ep[ \max_{1\leq i\leq n}\|x_i\|^2]}.$$ for some universal constant $\bar C$. Then, we have
$$ \Ep\left[ \sup_{\|\alpha\| =1} \left| \En\[ (\alpha'x_i)^2 - \Ep[(\alpha'x_i)^2] \]\right|\right] \leq \delta_n^2 + \delta_n\sup_{\|\alpha\| =1} \sqrt{\barEp[(\alpha'x_i)^2]}. $$
\end{lemma}

Based on results in \citen{RudelsonVershynin2008}, the following lemma for bounded regressors was derived in the supplementary material of \citen{BC-PostLASSO}

\begin{lemma}[Essentially in Theorem 3.6 of \citen{RudelsonVershynin2008}]\label{thm:RV34}
Let $x_i$, $i=1,\ldots,n$, be independent random vectors in $\RR^p$ be such that $\sqrt{\Ep[ \max_{1\leq i\leq n}\|x_i\|_\infty^2]} \leq K$. Let $\delta_n:= 2\left( \bar C K \sqrt{k} \log(1+k) \sqrt{\log (p\vee n)} \sqrt{\log n}  \right)/\sqrt{n}$, where $\bar C$ is the universal constant. Then,
$$ \Ep\left[ \sup_{\|\alpha\|_0\leq k, \|\alpha\| =1} \left| \En\[ (\alpha'x_i)^2 - \Ep[(\alpha'x_i)^2] \]\right|\right] \leq \delta_n^2 + \delta_n \sup_{\|\alpha\|_0\leq k, \|\alpha\| =1} \sqrt{\barEp[(\alpha'x_i)^2]}. $$
\end{lemma}
\begin{proof}
Let
$$ V_k = \sup_{\|\alpha\|_0\leq k, \|\alpha\|=1} \left| \En\[ (\alpha'x_i)^2 - \Ep[(\alpha'x_i)^2] \]\right|.$$

Then, by a standard symmetrization argument (\citen{GuedonRudelson2007}, page 804)
$$ \begin{array}{rl}
n\Ep[V_k] & \leq 2\Ep_x \Ep_\varepsilon\left[ \sup_{\|\alpha\|_0\leq k, \|\alpha\| =1}\left| \sum_{i=1}^n \varepsilon_i(\alpha'x_i)^2\right|\right].\\
\end{array}
$$
Letting $$\phi(k) = \sup_{\|\alpha\|_0\leq k, \|\alpha\|\leq 1}\En[(\alpha'x_i)^2] \ \ \mbox{and} \ \ \varphi(k)=\sup_{\|\alpha\|_0 \leq k, \|\alpha\|=1}\barEp[(\alpha'x_i)^2],$$ we have $\phi(k) \leq \varphi(k) + V_k$ and by  Lemma 3.8 in \citen{RudelsonVershynin2008} to bound the expectation in $\varepsilon$,
$$ \begin{array}{rl}
n\Ep[V_k] & \leq 2\left( \bar C \sqrt{k} \log(1+k) \sqrt{\log (p\vee n)} \sqrt{\log n}  \right) \sqrt{n}\Ep_x\[\max_{i\leq n}\|x_i\|_\infty\sqrt{\phi(k)}\]\\
& \leq 2\left( \bar C \sqrt{k} \log(1+k) \sqrt{\log (p\vee n)} \sqrt{\log n}  \right) \sqrt{n}\sqrt{\Ep_x\[\max_{i\leq n}\|x_i\|_\infty^2\]\Ep_x\[\phi(k)\]}\\
& \leq 2\left( \bar C K \sqrt{k} \log(1+k) \sqrt{\log (p\vee n)} \sqrt{\log n}  \right) \sqrt{n}\sqrt{\varphi(k) + \Ep[V_k]}.\\ \end{array}
$$ The result follows by noting that for positive numbers $v, A, B$,  $v \leq A(v+B)^{1/2}$ implies $v\leq A^2 + A \sqrt{B}$.
\end{proof}

The following result establishes an approximation bound for sub-Gaussian regressors and was developed in \citen{RudelsonZhou2011}. Recall that a random vector $Z\in \RR^p$ is isotropic if $\Ep[ZZ'] = I$, and it is called  $\psi_2$ with a constant $\alpha$ if for every $w\in \RR^p$ we have
$$\|Z'w\|_{\psi_2}:= \inf\{ t : \Ep[{\rm exp}( \ (Z'w)^2/t^2)]\leq 2\} \leq \alpha\|w\|_2.$$ 

\begin{lemma}[Essentially in Theorem 3.2 of \citen{RudelsonZhou2011}]\label{thm:RZ32}
Let $\Psi_i$, $i=1,\ldots,n$, be i.i.d. isotropic random vectors in $\RR^p$ that are $\psi_2$ with a constant $\alpha$. Let $x_i = \Sigma^{1/2}\Psi_i$ so that $\Sigma = \Ep[x_ix_i']$.  For $m\leq p$ and $\tau \in (0,1)$ assume that
$$n\geq \frac{80m\alpha^4}{\tau^2}\log\left(\frac{12ep}{m\tau}\right).$$
Then with probability at least $1-2\exp(-\tau^2n/80\alpha^4)$, for all $u\in \RR^p$, $\|u\|_0\leq m$, we have
$$ (1-\tau)\|\Sigma^{1/2}u\|_2 \leq \sqrt{\En[(x_i'u)^2]} \leq (1+\tau)\|\Sigma^{1/2}u\|_2.$$\end{lemma}

For example, Lemma \ref{thm:RZ32} covers the case of $x_i\sim N(0,\Sigma)$ by setting $\Psi_i\sim N(0,I)$ which is isotropic and $\psi_2$ with a constant $\alpha=\sqrt{8/3}$.


\bibliographystyle{econometrica}
\bibliography{mybibVOLUME}

\pagebreak

\begin{figure}
	 \includegraphics[width=\textwidth]{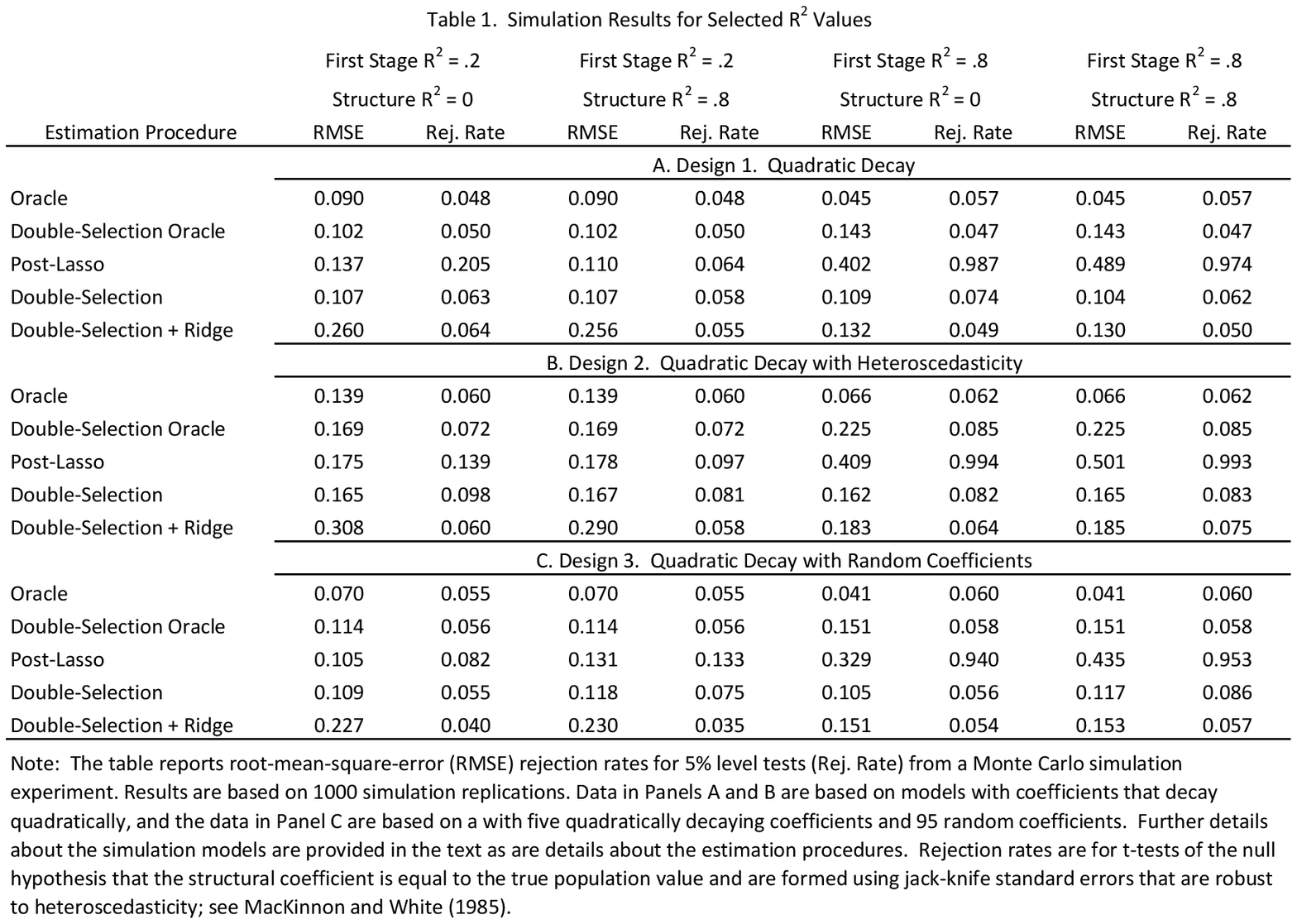}
	\label{fig:table1}
\end{figure}

\pagebreak

\begin{figure}
\includegraphics[width=\textwidth]{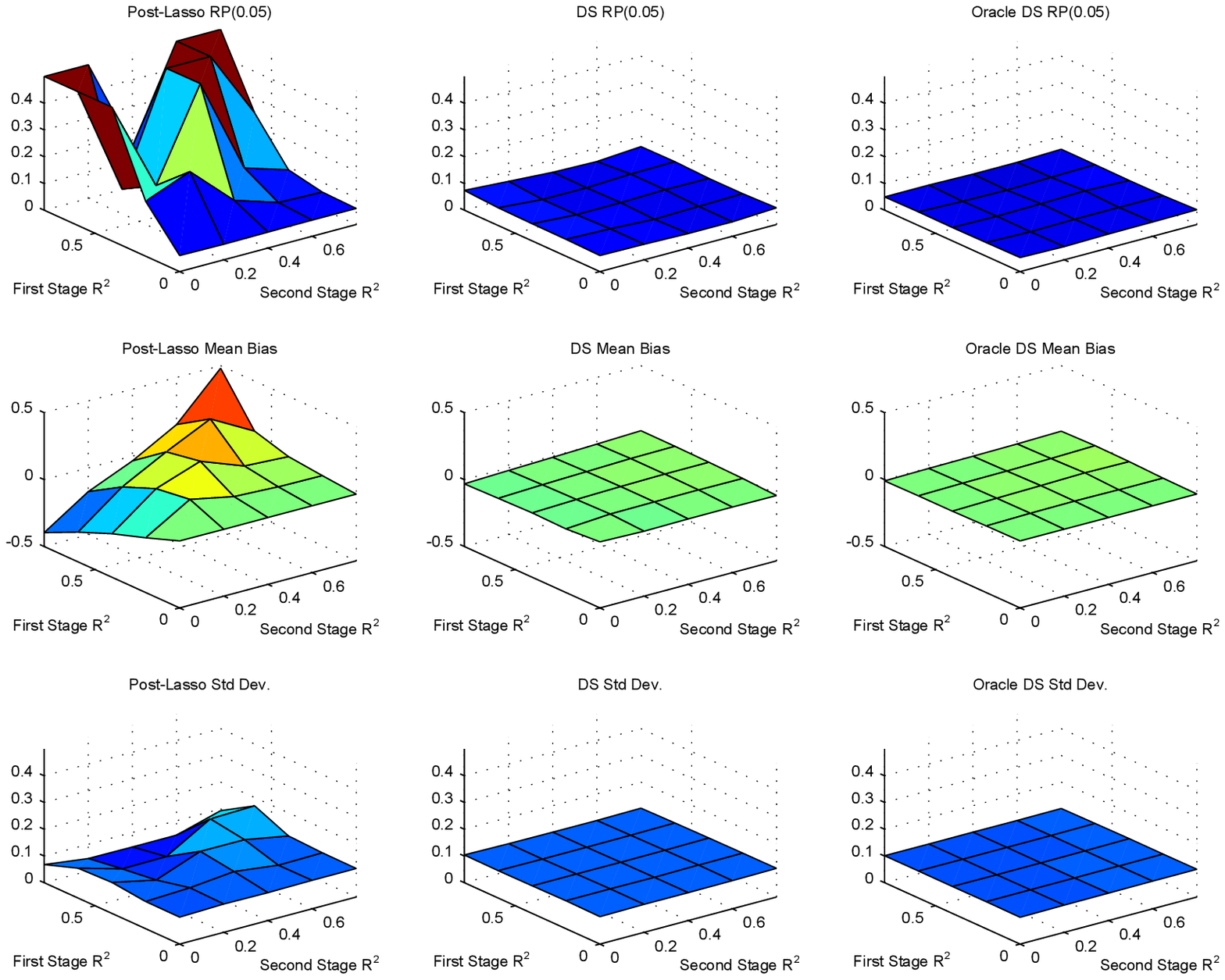}
	\label{fig:figure1}
\caption{This figure presents rejection frequencies for 5\% level tests, biases, and standard deviations for estimating the treatment effect from Design 1 of the simulation study which has quadratically decaying coefficients and homoscedasticity.  Results are reported for a one-step Post-Lasso estimator, our proposed double selection procedure, and the infeasible OLS estimator that uses the set of variables that have coefficients larger than 0.1 in either equation (\ref{eq: RPL1}) or (\ref{eq: RPL2}).  Reduced form and first stage $R^2$ correspond to the population $R^2$ of (\ref{eq: RPL1}) and (\ref{eq: RPL2}) respectively.  Note that rejection frequencies are censored at 0.5. }
\end{figure}

\pagebreak

\begin{figure}
\includegraphics[width=\textwidth]{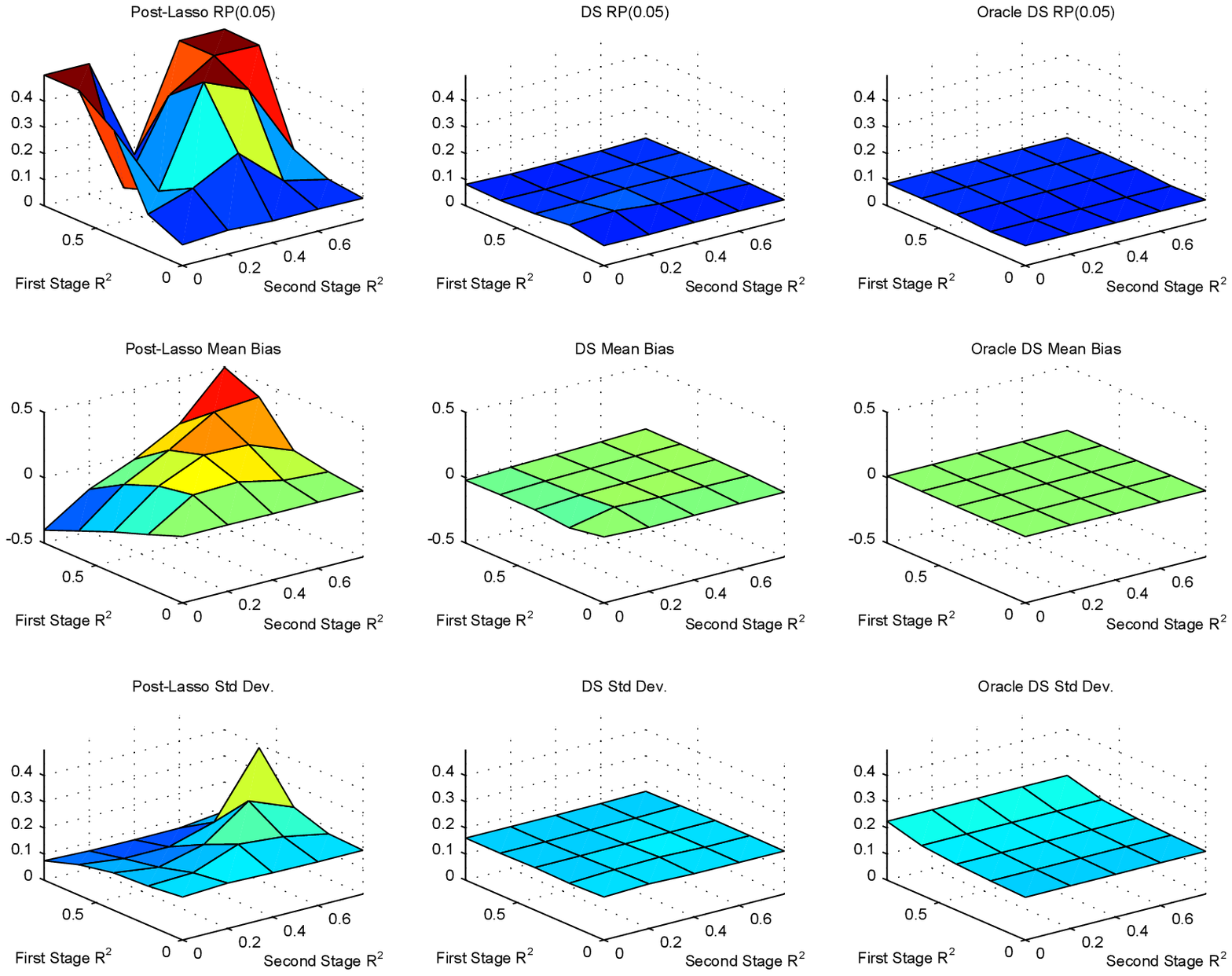}
	\label{fig:figure2}
\caption{This figure presents rejection frequencies for 5\% level tests, biases, and standard deviations for estimating the treatment effect from Design 2 of the simulation study which has quadratically decaying coefficients and heteroscedasticity.  Results are reported for a one-step Post-Lasso estimator, our proposed double selection procedure, and the infeasible OLS estimator that uses the set of variables that have coefficients larger than 0.1 in either equation (\ref{eq: RPL1}) or (\ref{eq: RPL2}).  Reduced form and first stage $R^2$ correspond to the population $R^2$ of (\ref{eq: RPL1}) and (\ref{eq: RPL2}) respectively.  Note that rejection frequencies are censored at 0.5. }
\end{figure}

\begin{figure}
\includegraphics[width=\textwidth]{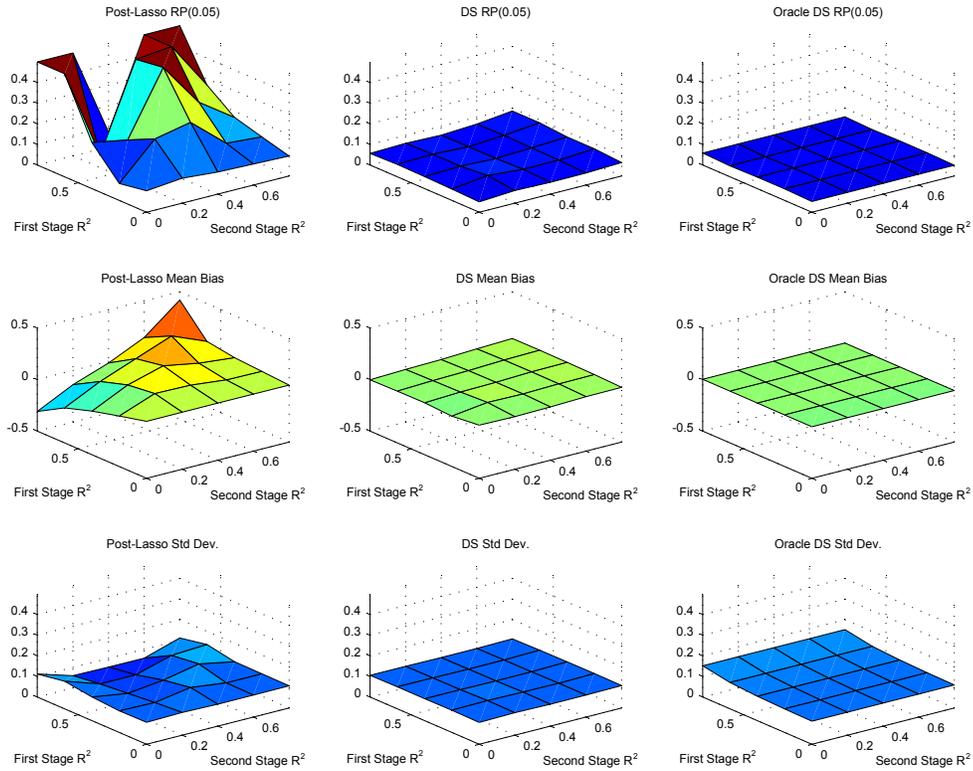}
	\label{fig:figure3}
\caption{This figure presents rejection frequencies for 5\% level tests, biases, and standard deviations for estimating the treatment effect from Design 3 of the simulation study which has five quadratically decaying coefficients and 95 Gaussian random coefficients.  Results are reported for a one-step Post-Lasso estimator, our proposed double selection procedure, and the infeasible OLS estimator that uses the set of variables that have coefficients larger than 0.1 in either equation (\ref{eq: RPL1}) or (\ref{eq: RPL2}).  Reduced form and first stage $R^2$ correspond to what would be the population $R^2$ of (\ref{eq: RPL1}) and (\ref{eq: RPL2}) if all of the random coefficients were equal to zero.  Note that rejection frequencies are censored at 0.5. }
\end{figure}

\begin{figure}
\includegraphics[width=\textwidth]{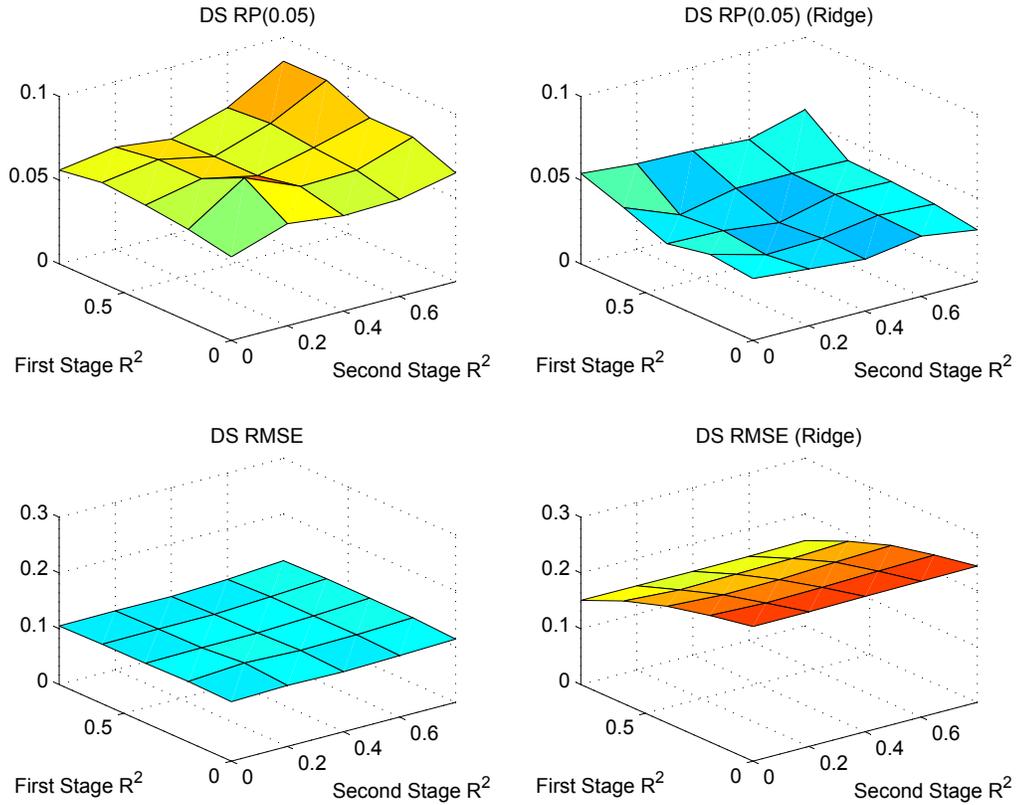}
	\label{fig:figure4}
\caption{This figure presents rejection frequencies for 5\% level tests and RMSE's for estimating the treatment effect from Design 3 of the simulation study which has five quadratically decaying coefficients and 95 Gaussian random coefficients.  Results in the first column are for the proposed double selection procedure, and the results in the second column are for the proposed double selection procedure when the ridge fit from (\ref{eq: RPL1}) is added as an additional potential control.  Reduced form and first stage $R^2$ correspond to what would be the population $R^2$ of (\ref{eq: RPL1}) and (\ref{eq: RPL2}) if all of the random coefficients were equal to zero.  Note that the vertical axis on the rejection frequency graph is from 0 to 0.1. }
\end{figure}

%
%

\pagebreak

\begin{figure}
	\includegraphics[width=\textwidth]{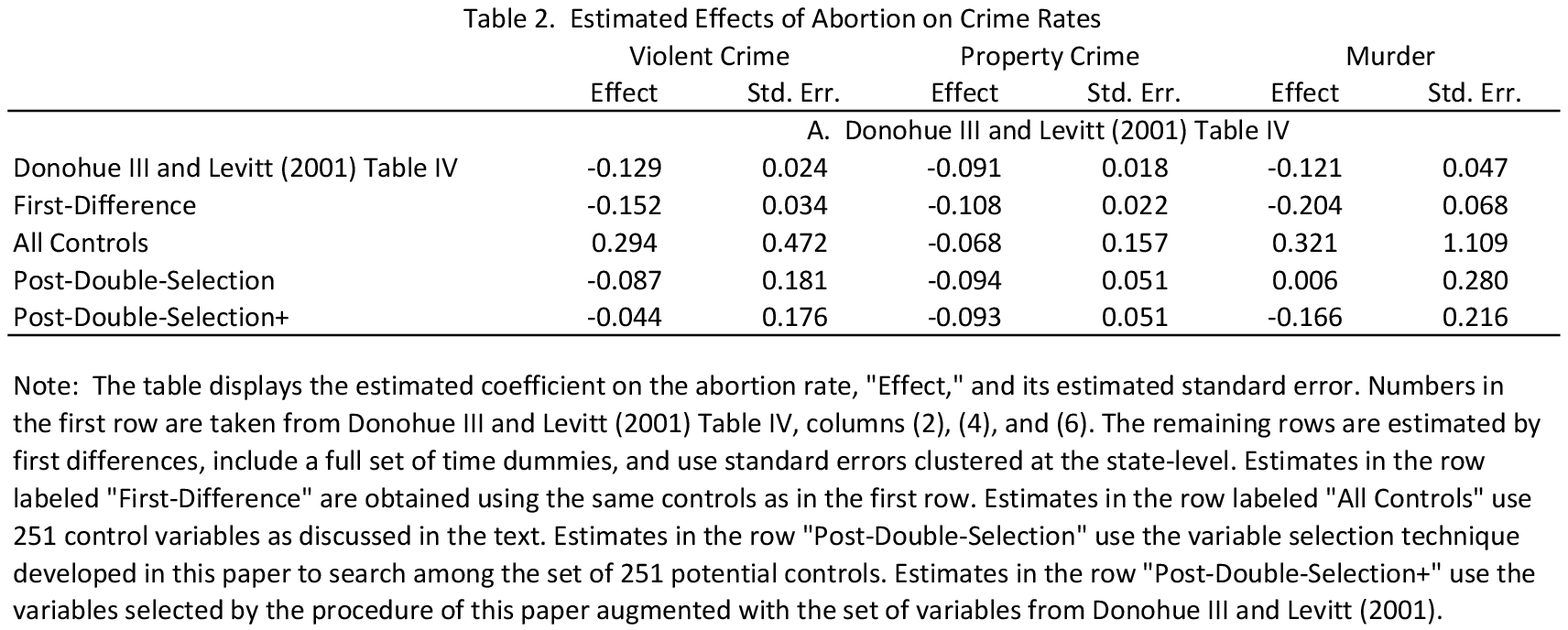}
	\label{fig:table2}
\end{figure}

\end{document}